%% file: main.tex
\documentclass[journal,twoside,print]{ieeecolor}

\input{settings/packages.tex}
\input{settings/math.tex}
\def\BibTeX{{\rm B\kern-.05em{\sc i\kern-.025em b}\kern-.08em
    T\kern-.1667em\lower.7ex\hbox{E}\kern-.125emX}}
\begin{document}

\title{A Generalized Plant Perspective on Linear-Convex Feedback Optimization}

\author{Fabian Jakob, Andrea Iannelli
\thanks{Authors are with the University of Stuttgart, Institute for Systems Theory and Automatic Control, Stuttgart, Germany (e-mail: \{fabian.jakob, andrea.iannelli\}@ist.uni-stuttgart.de). 
}
}

\maketitle

\begin{abstract}
Feedback optimization is a control approach for driving a dynamical system to the solution of an optimization problem by interconnecting the plant with an algorithm. Existing stability guarantees typically rely on timescale separation, enforced by conservative gain bounds that limit transient performance and require a pre-stabilized plant. This paper 
revisits the robust control perspective on feedback optimization.
We formulate the plant-optimizer interconnection as a generalized plant, where the cost gradients are characterized by Zames--Falb Integral Quadratic Constraints. Classical timescale-separation bounds are recovered as a special case of static multipliers, with dynamic multipliers yielding substantially tighter stability margins. The formulation also enables IQC based synthesis of dynamic output feedback controllers that jointly stabilize the plant and optimize transient performance, with possible model uncertainty absorbed into an uncertainty channel. For constrained problems, the framework extends to dynamic controllers that generalize projected gradient flows. Numerical examples illustrate the benefits and flexibility of the proposed approach.
\end{abstract}


\input{sections/intro.tex}

\input{sections/problem_setting.tex}

\input{sections/stability.tex}

\input{sections/perf_synth_uncert.tex}
\input{sections/constraints_.tex}

\input{sections/numerics.tex}

\input{sections/conclusion.tex}

\appendix
\input{sections/appendix}

\section*{References}
\bibliographystyle{ieeetr}
\bibliography{settings/references.bib}



\end{document}

%% file: settings/packages.tex
\usepackage{generic}
\usepackage{cite}
\usepackage{amsmath,amssymb,amsfonts}
\usepackage{algorithmic}
\usepackage{graphicx}
\usepackage{algorithm,algorithmic}
\usepackage{hyperref}
\hypersetup{hidelinks=true}
\usepackage{textcomp}

\usepackage{mathtools}
\usepackage{arydshln}

\usepackage{subcaption}
\usepackage{booktabs}

\usepackage{nicefrac}

\usepackage{pgfplots}
\usepackage{tikz}
\usetikzlibrary{arrows.meta, positioning, calc, shapes, decorations.pathreplacing, fit, backgrounds}
\usetikzlibrary{patterns, matrix, shapes.geometric}
\usetikzlibrary{tikzmark, overlay-beamer-styles}
\usetikzlibrary{fadings}
\pgfplotsset{compat=1.18}

\usepackage{xcolor}
 
\newtheorem{theorem}{Theorem}
\newtheorem{lemma}[theorem]{Lemma}
\newtheorem{proposition}[theorem]{Proposition}
\newtheorem{corollary}[theorem]{Corollary}

\newtheorem{assumption}{Assumption}

\newtheorem{remark}{Remark}
 

%% file: settings/math.tex
\newcommand{\Pixu}{\Pi_{xu}}
\newcommand{\Piyu}{\Pi_{yu}}
\newcommand{\Pixw}{\Pi_{xw}}
\newcommand{\Piyw}{\Pi_{yw}}
 
\newcommand{\Ahat}{\hat{A}}
\newcommand{\Bhat}{\hat{B}}
\newcommand{\Chat}{\hat{C}}

\newcommand{\R}{\mathbb{R}}

\newcommand{\Lone}{\mathcal{L}_1}
\newcommand{\Ltwo}{\mathcal{L}_2}

\newcommand{\Hinf}{\mathcal{H}_\infty}

\newcommand{\lmin}{\lambda_{\min}}
\newcommand{\lmax}{\lambda_{\max}}
\newcommand{\eps}{\varepsilon}
\newcommand{\proj}{\mathcal{P}}

\DeclarePairedDelimiter{\norm}{\lVert}{\rVert}

\newcommand{\T}{^\top}
 
\newcommand{\bmat}[1]{\begin{bmatrix}#1\end{bmatrix}}
\newcommand{\smat}[1]{\left[\begin{smallmatrix}#1\end{smallmatrix}\right]}

%% file: sections/intro.tex
\section{Introduction}\label{sec:introduction}

\IEEEPARstart{C}{ontrol} systems are increasingly expected to operate at setpoints that are economically or physically optimal. Computing such setpoints offline and tracking them with a feedback controller is breaks down when the plant model is uncertain or the operating conditions change. Feedback optimization~\cite{colombino2018,lawrence2021,hauswirth2024} offers an alternative by closing the loop directly around the plant and an optimization algorithm, replacing offline setpoint computation with continuous, measurement-based adaptation. The approach has found successful application in power systems~\cite{li_power}, traffic networks~\cite{bianchinPrimalDual}, process control~\cite{zagarowska}, and satellite navigation~\cite{chuySatellite}.

Despite its practical appeal, feedback optimization still presents open challenges in the analysis and design of the plant-optimizer interconnection. Stability guarantees are often conservative, transient performance is rarely addressed, and controller synthesis beyond gain tuning within a pre-fixed structure has received limited attention. Moreover, model uncertainty and constraints are typically handled ad-hoc rather than through a unified design methodology. These challenges are connected, and addressing them together has the potential to exploit synergies in the analysis and design process.

\begin{figure}[t]
    \centering
    \input{figures/IC.tex}
    \caption{Feedback optimization of LTI plants driven by structured dynamic algorithms.}
    \label{fig:algoStructure}
\end{figure}
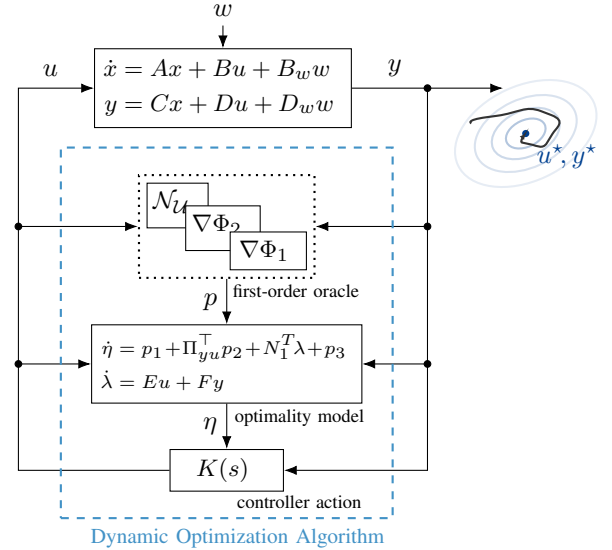

The dominant approach to stability analysis is timescale separation, where the controller gain is reduced until the plant appears quasi-static~\cite{menta2018,colombino2018,hauswirth2020,bianchinPrimalDual,simpsonporco2022,cothren}. This is broadly applicable, and despite efforts to relax this requirement~\cite{bianchi,bianchinIMPfeedback,yousefi}, it remains the dominant paradigm.
For linear time-invariant (LTI) plants with convex costs, however, the structure of the problem allows for sharper tools, and several works have explored the use of robust control methods for feedback optimization~\cite{colombino2018,nelsonMallada,lawrence2018,lawrence2021,simpsonporco2022}. These range from IQC-based stability analysis with static multipliers~\cite{colombino2018,nelsonMallada} to dynamic output-feedback synthesis via $H_\infty$ methods~\cite{lawrence2018}, the latter enabling joint stabilization and optimization even for unstable plants. 

While these works demonstrate the potential of a robust control approach, significant gaps remain. First, the controller synthesis approach to feedback optimization \cite{nelsonMallada,lawrence2018} has remained somewhat disconnected from the timescale separation literature. We believe part of the reason is that a clear analytical link between the two viewpoints is still missing. Second, the robust control toolbox has not yet been fully leveraged. Controller synthesis for performance was explored only on the timescale separated boundary system~\cite{simpsonporco2022}. Sparse efforts to address model mismatch exist~\cite{colombino2019,heGreybox,chan_robustOFO}, but outside a framework that would jointly accommodate performance and synthesis. Constrained problems were analyzed through specific algorithmic modifications~\cite{colombino2018,bianchinPrimalDual,heemelsCBF}, while a synthesis of dynamic controllers for such problems is missing. Third, the intrinsic properties of the optimization problem have not been fully exploited. Cost gradients are modelled as norm-bounded uncertainties in~\cite{lawrence2018} and as sector-bounded operators in~\cite{nelsonMallada,colombino2018,simpsonporco2022}, both known to be conservative descriptions. Dynamic Zames--Falb IQCs~\cite{zamesFalb,carrasco2014} provide a significantly sharper characterization of slope-restricted nonlinearities, and while this is well understood in the analysis of discrete-time optimization algorithms~\cite{lessard2016,schererEbenbauer,millerJakob}, it has not yet been leveraged in the feedback optimization context.

This paper addresses these gaps. We place the closed-loop of an LTI plant and a gradient-based feedback optimizer in the generalized plant framework, treating first-order oracles as slope-restricted operators characterized by Zames--Falb IQCs. The optimization algorithm arises in a structured way as a cascade of first-order oracle, optimality model, and dynamic controller; see Fig. \ref{fig:algoStructure}. The resulting formulation places stability analysis, performance optimization, controller synthesis, model uncertainty, and constrained optimization within a single, unified framework. Besides flexibility, we remove the important restriction of pre-stabilized plants, and we allow the designed controller to take care of both stabilization and optimization. Specifically, we make the following contributions.
\begin{enumerate}
\item We show that classical gradient-flow controllers can be re-interpreted as P-controllers acting on a generalized plant. Timescale separation bounds are recovered as a sufficient condition for a conservative IQC stability test based on static multipliers. 
\item By augmenting the generalized plant with suitable performance channels, we formulate a dynamic multiplier based performance test that allows both for principled tuning of classical algorithms and systematic synthesis of new algorithms. The latter produces dynamic output-feedback controllers that jointly stabilize and optimize the plant with guaranteed transient performance, while seamlessly accommodating model uncertainty.
\item In particular, we provide a methodology to systematically synthesize dynamic projected primal-dual gradient-flow controllers that can handle constrained optimization problems. Existing projected gradient flows are re-interpreted as \mbox{P-controllers}, and two novel architectures are derived that generalize them to dynamic controllers.
\end{enumerate}
\textbf{Organization.} The remainder of this work is organized as follows. Section~\ref{sec:setup} introduces the optimization problem, and the generalized-plant representation underlying all subsequent developments. Section~\ref{sec:stability} develops IQC-based stability conditions and connects them to classical bounds. Performance, synthesis, and model uncertainty are addressed in Section~\ref{sec:performance}, followed by constrained problems in Section~\ref{sec:constraints}. Numerical examples are presented in Section~\ref{sec:numerics}, and Section~\ref{sec:conclusion} concludes the paper.

\textbf{Notation.}
The $n\times n$ identity matrix and a zero matrix of compatible size are denoted as $I_n$ and $0$, respectively. For a symmetric matrix $M$, $M\succ 0$ ($M\succeq 0$) means positive (semi)definite, $\lmin(M)$, $\lmax(M)$ denote its smallest and largest eigenvalue, and $\bar{\sigma}(M)$ denotes its largest singular value. The Euclidean norm of a vector $v$ is $\norm{v}$, the vertical concatenation of two vectors (or signals) $v_1,v_2$ is $\mathrm{col}(v_1, v_2)$, and the block diagonal matrix of two matrices $M_1, M_2$ is $\mathrm{diag}(M_1,M_2)$.

We write $\mathcal{L}_2$ for the space of square integrable signals and $\| \cdot \|_{\mathcal{L}_2}$ for the respective $\mathcal{L}_2$-norm. For a real rational and stable transfer matrix $H$ with impulse response $h$, the $\Hinf$-norm is $\norm{H}_{\infty} = \sup_{\omega} \bar{\sigma}(H(j\omega))$ and the $\Lone$-norm is \mbox{$\norm{H}_{1} = \int_{-\infty}^{\infty} |h(t)|\mathrm{d}t$}.
If $G(s)\!=\!D + C(sI-A)^{-1}B$ is the transfer matrix of an LTI system, we write $y = G[\,u\,]$ equivalently for \mbox{$\dot{x} = A x + B u$}, \mbox{$y = C x + D u$}. We also introduce the notation 
\begin{equation*}
    G(s) \doteq \left[ \begin{array}{c|c} A & B \\ \hline C & D \end{array} \right]
\end{equation*}
if $(A, B, C, D)$ is a realization of $G$. For two LTI systems $G_1$ mapping some signals $\mathrm{col}(p,u) \mapsto \mathrm{col}(q,y)$ and $G_2$ mapping $y \mapsto u$, the star product \mbox{$G_1 \star G_2$} denotes the lower linear fractional transformation (LFT), mapping $p \mapsto q$~\cite{doyle}.

For a convex function $f$, $\nabla f$ denotes its gradient and $\partial f$ its subdifferential. For a closed convex set $\mathcal{U}$, the Euclidean projection is $\proj_{\mathcal{U}}(u) = \arg \min_{z\in\mathcal{U}} \tfrac{1}{2} \|z-u\|^2$ and the indicator function is
\begin{equation*}
    \mathcal{I}_{\mathcal{U}}(u) = \begin{cases}
        0 & ,\,u \in \mathcal{U}, \\
        \infty & ,\, u \notin \mathcal{U}.
    \end{cases}
\end{equation*}
The subdifferential of the indicator function is the normal cone \mbox{$\partial \mathcal{I}_{\mathcal{U}}(u) = \mathcal{N}_{\mathcal{U}}(u) \coloneq \{ s \in \mathbb{R}^{n_u} \; | \; s^\top ( u - y ) \leq 0, \; \forall y \in \mathcal{U} \}$}. For vectors $a,b$ and scalar $\alpha >0$, we write $a \in b + \alpha \mathcal{N}_{\mathcal{U}}(u)$ equivalently for \mbox{$\tfrac{1}{\alpha}(a - b) \in \mathcal{N}_{\mathcal{U}}(u)$}.

%% file: figures/IC.tex
\definecolor{istblue}{RGB}{0,65,145}
\definecolor{ozf1}{RGB}{66,146,198}
\begin{tikzpicture}[
    block/.style = {
    draw,
    minimum width=0.8cm,
    minimum height=0.5cm,
    align=center
    },
    ->,
    >=Latex,
    node distance=1.5em,
    split/.style={
        draw,
        fill,
        circle,
        minimum size=0.25em,
        inner sep=0pt
    },
    ]


\node[block] (plant)
{   \small
    $
    \begin{aligned}
        \dot{x} &= Ax + Bu + B_w w \\
        y &= Cx + Du + D_w w
    \end{aligned}
    $
};


\node[
    below=12mm of plant
] (GradCenter) {};

\node[
    block,
    anchor=center
] (GradA)
at ([xshift=-6mm,yshift=3mm]GradCenter.center)
{
    \raisebox{0.4em}{\hspace{-1pt}\small $\mathcal{N}_\mathcal{U}$}
};

\node[
    block,
    fill=white,
    anchor=center
] (GradB)
at ([xshift=0mm,yshift=0mm]GradCenter.center)
{
    \raisebox{0.4em}{\hspace{-1pt}\small $\nabla \Phi_2$}
};

\node[
    block,
    fill=white,
    anchor=center
] (GradC)
at ([xshift=6mm,yshift=-3mm]GradCenter.center)
{
    \small $\nabla \Phi_1$
};

\node[
    draw,
    dotted,
    thick,
    fit=(GradA)(GradB)(GradC),
    inner sep=1mm
] (GradBox) {};

\node[
    block,
    below=6mm of GradBox
] (observer)
{   \scriptsize $\begin{aligned}
        \dot{\eta} &= p_1 \!+\! \Pi_{yu}^\top p_2 \!+\! N_1^T \! \lambda \!+\! p_3 \\
        \dot{\lambda} &= Eu + Fy
    \end{aligned}
    $
};

\node[
    block,
    minimum width=1.5cm,
    below=6mm of observer
] (controller)
{   
    \small
    $K(s)$
};

\node[above left=1pt and 1em of plant.west] {$u$};
\node[above right=1pt and 1em of plant.east] {$y$};


\path coordinate (LeftBus)  at ([xshift=-1cm]plant.west);
\path coordinate (RightBus) at ([xshift=1cm]plant.east);

\path coordinate (LeftGrad)  at (LeftBus  |- GradBox.center);
\path coordinate (RightGrad) at (RightBus |- GradBox.center);

\path coordinate (LeftObs)   at (LeftBus  |- observer.center);
\path coordinate (RightObs)  at (RightBus |- observer.center);


\draw[<-]
    (plant.west)
    -- ([xshift=-1cm]plant.west)
    |- 
    (controller.west);

\draw[->]
    (plant.east)
    -- ([xshift=1cm]plant.east)
    |-
    (controller.east);

\draw[->]
    ([xshift=1cm]plant.east)
    -- ++(1cm,0);

\node[split] at ([xshift=1cm]plant.east) {};


\draw[->]
    ([yshift=0.3cm]plant.north)
    node[above] {$w$}
    -- (plant.north);


\draw[->]
    (LeftGrad)
    -- (GradBox.west);

\draw[->]
    (RightGrad)
    -- (GradBox.east);

\node[split] at (LeftGrad) {};
\node[split] at (RightGrad) {};


\draw[->]
    (LeftObs)
    -- (observer.west);

\draw[->]
    (RightObs)
    -- (observer.east);

\node[split] at (LeftObs) {};
\node[split] at (RightObs) {};

\draw[->]
    (GradBox.south)
    -- node[left] {$p$}
    (observer.north);


\draw[->]
    (observer.south)
    -- node[left] {$\eta$}
    (controller.north);

\node[
    draw,
    ozf1,
    thick,
    dashed,
    fit=(GradBox)(observer)(controller),
    inner xsep=4mm,
    inner ysep=3.5mm,
    label={[anchor=north east]south east:{\footnotesize \textcolor{ozf1}{Dynamic Optimization Algorithm}}}
] {};

\node[
    anchor=north east,
    font=\scriptsize,
    xshift=1.5mm,
    yshift=0.5mm
] at (observer.south east)
{optimality model};

\node[
    anchor=north east,
    font=\scriptsize,
    xshift=7mm,
    yshift=0.5mm
] at (GradBox.south east)
{first-order oracle};

\node[
    anchor=north east,
    font=\scriptsize,
    xshift=11.5mm,
    yshift=0.8mm
] at (controller.south east)
{controller action};


\coordinate (landscape) at ([xshift=2.3cm, yshift=-0.6cm]plant.east);

\begin{scope}[shift=(landscape), scale=0.45]
    \foreach \r/\op in {1.0/0.1, 0.72/0.15, 0.48/0.2, 0.28/0.3} {
        \draw[istblue, thick, opacity=\op, rotate=25]
            (0,0) ellipse ({\r*2.2} and {\r*1.4});
    }

    \fill[istblue] (0,0) circle (3pt);
    \node[below right, istblue, align=center] at (0.05,-0.05) {$u^\star\!,y^\star$};

    \draw[-{Latex[length=0.8mm, width=0.8mm]}, black!80, rounded corners=2pt, thick]
    plot[smooth, tension=1.3] coordinates {
        (-1.6, 0.3)
        (-1.0, 0.5)
        (0.6, 0.6)
        (0.7, -0.15)
        (0.05, -0.3)
        (0, 0)
    };
\end{scope}

\node[
    anchor=north east,
    font=\scriptsize,
    xshift=-4cm,
    yshift=0.5mm
] at (controller.south east)
{};

\end{tikzpicture}

%% file: sections/problem_setting.tex
\section{Problem Formulation and Motivation}\label{sec:setup}

\subsection{The Plant and Optimal Steady-State Problem}

We consider the LTI plant
\begin{align}\label{eq:plant}
\begin{aligned}
    \dot{x} &= A x + B u + B_w w, \\
    y &= C x + D u + D_w w,
\end{aligned}
\end{align}
with state $x(t)\in\R^{n_x}$, control input $u(t)\in\R^{n_u}$, measured output $y(t)\in\R^{n_y}$, and constant exogenous disturbance \mbox{$w\in\R^{n_w}$}. Throughout the paper we impose the following assumption.
\begin{assumption}\label{assum:sys}
    The matrix $A$ is invertible, and the pairs $(A,B)$ and $(A,C)$ are stabilizable and detectable, respectively.
\end{assumption}
Under the invertibility assumption, \eqref{eq:plant} admits a unique (possibly unstable) equilibrium for any constant \mbox{$u(t)\equiv \bar u$}, given by the solution to \mbox{$0=A \bar x + B \bar u + B_w w$}. Solving for $\bar{x}$ yields the steady-state maps
\begin{subequations}\label{eq:steady-state}
\begin{align}
    \bar x &= \underbrace{(-A^{-1} B)}_{\eqcolon \Pixu} \bar u + \underbrace{(-A^{-1} B_w)}_{\eqcolon \Pixw} w, \\
    \bar y &= \underbrace{(D + C \Pixu)}_{\eqcolon \Piyu} \bar u + \underbrace{(D_w + C \Pixw)}_{\eqcolon \Piyw} w,
\end{align}
\end{subequations}
which, given some $w$, associates every input $\bar u$ with a corresponding equilibrium output $\bar y$. The invertibility assumption is made for notational convenience and can be relaxed, in which case we can define steady-state maps as described in Appendix~\ref{appendix:generalization}.

Given objectives $\Phi_1 \colon \R^{n_u}\to\R$ and $\Phi_2 \colon \R^{n_y}\to\R$, our goal is to design a controller that steers the plant input and output towards the solution of the \emph{optimal steady-state (OSS) problem} \cite{lawrence2018,lawrence2021}
\begin{subequations}\label{eq:oss}
\begin{align}
    \min_{ u, y} &\quad \Phi_1( u) + \Phi_2( y) \\
    \label{eq:oss:b}
    \text{s.t.} &\quad  y = \Piyu  u + \Piyw w.
\end{align}
\end{subequations}
The OSS problem is parametric in $w$ and unconstrained if the variable ${y}$ is eliminated. In Sections \ref{sec:performance} and \ref{sec:constraints} we will consider also more general cases of time-varying disturbances and constrained problems.
We introduce the following regularity assumptions on the OSS problem.
\begin{assumption}\label{assum:cost}
    The OSS problem \eqref{eq:oss} satisfies:
    \begin{enumerate}
        \item[(i)] $\Phi_1$ is $m_u$-strongly convex and $L_u$-smooth, where \mbox{$0<m_u\leq L_u < \infty$.}
        \item[(ii)] $\Phi_2$ is convex and $L_y$-smooth, where $0\leq L_y < \infty$.
        \item[(iii)] The input-output sensitivity $\Pi_{yu}$ has full row rank.
    \end{enumerate}
\end{assumption}
Assumption~\ref{assum:cost} is fairly standard (e.g. \cite{colombino2018,simpsonporco2022,yousefi}) and ensures that~\eqref{eq:oss} admits a unique parametric solution $(u^\star(w), y^\star(w))$. 
If $w$ is known, a natural approach to solve the OSS control problem is to eliminate \eqref{eq:oss:b}, giving
\begin{equation}\label{eq:oss-unconstrained}
    \min_{u} \Phi_1(u) + \Phi_2(\Piyu u + \Piyw w),
\end{equation}
and to run a gradient flow 
\begin{align}\label{eq:reduced-gradient}
    \begin{aligned}
    \dot{u} &= - \nabla \Phi_1(u) - \Piyu\T \nabla \Phi_2(\Piyu u + \Piyw w).
    \end{aligned}
\end{align}
Under Assumption~\ref{assum:cost}, the dynamics~\eqref{eq:reduced-gradient} converge asymptotically to~${u}^\star$ \cite{hauswirth2020}. If $A$ is additionally Hurwitz, then applying the feedforward control $u$ to~\eqref{eq:plant} also drives the output to~${y}^\star$. 
In practice, however, evaluating $\nabla \Phi_2(\Piyu u + \Piyw w)$ exactly is impossible. The idea of feedback optimization is to replace the unknown quantity $\Piyu u + \Piyw w$ with the real-time measurement $y$, relying on a stable and sufficiently fast plant to make this approximation accurate. This motivates the \emph{gradient-flow controller} (GFC) \cite{menta2018,colombino2018}
\begin{equation}\label{eq:gradient-flow}
    \dot u = - \varepsilon \left( \nabla \Phi_1(u) + \Piyu\T \nabla \Phi_2(y) \right),
\end{equation}
where $\varepsilon > 0$ is a gain that enforces timescale separation between the plant~\eqref{eq:plant} and optimizer \eqref{eq:gradient-flow}. The following condition on $\varepsilon$ guaranteeing asymptotic stability of this feedback interconnection was established via singular perturbation analysis.
\begin{proposition}[\!\!\cite{menta2018},\cite{hauswirth2020},\cite{cothren}]\label{prop:singular-perturb}
Suppose $A$ is Hurwitz and let $X, Q \succ 0$ be matrices satisfying \mbox{$A^\top X + X A = - Q$}. Under Assumption~\ref{assum:cost}, the closed loop of \eqref{eq:plant} and \eqref{eq:gradient-flow} is asymptotically stable and converges to the solution of \eqref{eq:oss} if
\begin{equation}\label{eq:eps-bound-menta}
    \varepsilon < \frac{\lambda_{\min}(Q)}{2 \| X \Pi_{xu} \| \ell },
\end{equation}
where $\ell = L_y \| C \| \| \Pi_{yu} \|$.
\end{proposition}

The appeal of the GFC lies in its simplicity and interpretability; the dynamics of~\eqref{eq:gradient-flow} simply follow the first-order optimality residual.
The only plant knowledge required is the input-output steady-state map $\Pi_{yu}$ (i.e., the DC gain), and the bound~\eqref{eq:eps-bound-menta} gives an explicit recipe to tune $\varepsilon$ as a function of known quantities. In practice, however, this bound is often highly conservative, forcing the optimizer to run far slower than necessary, degrading both transient performance and disturbance tracking \cite{hauswirth2020,colombino2018}. This motivates a more careful analysis of the closed loop and controller structure.

\subsection{Integration into the Robust Control Framework}\label{sec:framework}

The conservatism of the bound~\eqref{eq:eps-bound-menta} can be traced to the singular perturbation analysis. A tighter analysis becomes possible once the closed loop is rewritten in a form amenable to robust control tools. To this end, instead of just introducing a proportional action on the first-order residual like in the GFC, we define an optimality model \cite{lawrence2021}
\begin{equation}\label{eq:controller-eta}
    \dot \eta = \nabla \Phi_1(u)+ \Pi_{yu}^\top \nabla \Phi_2(y),
\end{equation}
and note that \eqref{eq:gradient-flow} can be expressed as \mbox{$u = - \varepsilon \eta$}.
Now define the auxiliary inputs and outputs \mbox{$(q_1, q_2) \coloneq (u,y)$} and write \eqref{eq:plant} and \eqref{eq:gradient-flow} as the following generalized plant
\begin{subequations}\label{eq:genplant}
\begin{align}\label{eq:genplant:P}
    \setlength{\arraycolsep}{3pt}
    \left[\!
    \begin{array}{c}
        \dot x \\ \dot \eta \\ \hline q_1 \\ q_2 \\ \hdashline y_c
    \end{array}
    \!\right]
    &\!=\!
    \left[\!
    \setlength{\arraycolsep}{4pt}
    \begin{array}{cc|cc:c:c}
        A & 0 & 0 & 0 & B_w & B \\
        0 & 0 & I_{n_u} & \Piyu^\top & 0 & 0 \\ \hline
        0 & 0 & 0 & 0 & 0 & I_{n_u} \\
        C & 0 & 0 & 0 & D_w & D \\ \hdashline
        0 & I_{n_u} & 0 & 0 & 0 & 0
    \end{array}
    \!\right]
    \!\!
    \setlength{\arraycolsep}{2pt}
    \left[\!
    \begin{array}{c}
        x \\ \eta \\ \hline p_1 \\ p_2 \\ \hdashline w \\ \hdashline u
    \end{array}
    \!\right]\!\!, \\
    \label{eq:genplant:oracle}
    p_1 &= \nabla \Phi_1(q_1), \quad p_2 = \nabla \Phi_2(q_2), \\
    \label{eq:genplant:K}
    u &= -\varepsilon \, y_c.
\end{align}
\end{subequations}
\begin{figure}
	\centering
	\input{figures/PKD.tex}
	\caption{Feedback-optimization in $P$-$K$-$\Delta$ structure, with an augmented IQC filter $\Psi$ for analysis and synthesis.}
	\label{fig:pkdelta}
\end{figure}
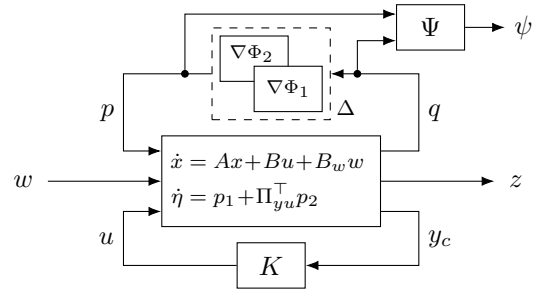
Eq.~\eqref{eq:genplant} emphasizes that in the LTI case, feedback optimization with the GFC is just a P-controller acting on a generalized plant in $P$-$K$-$\Delta$ structure, with~\eqref{eq:genplant:P} as plant $P$, and $\nabla \Phi_1, \nabla \Phi_2$ as uncertainty $\Delta$. We refer to the $(q,p)$ channel in the following as the \emph{oracle channel}.
The formulation~\eqref{eq:genplant} naturally suggests to consider the robust control tools. By Assumption~\ref{assum:cost}, both $\nabla \Phi_1$ and $\nabla \Phi_2$ are slope-restricted operators, so the closed loop is a Lur'e system~\cite{lessard2016}. The analysis of such systems through Integral Quadratic Constraints is well established \cite{carrasco2014,veenman},
and carries over to this specific setting with minimal effort. Moreover, the $P$-$K$-$\Delta$ structure (see Fig.~\ref{fig:pkdelta}) suggests three directions that go beyond reducing the conservatism of the stability condition:
\begin{enumerate}
    \item Including a performance output $z$ in \eqref{eq:genplant:P} allows to consider transient performance on the channel $(w,z)$ in addition to only stabilization.
    \item Synthesizing a dynamic controller $K(s)$ in \eqref{eq:genplant:K} instead of only a static gain can provide additional degrees of freedom and performance improvements.
    \item The availability of a $\Delta$ channel provides a principled way to characterize both model uncertainty and other monotone operators, enabling feedback optimization with gray-box models and proximal algorithms \cite{bauschke}.
\end{enumerate}
We develop these three directions in the remainder.

%% file: figures/PKD.tex
\begin{tikzpicture}[
  block/.style = {
    draw,
    minimum width=0.9cm,
    minimum height=0.6cm,
    align=center
  },
  ->,
  >=Latex,
  node distance=1.5em,
  split/.style={
    draw,
    fill,
    circle,
    minimum size=0.25em,
    inner sep=0pt
  },
]


\node[block] (GradA)
{\raisebox{5pt}{\hspace{-2pt}\scriptsize $\nabla \Phi_2$}};

\node[
    block,
    fill=white,
    xshift=4.5mm,
    yshift=-4mm
] (GradB)
at (GradA.center)
{\scriptsize $\nabla \Phi_1$};

\node[
    draw,
    dashed,
    fit=(GradA)(GradB),
    inner sep=1mm,
    color=black
] (GradBox) {};

\node[
    anchor=south west,
    font=\footnotesize,
    xshift=-0.5mm,
    yshift=-0.5mm,
    color=black
] at (GradBox.south east)
{$\Delta$};


\node[
    block,
    below=2mm of GradBox,
    minimum height=1.2cm,
    minimum width=1.8cm
] (Algo)
{ 
  \footnotesize
  $\begin{aligned}
  \dot{x} &= A x \!+\! B u \!+\! B_w w \\
  \dot{\eta} &= p_1\! +\! \Pi_{yu}^\top p _2
  \end{aligned}$
};

\node[
    block,
    below=2mm of Algo
] (K)
{ $K$};


\node[split]
    (splitX)
    at ($(GradBox.west) + (-1em,0)$) {};

\node[split]
    (splitU)
    at ($(GradBox.east) + (1em,0)$) {};


\node[
    block
] (Psi)
at ($(GradBox.north)+(6em,0em)$)
{$\Psi$};

\draw[->]
    (splitX)
    |- ($(Psi.west)+(0em,0.5em)$);

\draw[->]
    (splitU)
    |- ($(Psi.west)+(0em,-0.5em)$);

\draw[->]
    (Psi.east)
    -- ++(1.5em,0)
    node[right] { $\psi$};


\coordinate (AlgoWestUp)
    at ([yshift=0.4cm]Algo.west);

\coordinate (AlgoEastUp)
    at ([yshift=0.4cm]Algo.east);

\coordinate (AlgoWestDown)
    at ([yshift=-0.4cm]Algo.west);

\coordinate (AlgoEastDown)
    at ([yshift=-0.4cm]Algo.east);

\draw[<-]
    (AlgoWestUp)
    -- ++(-0.5,0)
    |- node[pos=0.25, left] { $p$}
    (GradBox.west);

\draw[->]
    (AlgoEastUp)
    -- ++(0.5,0)
    |- node[pos=0.25, right] { $q$}
    (GradBox.east);

\draw[<-]
    (AlgoWestDown)
    -- ++(-0.5,0)
    |- node[pos=0.25, left] { $u$}
    (K.west);

\draw[->]
    (AlgoEastDown)
    -- ++(0.5,0)
    |- node[pos=0.25, right] { $y_c$}
    (K.east);


\draw[->]
    ([xshift=-1.5cm]Algo.west)
    -- node[pos=0,left,xshift=-2pt]
    { $w$}
    (Algo.west);

\draw[->]
    (Algo.east)
    -- node[pos=1,right,xshift=2pt]
    { $z$}
    ++(1.5cm,0);

\end{tikzpicture}

%% file: sections/stability.tex
\section{Stability Analysis with IQCs}\label{sec:stability}

In this section, we briefly recall the Zames--Falb characterization of slope-restricted operators and formulate a general stability test for feedback optimization controllers. We then specialize to the GFC and illustrate how to deal with non-strongly convex problems and establish a connection to classical timescale separation bounds.

\subsection{IQCs for Slope Restricted Operators}\label{sec:iqc_slope}

We first revisit characterizations of first-order oracles with IQCs.
An operator $\phi: \mathbb{R}^n \to \mathbb{R}^n$ is said to be slope-restricted in the sector $[m,L]$ if
\begin{multline}\label{eq:slope-restriction}
m \| x - y \|_2^2 \leq (x - y)^\top (\phi(x) - \phi(y)) \leq L \| x - y \|_2^2 \\ \forall x, y \in \mathbb{R}^n.
\end{multline}
When $\phi$ is set-valued, condition~\eqref{eq:slope-restriction} is required to hold for every element in the sets $\phi(x)$ and $\phi(y)$.
By Assumption~\ref{assum:cost}, the gradients $\nabla\Phi_1$ and $\nabla\Phi_2$ are slope-restricted in the sectors $[m_u, L_u]$ and $[0, L_y]$, respectively, so that this property applies directly to the oracle channel~\eqref{eq:genplant:oracle} \cite{lessard2016}. The following classical result provides an IQC characterization that we will use throughout.

\begin{lemma}[Zames--Falb IQC]\label{lem:zames-falb}
Let $\phi: \mathbb{R}^n \to \mathbb{R}^n$ be slope-restricted in $[0, L]$. Let $H$ be a real rational transfer function with $\norm{H}_{1} \leq 1$. Then, for any $q \in \mathcal{L}_2$ and $p(t) = \phi(q(t)) \, \forall t$, it holds 
\begin{subequations}\label{eq:zf-iqc}
\begin{equation}\label{eq:zf-iqc:ineq}
    \int_0^T \psi(t)^\top J_n\,  \psi(t)\, \mathrm{d}t \geq 0, \quad 
\end{equation}
for all $T \geq 0$, where $\psi = \Psi \left[ \,\mathrm{col}(q, p) \, \right]$ and
\begin{equation}\label{eq:zf-iqc:MPsi}
    \setlength{\arraycolsep}{4pt}
    J_n \coloneq \bmat{0 & 1 \\ 1 & 0} \!\otimes\! I_n, \quad
    \Psi(s) \! \coloneq 
    \setlength{\arraycolsep}{4pt}
    \bmat{(1 \!-\! H(s)) L & H(s) \!-\! 1 \\ 0 & 1} \!\otimes\! I_n.
\end{equation}
\end{subequations}
\end{lemma}

Strictly speaking, Lemma~\ref{lem:zames-falb} is a time-domain condition resulting from a factorization of the more general frequency-domain definition \cite{carrasco2014}. The transfer matrix $\Psi^* J_n \Psi$ is the so-called IQC \emph{multiplier}. Lemma~\ref{lem:zames-falb} extends to other sector definitions. E.g., for $\Phi_1$ with sector $[m_u,L_u]$, one defines the shifted operator \mbox{$\nabla \tilde \Phi_1(q_1) \coloneqq \nabla \Phi_1(q_1) - m_u q_1$}, which becomes slope-restricted in $[0, L_u\!\!-\!\!m_u]$. For the sector $[0, \infty]$, we can replace the filter in \eqref{eq:zf-iqc:MPsi} with \mbox{$\Psi(s) = \mathrm{diag}(1-H(s), 1) \otimes I_n$}, which will become relevant for normal-cone operators in Section~\ref{sec:constraints}.

For the specific choice of $H = 0$, the filter $\Psi$ becomes memoryless, and~\eqref{eq:zf-iqc} reduces to the \mbox{pointwise-in-time condition}
\begin{equation*}
\bmat{q(t) \\ p(t)}^\top \bmat{0 & L I_n \\ L I_n & -2 I_n} \bmat{q(t) \\ p(t)} \geq 0, \quad \forall t \geq 0.
\end{equation*}
Generally, the search for a suitable parametrization $H$ is nontrivial. Simple choices satisfying $\| H \|_1 \leq 1$ are \mbox{$H(s) = \frac{\omega}{s + \omega}$} for some $\omega > 0$, or a convex combination of such basis functions for different $\omega$. Typically, the richer the dynamics of~$H$, the less conservative is the operator description. In practice, the parametrization is chosen so that the filter parameters enter affinely into some linear matrix inequalities (LMIs) that result from the IQC application; we refer to~\cite{carrasco2014,veenman,veenmanAnalysis} and references therein for further details.

\subsection{An IQC Stability Test for Feedback Optimization}\label{sec:stability-general}

We now apply the Zames--Falb characterization to the generalized plant~\eqref{eq:genplant} and formulate a verifiable stability condition.

Note that the fixed-point of~\eqref{eq:genplant} coincides with the solution of the OSS problem~\eqref{eq:oss}. Since the disturbance $w$ enters only as a constant bias, the following analysis is done in centered coordinates relative to its fixed-point.
Let $(u^\star, y^\star)$ denote the solution to~\eqref{eq:oss} and $x^\star$ be the corresponding equilibrium state, and define the optimal gradients \mbox{$p_1^\star \coloneqq \nabla \Phi_1(u^\star)$} and \mbox{$p_2^\star \coloneqq \nabla \Phi_2(y^\star)$}, satisfying the first-order optimality condition \mbox{$p_1^\star + \Pi_{yu}^{\top} p_2^\star = 0$}. We introduce $\tilde{u}\coloneqq u - u^\star$, $\tilde{y}\coloneqq y - y^\star$, $\tilde{x}\coloneq x - x^\star$, and define $\tilde{q}_1, \tilde{q}_2, \tilde{\eta}, \tilde{y}_c$ accordingly, cf.~\eqref{eq:genplant}. Using straightforward calculations, we can rewrite \eqref{eq:genplant:P} as
\begin{subequations}\label{eq:genplant-stab}
\begin{equation}\label{eq:genplant-stab:P}
\setlength{\arraycolsep}{3pt}
\left[\!
\begin{array}{c}
\dot{\tilde x} \\ \dot{\tilde{\eta}} \\ \hline \tilde q_1 \\ \tilde q_2 \\
\hdashline \tilde y_c
\end{array}\!
\right]
=
\underbrace{\left[\!
\begin{array}{cc|cc:c}
A & 0  & 0 & 0 & B \\
0 & 0  & I_{n_u} & \Pi_{yu}^\top & m_u I_{n_u} \\ \hline
0 & 0  & 0 & 0  & I_{n_u} \\
C & 0  & 0 & 0  & D \\ \hdashline
0 & I_{n_u} & 0  & 0 & 0
\end{array}
\right]}_{\doteq P(s)}\!\!
\left[\!
\begin{array}{c}
\tilde x \\ \tilde{\eta} \\ \hline \tilde{p}_1 \\ \tilde{p}_2 \\ \hdashline \tilde u
\end{array}\!
\right],
\end{equation}
where 
\begin{align}\label{eq:genplant-stab:oracle}
    \begin{aligned}
	\tilde p_1 &= \nabla \Phi_1(\tilde q_1 + u^\star) - m_u \tilde q_1 - p_1^\star \\
	\tilde p_2 &= \nabla \Phi_2(\tilde q_2 + y^\star) - p_2^\star
    \end{aligned}
\end{align}
\end{subequations}
are the shifted oracles, each defining slope-restricted mappings in $[0, L_u - m_u]$ and $[0, L_y]$, respectively.
Let $K$ be a controller for \eqref{eq:genplant-stab}, i.e., $\tilde{u} = K [\, \tilde{y}_c \,]$, and let \mbox{$P \star K$} be the resulting LFT that maps \mbox{$\mathrm{col}(\tilde{p}_1, \tilde{p}_2) \mapsto \mathrm{col}(\tilde{q}_1, \tilde{q}_2)$}. 

To formulate an IQC stability test, let $H_u(s)$ and $H_y(s)$ be Zames--Falb parametrizations and define the filter
\begin{align}
    \Psi_{uy}(s) &\coloneq 
    \setlength{\arraycolsep}{1pt}
    \bmat{
        M_u(s)(L_u \!-\! m_u) & 0 & -M_u(s) & 0 \\
        0 & 0 & I_{n_u} & 0 \\
        0 &  M_y(s) L_y & 0 & -M_y(s)  \\
        0 & 0 & 0 & I_{n_y}
    }
    \notag
    \\
    M_u(s) &\coloneq (1 - H_u(s)) \otimes I_{n_u}\notag\\ 
    M_y(s) &\coloneq (1 - H_y(s)) \otimes I_{n_y}.\label{eq:ZF-filters}
\end{align}
If we define $\psi \coloneq \Psi_{uy} [\, \mathrm{col}(\tilde{q}_1, \tilde{q}_2 , \tilde{p}_1, \tilde{p}_2) \, ]$ and partition \mbox{$\psi = \mathrm{col}(\psi_1, \psi_2)$}, where $\psi_1(t) \in \mathbb{R}^{2 n_u}$, $\psi_2(t) \in \mathbb{R}^{2 n_y}$, then $\psi_1, \psi_2$ each satisfy the IQC inequality \eqref{eq:zf-iqc:ineq} with $J_{n_u}$ and $J_{n_y}$ as in \eqref{eq:zf-iqc:MPsi}, respectively.

Now define the augmented plant
\begin{equation}\label{eq:augmented-plant}
    \mathcal{G}(s) \coloneq \Psi_{uy}(s)
    \bmat{(P \star K)(s) \\ I_{n_u + n_y}} \doteq \left[ \begin{array}{c|c} \mathcal{A} & \mathcal{B} \\ \hline
    \mathcal{C} & \mathcal{D} \end{array} \right],
\end{equation}
which maps $\mathrm{col}(\tilde{p}_1, \tilde{p}_2) \mapsto \mathrm{col}(\psi_1, \psi_2)$. Then, asymptotic stability of \eqref{eq:genplant-stab} can then be certified by imposing the following condition on $\mathcal{G}$.

\begin{proposition}\label{prop:stability-general}
    Let Assumptions~\ref{assum:sys} and~\ref{assum:cost} hold. If there exist admissible filter parametrizations $H_u$, $H_y$ satisfying \mbox{$\| H_u \|_1 \leq 1$} and $\| H_y \|_1 \leq 1$, a matrix $\mathcal{X} = \mathcal{X}^\top \succ 0$, and a scalar $\lambda > 0$ such that 
    \begin{equation}\label{eq:LMI-zf-general}
        \bmat{
            \mathcal A^\top\mathcal X+\mathcal X\mathcal A & \mathcal X \mathcal B
            \\
            \mathcal B^\top\mathcal X & 0
        }
        + \lambda
        \bmat{
            \mathcal C & \mathcal D
        }^\top
        \mathcal{J}
        \bmat{
            \mathcal C & \mathcal D
        } \prec 0,
    \end{equation}
    with $\mathcal{J} = \mathrm{diag}( J_{n_u}, J_{n_y})$, then the closed loop of \eqref{eq:genplant-stab} and the controller $\tilde{u} = K[\, \tilde{y}_c \,]$ is asymptotically stable.
\end{proposition}

\begin{proof}
    See Appendix~\ref{appendix:a}.
\end{proof}

Proposition~\ref{prop:stability-general} is a standard result and eq.~\eqref{eq:LMI-zf-general} is an efficiently verifiable LMI if $\mathcal{A},\mathcal{B},\mathcal{C},\mathcal{D}$ are constant or if $\mathcal{C},\mathcal{D}$ are affine functions of the Zames--Falb parameters and $\lambda=1$. By asymptotic stability of \eqref{eq:genplant-stab}, the closed loop trajectories converge to the optimal solution of~\eqref{eq:oss}. We then conclude by going back to original coordinates that plant \eqref{eq:plant} in feedback with the controller
\begin{align}\label{eq:controller-implementation}
    \begin{aligned}
    \dot{\eta} &= \nabla \Phi_1(u) + \Pi_{yu}^\top \nabla \Phi_2(y) \\
    u &= K \left[\, \eta\, \right]
    \end{aligned}
\end{align}
is stabilized and the OSS problem is solved asymptotically.

\begin{remark}
    If additionally guarantees on the rate of convergence are desired, the IQC test can be modified straightforwardly with the methodology proposed in~\cite{hu_expIQC}, which provides IQC-based conditions to also establish exponential convergence rates.
\end{remark}

\subsection{Stability Analysis of the GFC}\label{sec:convex-case}

Proposition~\ref{prop:stability-general} applies to any controller $K$. To recover the GFC~\eqref{eq:gradient-flow}, one simply sets \mbox{$K(s) \equiv -\varepsilon I_{n_u}$}. Inspecting~\eqref{eq:genplant-stab}, the realization of $P \star (-\varepsilon I_{n_u})$ depends affinely~on $\varepsilon$, and \eqref{eq:LMI-zf-general} therefore becomes a quasi-convex program in $\varepsilon$. Proposition~\ref{prop:stability-general} then provides a principled way to find upper bounds on $\varepsilon$ by performing a bisection. 

In the following, we show that in absence of $\Phi_1$ we can exactly recover the singular perturbation bound \eqref{eq:eps-bound-menta} of Proposition~\ref{prop:singular-perturb} as a conservative special case of the IQC analysis. To this end, we impose the following assumption.

\begin{assumption}\label{assum:noPhi1}
    The input cost is absent, i.e., $\Phi_1 = 0$. 
\end{assumption}

In absence of $\Phi_1$ we lose strong convexity and the \mbox{$\tilde\eta$-dynamics} cannot be simply stabilized with a static output feedback. Strict inequality of \eqref{eq:LMI-zf-general} is therefore not possible. Moreover, the minimizer $(u^\star, y^\star)$ is generally not unique anymore. In the following, we therefore present an alternative strategy to establish closed loop convergence in absence of strong convexity. We note that the proof technique extends \emph{mutatis mutandis} to the case where $\Phi_1$ is non-zero and convex.

To aid the analysis, we consider the steady-state error
\begin{equation}\label{eq:error}
    e \coloneqq x - \Pixu u - \Pixw w,
\end{equation}
which vanishes on the equilibrium manifold. Take a minimizer $(u^\star, y^\star)$ solving the OSS problem. Under the GFC~\eqref{eq:gradient-flow}, the dynamics in centered coordinates become
\begin{align}\label{eq:lure}
\begin{aligned}
\bmat{\dot{\tilde u} \\ \dot e} \!&=\! 
\underbrace{\bmat{0 & 0 \\ 0 & A}}_{\eqqcolon \Ahat}
\!\bmat{\tilde u \\ e}
+
\underbrace{\bmat{-\varepsilon \Piyu^{\top} \\ \varepsilon \Pixu \Piyu^{\top}}}_{\eqqcolon \Bhat}
\tilde p_2, \\[-1ex]
\tilde q_2 \!&=\!
\underbrace{\bmat{\Piyu & C}}_{\eqqcolon \Chat}
\bmat{\tilde u \\ e}.
\end{aligned}
\end{align}
Let \mbox{$\hat G(s) = \Chat(sI - \Ahat)^{-1}\Bhat$} denote the transfer function of~\eqref{eq:lure}, and analogously to the last section, build an augmented system
\begin{equation}\label{eq:augmented-tf}
  \mathcal{G}(s)
  \coloneq
  \setlength{\arraycolsep}{4pt}
    \left(\bmat{(1 \!-\! H_y(s)) L_y & H_y(s) \!-\! 1 \\ 0 & 1} \!\otimes\! I_{n_y} \right)
    \bmat{\hat G(s) \\ I_{n_y}},
\end{equation}
with a Zames--Falb parametrization $H_y(s)$. Let $H_y$ have the state-space realization $(A_H, B_H, C_H, D_H)$. Then, a realization of \eqref{eq:augmented-tf} is
\begin{align}
    \mathcal A
    &\!=\!\!\!
    \setlength{\arraycolsep}{1.5pt}
    \begin{bmatrix}
        0 & 0 & 0 \\
        0 & A & 0 \\
        \!B_H \!\otimes\! (L_y\Pi_{yu}) & B_H \!\otimes\! (L_y C) & A_H \!\otimes\! I_{n_y}
    \end{bmatrix}\!\!\!, 
    \mathcal B
    \!=\!\!\!
    \setlength{\arraycolsep}{1pt}
    \begin{bmatrix}
        -\varepsilon\Pi_{yu}^\top \\
        \varepsilon\Pi_{xu}\Pi_{yu}^\top \\
        \! -B_H \!\otimes\! I_{n_y} 
    \end{bmatrix}
    \notag \\
    \mathcal C
    &\!=\!
    \setlength{\arraycolsep}{3pt}
    \begin{bmatrix}
        \tilde{d} L_y\Pi_{yu} & \tilde{d} L_yC & -C_H \!\otimes\! I_{n_y} \\
        0 & 0 & 0
    \end{bmatrix}, 
    \qquad \quad  \mathcal D
    \!=\!
    \setlength{\arraycolsep}{2pt}
    \begin{bmatrix}
        -\tilde{d} I_{n_y} \\
        I_{n_y}
    \end{bmatrix}
    \label{eq:augmented-plant-realization}
\end{align}
where $\tilde{d} \coloneq 1 - D_H$.
We note that \eqref{eq:augmented-tf} is marginally stable, and $\varepsilon$ only enters through $\mathcal{B}$. The next result ensures convergence to an optimum by placing an alternative LMI condition.

\begin{proposition}\label{prop:zf-LMI}
Let Assumption~\ref{assum:sys}, Assumption~\ref{assum:cost} (ii) and (iii), and Assumption~\ref{assum:noPhi1} hold. Moreover, let $A$ be Hurwitz and let $\varepsilon^\star_{\mathrm{ZF}}$ be the maximum $\varepsilon>0$ for which there exist a filter parametrization $H_y$ satisfying $\| H_y \|_1 \leq 1$, a matrix $\mathcal{X} = \mathcal{X}^\top \succ 0$ and scalar $\lambda > 0$, such that
\begin{multline}\label{eq:LMI-zf}
    \Lambda(\mathcal X, \lambda, \varepsilon)
    \coloneqq
    \bmat{
        \mathcal A^\top\mathcal X+\mathcal X\mathcal A
        &
        \mathcal X\mathcal B
        \\
        \mathcal B^\top\mathcal X
        &
        0
    }
    \\
    + \lambda
    \bmat{
        \mathcal C & \mathcal D
    }^\top
    J_{n_y}
    \bmat{
        \mathcal C & \mathcal D
    } \preceq 0.
\end{multline}
Then for every $\varepsilon\in(0,\varepsilon^\star_{\mathrm{ZF}}]$ the closed loop trajectories of~\eqref{eq:lure} are bounded. Moreover, partition the LMI conformally into 
\begin{equation*}
\mathcal{X} = \bmat{\star & \star \\ \star & \mathcal{X}_0}, 
\quad 
\Lambda(\mathcal{X},\lambda,\varepsilon) = \bmat{\star & \star \\ \star & \Lambda_0(\mathcal{X}_0,\lambda,\varepsilon)},
\end{equation*}
where $\star$ indicates the block-row and column associated with the $\tilde{u}$ direction. If additionally $\Lambda_0(\mathcal{X}_0,\lambda,0) \prec 0$ holds, then for all $\varepsilon\in(0,\varepsilon^\star_{\mathrm{ZF}})$ the closed loop trajectories converge to a solution of \eqref{eq:oss}.
\end{proposition}

\begin{proof}
    See Appendix~\ref{appendix:b}.
\end{proof}

The extension to include a convex input cost $\Phi_1 \neq 0$ with $m_u=0$ follows the same reasoning. However, omitting $\Phi_1$ enables us to provide a direct connection of IQC-based stability analysis with Proposition~\ref{prop:singular-perturb}, as stated in the following Corollary.

\begin{corollary}\label{corr:static-LMI}
Suppose the same assumptions of Proposition~\ref{prop:zf-LMI} hold. Let $X \succ 0$ satisfy $A^\top X + X A = -Q$ with $Q \succ 0$. Let $\varepsilon^\star_{\mathrm{static}}$ be the maximum $\varepsilon>0$ for which there exists a matrix~\mbox{$\mathcal{X} = \mathcal{X}^\top \succ 0$} and scalar~$\lambda > 0$, such that \eqref{eq:LMI-zf} holds with $H_y=0$. 
Then
\begin{equation}\label{eq:eps-bound-LMI}
    \varepsilon^\star_{\mathrm{ZF}} \geq \varepsilon^\star_{\mathrm{static}} \geq \frac{\lmin(Q)}{2 \| X \Pixu \| \, \ell},
\end{equation}
where $\ell = L_y \| C \| \| \Piyu \|$. Moreover, if in addition \mbox{$A^\top X + X A + \frac{\lambda L_y^2}{2} C^\top C \prec 0$} holds, then \mbox{$\Lambda_{0}(\mathcal{X}_0, \lambda, 0) \prec 0$}.
\end{corollary}

\begin{proof}
 See Appendix~\ref{appendix:c}.
\end{proof}

The right-hand side of~\eqref{eq:eps-bound-LMI} recovers the singular perturbation bound~\eqref{eq:eps-bound-menta}, while the chain of inequalities shows its conservatism in comparison to both the static and dynamic IQC tests. The proof of Corollary~\ref{corr:static-LMI} traces this conservatism to a Cauchy--Schwarz relaxation of the LMI~\eqref{eq:LMI-zf} for one specific choice of the certificate $(\mathcal{X},\lambda)$ and a static multiplier ($H_y=0$). Searching over all certificates and multipliers may produce substantially tighter bounds.

\begin{remark}
The condition \mbox{$A^\top X + X A + \tfrac{\lambda L_y^2}{2} C^\top C \prec 0$} required is mild. When $A$ is Hurwitz, detectability of $(A, C)$ guarantees the existence of $\bar X \succ \bar X_0$ satisfying \mbox{$A^\top \bar X + \bar X A + C^\top C \prec 0$}, where $\bar X_0$ is the observability Gramian. Since~\eqref{eq:LMI-zf} is homogeneous in $(\mathcal{X}, \lambda)$, one can always rescale these variables to match any given value of~$\tfrac{\lambda L_y^2}{2}$.
\end{remark}

As an illustrating comparison, consider Table~\ref{tab:comparison-stab} which compares \eqref{eq:eps-bound-menta}, $\varepsilon^\star_{\mathrm{static}}$ and $\varepsilon^\star_{\mathrm{ZF}}$ when computed for the example system~\eqref{eq:example-sys} and a quadratic output cost \mbox{$\Phi_2(y) = \tfrac{1}{2} y^\top Q_y y, \; Q_y\succeq 0$}. Since the gradient $\nabla \Phi_2$ is linear in this case, we can compute the exact stability margin $\varepsilon_{\max}$ with an eigenvalue analysis. Observe that the singular perturbation bound captures less than 9\% of the true stability margin, while already a stability test with static multiplier improves the range almost by a factor of 3, and a test with dynamic multiplier recovers over 85\%.
We conclude that the stability margins obtained from classical singular perturbation arguments can be conservative and systematically reduced by exploiting the IQC analysis. 

\begin{table}
    \centering
    \caption{Stability bounds for the example system~\eqref{eq:example-sys} with quadratic cost $\Phi_2(y) = \tfrac{1}{2} y^\top Q_y y$.}
    \label{tab:comparison-stab}
    \begin{tabular}{lccc}
        \toprule
        Method & $\varepsilon$ bound & Value & \% of $\varepsilon_{\max}$ \\
        \midrule
        Singular Perturbation & \eqref{eq:eps-bound-menta} & $0.010$ & $8.7\%$ \\
        Static IQC & $\varepsilon_\mathrm{static}^\star$ & $0.028$ & $24.5\%$ \\
        Zames--Falb IQC & $\varepsilon_\mathrm{ZF}^\star$ & $0.098$ & $85.6\%$ \\
        \bottomrule
    \end{tabular}
\end{table}

For the remainder of the paper, we work with a strongly convex input cost~$\Phi_1$ again. This simplifies the subsequent developments on performance and synthesis, without the technical overhead of handling non-strict feasibility.

%% file: sections/perf_synth_uncert.tex
\section{From Analysis to Systematic Design}\label{sec:performance}

While stability guarantees that the optimum is reached asymptotically, little is said about the transient behaviour or robustness.
The generalized-plant framework enables a fundamentally different perspective. Rather than analyzing a given feedback optimization algorithm, we can \emph{synthesize} a dynamic output-feedback controller that can account for robustness and transient performance, making a shift from handcrafted optimization-inspired controllers to systematic, LMI-based design \cite{schererGahinet,schererLMIs,veenman}.

\subsection{Optimizing Transients via Robust Performance}\label{subsec:perf}

We propose to capture transient behaviour of the closed loop through classical input-output gains. The central quantity of interest is the optimality residual $\tilde\eta = \eta - \eta^\star$, which integrates the violation of the first-order conditions. If the closed loop is stable, $\tilde{\eta}$ asymptotically vanishes at any constant disturbance, regardless of its magnitude. What triggers a transient is a \emph{change} in the disturbance, and the performance measure is how much optimality violation such a change can cause.

To formalize this, consider a time-varying perturbation $\hat{w}$ around the nominal constant disturbance $w$, so that the plant becomes
\begin{align}
\begin{aligned}
\dot{x} &= A x + B u + B_w (w + \hat{w}) \\
y &= C x + D u + D_w (w + \hat{w}).
\end{aligned}
\end{align}
A nonzero $\hat{w}(t)$ shifts the optimal steady state and triggers a re-optimization transient. We define the performance output $z \coloneq \tilde{\eta}$ and measure performance by the amplification from $\hat{w}$ to $z$ according to some signal norm. 
Shifting to centered coordinates as before, the resulting generalized plant reads
\begin{equation}\label{eq:genplant-perf}
\setlength{\arraycolsep}{3pt}
\left[\!
\begin{array}{c}
\dot{\tilde x} \\ \dot{\eta} \\ \hline \tilde q_1 \\ \tilde q_2 \\ \hdashline z \\
\hdashline \tilde{y}_c
\end{array}\!
\right]
=
\underbrace{\left[\!
\begin{array}{cc|cc:c:c}
A & 0  & 0 & 0 & B_w & B \\
0 & 0  & I_{n_u} & \Pi_{yu}^\top & 0 & m_u I_{n_u} \\ \hline
0 & 0  & 0 & 0 & 0 & I_{n_u} \\
C & 0  & 0 & 0 & D_w & D \\ \hdashline
0 & I_{n_u} & 0 & 0 & 0 & 0 \\ \hdashline
0 & I_{n_u} & 0 & 0 & 0 & 0
\end{array}
\right]}_{\doteq P_\mathrm{perf}(s)}\!\!
\left[\!
\begin{array}{c}
\tilde x \\ \eta \\ \hline {\tilde{p}}_1 \\ \tilde{p}_2 \\ \hdashline \hat w \\ \hdashline \tilde u
\end{array}\!
\right],
\end{equation}
where $\tilde{p}_1, \tilde{p}_2$ are as in \eqref{eq:genplant-stab:oracle}, and where $\tilde{u},\tilde{y}_c$ close the loop through a controller $K$.

For intuition, note that for quadratic costs, \eqref{eq:genplant-perf} is a linear interconnection and the $\hat{w} \to \tilde{\eta}$ gain therefore coincides with the $\dot{\hat w} \to \dot{\tilde \eta}$ gain, i.e., the effect of disturbance changes on instantaneous optimality violation. For illustration purposes, let us consider a disturbance rejection problem in terms of the $\Ltwo$-gain
\begin{equation}
  \gamma \coloneq \sup_{0\neq \hat{w}\in \Ltwo} \frac{\| z \|_{\Ltwo}}{\| \hat{w} \|_{\Ltwo}}.
\end{equation}

To state an IQC-based performance test, fix a controller $K$ and form the LFT~\mbox{$P_\mathrm{perf} \star K$}, mapping \mbox{$(\tilde{p}_1, \tilde{p}_2, \hat{w}) \!\mapsto\! (\tilde{q}_1, \tilde{q}_2, z)$}. Partition the resulting transfer matrix conformally as
\begin{equation}
  (P_\mathrm{perf} \star K)(s) \triangleq \bmat{\hat G_{qp}(s) & \hat G_{qw}(s) \\
  \hat G_{zp}(s) & \hat G_{zw}(s)}.
\end{equation}
Choose Zames--Falb parametrizations $H_u(s)$ and $H_y(s)$, build $\Psi_{uy}(s)$ as in \eqref{eq:ZF-filters} and define \mbox{$\psi = \Psi_{uy}[\, \mathrm{col}(\tilde{q}_1, \tilde{q}_2, \tilde{p}_1, \tilde{p}_2) \,]$} as before. Define the augmented plant
\begin{equation}\label{eq:perf-channels}
  \mathcal{G}(s) \coloneq
  \underbrace{\bmat{
    \Psi_{uy}(s) \bmat{\hat G_{qp}(s) \\ I} & \Psi_{uy}(s) \bmat{\hat G_{qw}(s) \\ 0} \\
    \hat G_{zp}(s) & \hat G_{zw}(s)
  }}_{\doteq  
  \left[ \begin{array}{c|cc} 
    \mathcal{A} & \mathcal{B}_p & \mathcal{B}_w \\ \hline
    \mathcal{C}_\psi & \mathcal{D}_{\psi p} & \mathcal{D}_{\psi w} \\
    \mathcal{C}_z & \mathcal{D}_{zp} & \mathcal{D}_{zw}
  \end{array} \right]
  },
\end{equation}
which maps $(\tilde{p}_1, \tilde{p}_2, \hat{w}) \mapsto (\psi, z)$, where $\psi = \mathrm{col}(\psi_1, \psi_2)$ and $\psi_1,\psi_2$ satisfy the Zames--Falb IQC \eqref{eq:zf-iqc:ineq}.
We can then re-state a standard IQC-based $\Ltwo$-gain result.

\begin{proposition}[\!\!\cite{veenmanAnalysis}]\label{prop:perf-analysis}
Let Assumptions~\ref{assum:sys} and~\ref{assum:cost} hold.
If there exist admissible filter parametrizations $H_u, H_y$ satisfying $\| H_u \|_1 \leq 1$ and $\| H_y \|_1 \leq 1$, and a matrix $\mathcal{X} = \mathcal{X}^\top \succ 0$ such that
\begin{multline}\label{eq:LMI-perf}
  \bmat{\mathcal{A}^\top \mathcal{X} + \mathcal{X}\mathcal{A}
        & \mathcal{X}\mathcal{B}_p & \mathcal{X}\mathcal{B}_w \\[2pt]
        \star & 0 & 0 \\[2pt]
        \star & \star & -\gamma^2 I_{n_w}}
  \\[4pt]
  +\;
  \bmat{\mathcal{C}_\psi & \mathcal{D}_{\psi p}
        & \mathcal{D}_{\psi w}}^{\!\top}
  \mathcal{J} \,
  \bmat{\mathcal{C}_\psi & \mathcal{D}_{\psi p}
        & \mathcal{D}_{\psi w}}
  \\[4pt]
  +\;
  \bmat{\mathcal{C}_z & \mathcal{D}_{zp}
        & \mathcal{D}_{zw}}^{\!\top}
  \bmat{\mathcal{C}_z & \mathcal{D}_{zp}
        & \mathcal{D}_{zw}}
  \;\prec\; 0,
\end{multline}
with $\mathcal{J} = \mathrm{diag}( J_{n_u}, J_{n_y})$, then the $\Ltwo$-gain from $\hat w$ to $z$ is at most~$\gamma$.
\end{proposition}

Since $\gamma^2$ enters~\eqref{eq:LMI-perf} affinely, minimizing over $\gamma^2$ subject to the LMI constraint is a semidefinite program. Alternative performance measures, such as $\mathcal{H}_2$-performance, energy-to-peak, or other invariance specifications can be imposed analogously, with corresponding LMI conditions readily available (e.g.~\cite{veenmanAnalysis,schererLMIs,schererGahinet}). Similarly, other choices of performance output $z$ may be considered. E.g., a choice of \mbox{$z = \mathrm{col}(\tilde\eta, \rho \tilde u)$} for some $\rho > 0$ places an additional cost on the control effort, which may lead to a more balanced design. Adjusting the plant~\eqref{eq:genplant-perf} accordingly is straightforward.

For the special case of a GFC \mbox{$K(s) \equiv -\varepsilon I$}, Proposition~\ref{prop:perf-analysis} provides a principled way to select~$\varepsilon$ beyond simply maximizing it. As an illustrating example, Fig.~\ref{fig:l2gain_timescale} shows examples of performance gains obtained for a sweep over $\varepsilon$ for the example system~\eqref{eq:example-sys}, including the $\mathcal{L}_2$, $\mathcal{H}_2$ and energy-to-peak-gain. Fig.~\ref{fig:l2gain_timescale} reveals several clear sweet spots between small and large values of $\varepsilon$. The minimum of these curves identify an algorithm gain that is optimal in the sense of the desired performance specification.

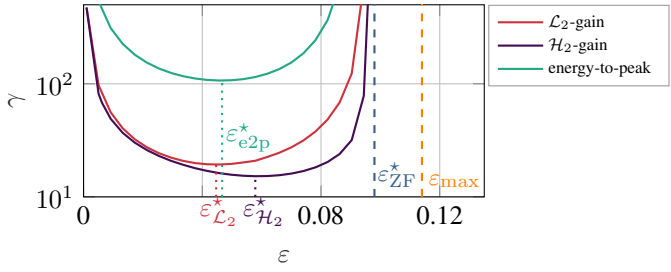
\begin{figure}
\centering
\input{plots/l2gain_timescale.tex}
\caption{Illustration of different gains on the performance channel $\hat w \to \tilde{\eta}$ for the example system \eqref{eq:example-sys} and a GFC \eqref{eq:gradient-flow}, as a function of $\varepsilon$.}
\label{fig:l2gain_timescale}
\end{figure}

\subsection{Synthesis of Dynamic Feedback Optimizers}\label{subsec:synthesis}

Rather than tuning a pre-fixed controller, the LMI formulation of Proposition~\ref{prop:perf-analysis} allows to treat the controller itself as a decision variable and synthesize a dynamic output-feedback controller that is tailored to the plant and optimality model. To this end, we seek a controller of the form
\begin{equation*}\label{eq:dynamic-K}
  K(s) \doteq
  \left[\begin{array}{c|c} A_K & B_K \\ \hline C_K & D_K \end{array}\right]
\end{equation*}
that minimizes the $\Ltwo$-gain~$\gamma$ jointly with all other decision variables.

A structural prerequisite for synthesis is that the nominal generalized plant, that is,~\eqref{eq:genplant-perf} without the oracle and performance channel, is stabilizable and detectable. Since $(A,B)$ and $(A,C)$ are stabilizable and detectable, respectively, the nominal generalized plant is stabilizable when $m_u \neq 0$, and detectable whenever the optimality model dynamics $\dot{\eta}$ is observable through $y_c$. This includes, but is not limited to~\mbox{$y_c = \eta$}. 

The synthesis problem itself follows standard procedures as in~\cite{veenman}. The controller matrices enter~\eqref{eq:LMI-perf} nonlinearly, but the standard linearizing change of variables~\cite{schererGahinet} restores convexity when either the multipliers or the controller are held fixed. An alternating optimization over the two convex problems, cf. Algorithm~\ref{alg:synthesis}, then guarantees a nonstrict decrease of $\gamma$ in every iteration.

\begin{algorithm}[H]
\caption{IQC-based controller synthesis}\label{alg:synthesis}
\begin{algorithmic}[1]
\REQUIRE Initial filters $H_u^{(0)}$, $H_y^{(0)}$, initial controller $K^{(0)}$ with performance $\gamma^{(0)}$ (e.g. from $H_\infty$-synthesis).
\WHILE{$|\gamma^{(k)} - \gamma^{(k-1)}| > \mathrm{tol}$}
  \STATE \textbf{Synthesis step.} Fix $H_u = H_u^{(k)}$, $H_y = H_y^{(k)}$, solve
    \eqref{eq:LMI-perf} over
    $(K,\,\mathcal{X},\,\gamma)$ via the change of
    variables in~\cite{schererGahinet} to minimize $\gamma$.
  \STATE \textbf{Analysis step.} Fix $K = K^{(k)}$, solve~\eqref{eq:LMI-perf}
    over $(\mathcal{X},\, H_u, H_y, \,\gamma)$ to minimize $\gamma$.
  
  \STATE $k \gets k+1$
\ENDWHILE
\end{algorithmic}
\end{algorithm}

Some features of the synthesis are worth noting. When $A$ is not Hurwitz, no pre-stabilization is needed, as the synthesis optimizes the $\Ltwo$-gain while automatically ensuring closed loop stability by a finite $\gamma$. The measurements available to the controller is determined by the output~$y_c$ of the generalized plant and can be adjusted freely. For instance, defining $y_c = \mathrm{col}(\eta, y)$ gives the controller direct access to the plant output as well, recovering the output-feedback architectures of~\cite{lawrence2018,lawrence2021}.



\subsection{Robust Design Under Model Uncertainty}\label{sec:uncertainty}

Both the analysis and the synthesis developed above require full model knowledge $(A,B,C,D,B_w,D_w)$. In contrast, related literature analyzing the gradient-flow controller requires only the steady-state sensitivity $\Piyu$. In practice, neither scenario is ideal, as perfect model knowledge is unrealistic, yet one most often has more insight than just the DC gain. 

The developed framework offers a systematic way to handle this. Any plant uncertainty description that preserves the steady-state map and admits an IQC characterization can be accommodated in an additional uncertainty channel. Proposition~\ref{prop:perf-analysis} and Algorithm~\ref{alg:synthesis} then apply without modification, and the synthesized controller is guaranteed to achieve the certified performance-gain for every plant consistent with the uncertainty description.

As an illustrative example that is particularly relevant to feedback optimization, consider the two-timescale plant
\begin{align}\label{eq:two-timescale}
  \begin{aligned}
  \bmat{\dot x_1 \\ \epsilon\,\dot x_2}
  &= \bmat{A_{11} & A_{12} \\ A_{21} & A_{22}} \bmat{x_1 \\ x_2}
  + \bmat{B_1 \\ B_2}u
  + \bmat{B_{w1} \\ B_{w2}}w \\
  y &= \bmat{C_1 & C_2}x + Du + D_w w
  \end{aligned}
\end{align}
as already studied in~\cite{yousefi}, which is motivated by well understood \emph{slow} modes~$x_1$ and residual fast modes~$x_2$. Setting \mbox{$\epsilon=0$} and resolving the algebraic constraint, a reduced model with transfer matrix $G_r(s)$ is obtained, with equal DC-gain and model mismatch at higher frequencies. Note that the case where only the steady-state sensitivity is available is covered by discarding the $x_1$-dynamics and setting \mbox{$G_r(s) \equiv \bmat{\Piyu & \Piyw}$}. The true plant $G$, mapping \mbox{$\mathrm{col}(u,w) \mapsto y$}, is then written as
\begin{equation}\label{eq:uncertainty-model}
  G(s) = G_{\mathrm{r}}(s) + W(s)\,\Delta(s),
\end{equation}
where $\|\Delta\|_\infty \leq 1$ is an unstructured LTI uncertainty and $W$ is a dynamic scalar weight whose magnitude upper bounds that of a high-pass filter with cut-off frequencies determined by the eigenvalues of $\tfrac{1}{\epsilon} A_{22}$ and feedthrough gain $\bar{\sigma}(C_2 A_{22}^{-1}B_2)$~\cite{skogestad}. The generalized plant is then set up by replacing the nominal dynamics by $y \!=\! G_r[\, \mathrm{col}(u,w)\,] \!+\! W[\, p_\Delta \,]$, where $p_\Delta$ is the input of an additional uncertainty channel \mbox{$p_\Delta = \Delta[ \,q_\Delta\, ]$}, \mbox{$q_\Delta = \mathrm{col}(u,w)$}. A valid IQC description for $\Delta$ is 
\begin{align}\label{eq:iqc-uncertainty}
  \begin{aligned}
  &\int_0^T \!
  \psi(t)^\top 
  \bmat{X \otimes I_{n_u+n_w} & 0 \\ 0 & -X \otimes I_{n_y}} 
  \psi(t) \,
  \mathrm{d}t \geq 0 \\[1.2ex]
  & \,\, \psi = \mathrm{diag}(\Psi_\Delta \otimes I_{n_u+n_w}, \Psi_\Delta \otimes I_{n_y})\, \left[\, \mathrm{col}(q_\Delta, p_{\Delta}) \, \right]
  \end{aligned}
\end{align}
where $\Psi_\Delta$ is any dynamic filter and $X = X^\top$ a free matrix matching the dimension of $\Psi_\Delta$, cf. \cite{veenmanAnalysis} for further details. The remaining analysis and synthesis proceeds as before with the modified generalized plant.

%% file: plots/l2gain_timescale.tex
\definecolor{mycolor1}{rgb}{0.00000,0.44700,0.74100}%
\definecolor{plotBlue}{rgb}{0.12157, 0.46667, 0.70588}
\definecolor{plotOrange}{rgb}{0.85098, 0.37255, 0.00784}
\definecolor{plotGreen}{rgb}{0.17255, 0.62745, 0.17255}

\definecolor{mycolor1}{RGB}{204, 51, 63} 
\definecolor{mycolor2}{RGB}{68, 12, 84}
\definecolor{mycolor3}{RGB}{53, 95, 141}
\definecolor{mycolor4}{RGB}{34, 168, 132}
\definecolor{mycolor5}{RGB}{122, 210, 81}

\begin{tikzpicture}

\begin{axis}[%
width=\columnwidth - 1.4in,
height=1.0in,
at={(0.758in,0.481in)},
scale only axis,
xmin=0,
xmax=0.135,
xtick={0, 0.04, 0.08, 0.12},
xticklabels={0, , 0.08, 0.12},
xlabel style={font=\color{white!15!black}, yshift=-1mm},
xlabel={$\varepsilon$},
ymode=log,
clip mode=individual,
ymin=10,
ymax=500,
yminorticks=true,
ylabel style={font=\color{white!15!black}, yshift=-1mm},
ylabel={$\gamma$},
axis background/.style={fill=white},
xmajorgrids,
ymajorgrids,
yminorgrids,
legend style={
    at={(1.01, 1)},
    anchor=north west,
    font=\footnotesize,
    draw=gray!80,
    fill=white,
    nodes={scale=0.8, transform shape}
},
legend cell align={left},
]

\addplot [color=mycolor1, thick, label={l1gain}]
  table[row sep=crcr]{%
0.001 475.821130060237\\
0.00506122448979592 97.2515232352423\\
0.00912244897959184 56.0105230379889\\
0.0131836734693878  40.384921869225\\
0.0172448979591837  32.3016620236275\\
0.0213061224489796  27.4697084974515\\
0.0253673469387755  24.3481040238191\\
0.0294285714285714  22.2489490815093\\
0.0334897959183674  20.8211787624234\\
0.0375510204081633  19.9056315216162\\
0.0416122448979592  19.4287297470922\\
0.0456734693877551  19.3262665805424\\
0.049734693877551   19.5685304229954\\
0.0537959183673469  20.1565132924916\\
0.0578571428571429  20.8577935508185\\
0.0619183673469388  22.5508343244327\\
0.0659795918367347  24.5771969073576\\
0.0700408163265306  27.4426408938981\\
0.0741020408163265  31.5875710340477\\
0.0781632653061225  37.8877052081439\\
0.0822244897959184  48.3319553121922\\
0.0862857142857143  68.5652454856093\\
0.0903469387755102  123.3017480035\\
0.0944081632653061  763.265555690101\\
};
\addlegendentry{$\mathcal{L}_2$-gain}

\addplot [color=mycolor2, thick]
  table[row sep=crcr]{%
0.001 470.821130060237\\
0.00506122448979592 81.551404354443\\
0.00621052631578947 67.9914858885194\\
0.00912244897959184 49.0685064339666\\
0.0131836734693878  36.8507538562777\\
0.0172448979591837  30.188646473689\\
0.0213061224489796  25.9678444893785\\
0.0253673469387755  23.0276375143649\\
0.0294285714285714  20.8488454641038\\
0.0334897959183674  19.2468673270011\\
0.0375510204081633  17.9389564799233\\
0.0416122448979592  16.9718251081718\\
0.0456734693877551  16.2332365792254\\
0.049734693877551   15.7047924558803\\
0.0537959183673469  15.3759315727931\\
0.0578571428571429  15.2315467527766\\
0.0619183673469388  15.2795245044492\\
0.0659795918367347  15.5441513277636\\
0.0700408163265306  16.0695962293535\\
0.0741020408163265  16.9258117392515\\
0.0781632653061225  18.2515378780504\\
0.0822244897959184  20.3451364086714\\
0.0862857142857143  23.9617686656506\\
0.0903469387755102  31.7968127427307\\
0.0944081632653061  78.2256265854026\\
0.096081632653061   780.2256265854026\\
};
\addlegendentry{$\mathcal{H}_2$-gain}

\addplot [color=mycolor4, thick]
  table[row sep=crcr]{%
0.001               2614.3233100024\\
0.00506122448979592 537.715385493621\\
0.00912244897959184 311.35800505256\\
0.0131836734693878  225.45688619278\\
0.0172448979591837  180.88014488787\\
0.0213061224489796  154.092113383859\\
0.0253673469387755  136.661626431771\\
0.0294285714285714  124.858867058199\\
0.0334897959183674  116.792925818141\\
0.0375510204081633  111.440889243621\\
0.0416122448979592  108.245053232502\\
0.0456734693877551  106.935764236797\\
0.049734693877551   107.456510737086\\
0.0537959183673469  109.950535118773\\
0.0578571428571429  114.799332203173\\
0.0619183673469388  122.734985039178\\
0.0659795918367347  135.094402088453\\
0.0700408163265306  154.496404867119\\
0.0741020408163265  186.368017774249\\
0.0781632653061225  243.738635024998\\
0.0822244897959184  366.498577665961\\
0.0862857142857143  747.35776043328\\
0.0903469387755102  5340.78157941779\\
};
\addlegendentry{energy-to-peak\,\,}

\addplot [color=mycolor3, dashed, thick] coordinates {(0.098, 10) (0.098, 1000)};
\node[anchor=south west] at (axis cs:0.096, 10) {\textcolor{mycolor3}{\small $\varepsilon^\star_{\mathrm{ZF}}$}};

\addplot [color=orange, dashed, thick] coordinates {(0.114, 10) (0.114, 1000)};
\node[anchor=south east] at (axis cs:0.137, 10) {\textcolor{orange}{\small $\varepsilon_{\max}$}};


\addplot [color=mycolor1, dotted, thick]
coordinates {(0.0446734693877551,10) (0.0446734693877551,19.3262665805424)};

\addplot [color=mycolor2, dotted, thick]
coordinates {(0.0578571428571429,10) (0.0578571428571429,15.2315467527766)};

\addplot [color=mycolor4, dotted, thick]
coordinates {(0.0466734693877551,10) (0.0466734693877551,106.935764236797)};

\node[
    anchor=north,
    text=mycolor1,
    yshift=0pt
] at (axis cs:0.0456734693877551,12)
{$\varepsilon^\star_{\mathcal{L}_2}$};

\node[
    anchor=north,
    text=mycolor2
] at (axis cs:0.0608571428571429,12)
{$\varepsilon^\star_{\mathcal{H}_2}$};

\node[
    anchor=north,
    text=mycolor4,
    yshift=28pt
] at (axis cs:0.0556734693877551,12)
{$\varepsilon^\star_{\mathrm{e2p}}$};

\end{axis}
\end{tikzpicture}

%% file: sections/constraints_.tex
\section{Dynamic Projected Primal Dual Controllers}\label{sec:constraints}

The framework developed so far has addressed unconstrained OSS problems. In many applications, however, constraints are imposed.
We now consider the constrained OSS problem
\begin{subequations}\label{eq:oss-constrained}
\begin{align}
    \label{eq:oss-constrained:a}
    \min_{u, y}\quad &\Phi_1(u) + \Phi_2(y) \\
    \text{s.t.}\quad &y = \Piyu u + \Piyw w \\
    & E u + F y = 0 \label{eq:linear-constraint}\\
    & u \in \mathcal{U} \label{eq:input-constraint}
\end{align}
\end{subequations}
with a closed convex input constraint set $\mathcal{U} \subset \mathbb{R}^{n_u}$ and linear equality constraints~\eqref{eq:linear-constraint} defined by \mbox{$E\in\mathbb{R}^{n_c \times n_u}$}, \mbox{$F\in\mathbb{R}^{n_c \times n_y}$}. We require the following assumption.
\begin{assumption}\label{assum:constraints}
    The matrix $N_1 \coloneq E + F \Pi_{yu}$ has full row rank, and $F$ is such that $(A,FC)$ is detectable. The feasible set of \eqref{eq:oss-constrained} is non-empty.
\end{assumption}

These regularity conditions are mild, and impose that \mbox{$n_c \leq n_u$}, i.e., there are no more equality constraints than control inputs, and that all unstable plant modes are observable through the constraint.

A previously proposed algorithm to solve the respective constrained OSS control problem is the projected primal-dual gradient flow controller (PPD-GFC) \cite{bianchinPrimalDual}
\begin{subequations}
\label{eq:primal-dual-gradient-flow}
\begin{align}
\label{eq:primal-dual-gradient-flow:a}
\dot{u} &\!=\! \varepsilon \biggl( \! - u \!+\! \mathcal{P}_{\mathcal{U}}\bigl( u \!-\! \alpha(\Phi_1(u)  \!+\! \Pi_{yu}^\top \nabla \Phi_2(y) \! + \! N_1^\top \lambda) \bigr) \!\! \biggr) \\
\label{eq:primal-dual-gradient-flow:b}
\dot{\lambda} &= \beta (E u + F y),
\end{align}
\end{subequations}
where $\alpha, \beta > 0$ are tuning parameters and $\varepsilon$ again induces a timescale separation with respect to the plant. Eq. \eqref{eq:primal-dual-gradient-flow:a} is a continuous vector field that renders $u$ smooth and $\mathcal{U}$ forward invariant. Moreover, it guarantees asymptotic satisfaction of the equality constraint \eqref{eq:linear-constraint}, if suitably tuned.
In the following, we show how to derive an optimality model for~\eqref{eq:oss-constrained}, how to embed it into a generalized plant, and how to synthesize dynamic projected primal dual (Dyn-PPD) controllers that generalize the PPD-GFC~\eqref{eq:primal-dual-gradient-flow}.

\subsection{Optimality Models}

Note that we can write~\eqref{eq:oss-constrained} equivalently in composite optimization form $\displaystyle\min  f_1(u) + f_2(N_1 u)$ with
\begin{align}
    &f_1(u) \coloneq \Phi_1(u) + \Phi_2(\Piyu u + \Piyw w) + \mathcal{I}_{\mathcal{U}}(u) \notag\\
    &f_2(z) \coloneq \mathcal{I}_{\{-F\Piyw w\}}(z),
\end{align}
where the steady-state constraint is eliminated and \eqref{eq:linear-constraint}--\eqref{eq:input-constraint} are made explicit through indicator functions. 
By Assumption~\ref{assum:constraints}, constraint qualification holds and Fenchel--Rockafellar duality can be applied. By \cite[Thm.~19.1]{bauschke}, $u^\star$ is optimal if and only if there exists $\lambda^\star \in \mathbb{R}^{n_c}$ such that
\begin{equation}\label{eq:fenchel}
    -N_1\T \lambda^\star \in \partial f_1(u^\star), \quad \lambda^\star \in \partial f_2(N_1 u^\star).
\end{equation}
Recall that \mbox{$\partial I_{\mathcal{U}} = \mathcal{N}_\mathcal{U}$}, and moreover,
\vspace{-2mm}
\begin{equation*}
    \partial \mathcal{I}_{\{-F\Piyw w\}}\!(z) \!=\! \mathcal{N}_{\{-F\Piyw w\}}\!(z) \!=\! \begin{cases}
        \mathbb{R}^{n_c} & \!\!\!\!\!, z = -F\Piyw w,\\
        \emptyset & \!\!\!\!\!, z \neq -F\Piyw w,
    \end{cases}
\end{equation*}
implying that $\lambda^\star$ is free if \mbox{$N_1 u^\star = - F\Piyw w$} holds. Note that $N_1 u^\star + F\Piyw w = E u^\star + F y^\star$. Hence, \eqref{eq:fenchel} is equivalent to
\begin{subequations}\label{eq:kkt}
\begin{align}
    \label{eq:kkt:primal}
    - N_1\T \lambda^\star &\in \nabla\Phi_1(u^\star) + \Piyu\T \nabla\Phi_2(y^\star) + \mathcal{N}_{\mathcal{U}}(u^\star)\\
    \label{eq:kkt:dual}
    0 &= E u^\star + F y^\star.
\end{align}
\end{subequations}

We use \eqref{eq:kkt} to construct an optimality model. Consider the fictitious set-valued dynamics
\begin{align}
    \dot\eta &\!=\! \alpha(\nabla\Phi_1(u) \!+\! \Piyu^\top \, \nabla\Phi_2(y) \!+\! N_1^\top \lambda + p_3), \;\, p_3 \!\in\! \mathcal{N}_{\mathcal{U}}(u) \notag \\
    \label{eq:internalmodel-constraints}
    \dot\lambda &= \beta( E u + F y ),
\end{align}
where $\alpha,\beta > 0$ are tuning parameters. At equilibrium, the first-order optimality conditions~\eqref{eq:kkt} hold. Stabilization of~\eqref{eq:internalmodel-constraints} subject to the plant dynamics is therefore sufficient for optimal steady-state operation. 

Being the subdifferential of a convex function, $\mathcal{N}_{\mathcal{U}}$ is a monotone operator and hence slope-restricted in the sector $[0,\infty]$; its input--output map therefore satisfies a Zames--Falb IQC \cite{safonov}. Although $\mathcal{N}_{\mathcal{U}}$ is set-valued and potentially unbounded, the dissipativity-based analysis of the previous sections remains valid \cite{scherer_dissipativity}, and carries over verbatim to a third oracle channel satisfying such an IQC. We accordingly extend the generalized plant with this channel and include~\eqref{eq:internalmodel-constraints}.

Consider moreover the alternative optimality model
\begin{align}
    \dot\eta &\!=\! \alpha(\nabla\Phi_1(u) \!+\! \Piyu^\top \, \nabla\Phi_2(y) \!+\! N_1^\top \lambda + p_3), \;\, p_3 \!\in\! \mathcal{N}_{\mathcal{U}}(u \!-\! \dot{\eta}) \notag \\
    \label{eq:internalmodel-constraints-deta}
    \dot\lambda &= \beta( E u + F y ),
\end{align}
where the only difference is in the evaluation of the normal cone.
Clearly, $\dot{\eta}=0$, $\dot{\lambda}=0$ also implies the first-order optimality condition~\eqref{eq:kkt}. While this choice might seem counterintuitive at first, as it makes $p_3 \in \mathcal{N}_{\mathcal{U}}(u - \dot{\eta})$ an implicit evaluation, \eqref{eq:internalmodel-constraints} and \eqref{eq:internalmodel-constraints-deta} will allow us to derive two distinct controller structures each with different properties. 

\subsection{Dyn-PPD Controller with Smooth Control Action}

Let $(u^\star, y^\star, \lambda^\star)$ denote a fixed-point that solves \eqref{eq:kkt} and $x^\star$ the corresponding optimal state of~\eqref{eq:plant}. As before, define centered coordinates $\tilde u$, $\tilde y$, $\tilde x$, $\tilde \eta$, and \mbox{$\tilde \lambda \coloneq \lambda - \lambda^\star$}. Let again \mbox{$p_1 ^\star = \nabla \Phi_1(u^\star)$}, \mbox{$p_2^\star = \nabla \Phi_2(y^\star)$} and define \mbox{$p_3^\star \coloneq - p_1^\star - \Piyu\T p_2^\star - N_1^\top \lambda^\star$}.
Considering the optimality model~\eqref{eq:internalmodel-constraints-deta},
we set up the generalized plant
\begin{align}\label{eq:genplant-constraints}
\setlength{\arraycolsep}{1.3pt}
\left[\!
\begin{array}{c}
\dot{\tilde x} \\ \dot{\tilde\eta} \\ \dot{\tilde\lambda} \\ \hline \tilde q_1 \\ \tilde q_2 \\ \tilde q_3 \\
\hdashline \tilde y_c
\end{array}\!
\right]
\!\!&=\!\!
\setlength{\arraycolsep}{1.3pt}
\left[\!
\begin{array}{ccc|ccc:c}
A & 0 & 0 & 0 & 0 & 0  & B \\
0 & 0 & \alpha N_1^\top & \alpha I_{n_u} & \alpha \Piyu^\top & \alpha I_{n_u} & N_3 \\
\beta FC & 0 & 0 & 0 & 0 & 0 & N_2 \\ \hline
0 & 0 & 0 & 0 & 0 & 0 & I_{n_u} \\
C & 0 & 0 & 0 & 0 & 0 & D \\
0 & 0 & -\alpha N_1^\top & -\alpha I_{n_u} & -\alpha \Piyu^\top & -\alpha I_{n_u} & N_4 \\ \hdashline
0 & I_{n_u} & 0 & 0 & 0 & 0 & 0
\end{array}
\right]\!\!
\setlength{\arraycolsep}{1.3pt}
\left[\!
\begin{array}{c}
\tilde x \\ \tilde\eta \\ \tilde\lambda \\ \hline \tilde p_1 \\ \tilde p_2 \\ \tilde p_3 \\ \hdashline \tilde u
\end{array}\!
\right] 
\notag \\
\tilde p_1 &= \nabla \Phi_1(\tilde q_1 + u^\star) - m_u \tilde q_1 - p_1^\star \notag \\
\tilde p_2 &= \nabla \Phi_2(\tilde q_2 + y^\star) - p_2^\star  \\
\tilde p_3 &= p_3 - p_3^\star, \quad p_3 \in \mathcal{N}_\mathcal{U}(\tilde{q}_3 + u^\star) \notag \\
\tilde{u} &= K [\, \tilde{y}_c \,], \notag
\end{align}
with $N_2 \coloneq \beta(E + FD)$, $N_3 \coloneq \alpha m_u I_{n_u}$ and \mbox{$N_4 \coloneq I_{n_u} - N_3$}. The structure of~\eqref{eq:genplant-constraints} mirrors that of the unconstrained plant~\eqref{eq:genplant-stab}, with the dual state ${\lambda}$ and the normal-cone channel $(\tilde p_3,\tilde q_3)$ in addition. For clarity, the performance channel was omitted, but can be added analogously.

It is straightforward to verify that the nominal generalized plant, i.e., \eqref{eq:genplant-constraints} without the oracle channel, is stabilizable and detectable if Assumption~\ref{assum:sys} and~\ref{assum:constraints} hold. We can proceed as before and use~\eqref{eq:genplant-constraints} to systematically synthesize a robust full-order dynamic controller $K$ with guaranteed performance level. 
Despite the set-valued and implicit $\eta$-dynamics, we can show that  $\tilde{u} \!=\! K [\,\tilde y_c\,] \!=\! K [\,\tilde \eta\,]$ admits an explicit single-valued implementation.
\begin{proposition}[Dyn-PPD controller I]\label{prop:controller-for-constrained}
    Let $K$ be a stabilizing controller for~\eqref{eq:genplant-constraints}. Then the controller has the following explicit implementation
    \begin{subequations}\label{eq:controller-for-constrained}
    \begin{align}
        \label{eq:controller-for-constrained:eta}
        \dot{\eta} &=  u \!-\! \mathcal{P}_{\mathcal{U}}( u \!-\! \alpha(\nabla \Phi_1(u)  \!+\! \Pi_{yu}^\top \nabla \Phi_2(y) \! + \! N_1^\top \lambda) ) \\
        u &= K [ \, \eta \, ].
    \end{align}
    \end{subequations}
\end{proposition}
\begin{proof}
    Shifting back from centered to original coordinates, we have that $\tilde{u} = K [\,\tilde y_c\,]$ is equivalent to $u = K [\,\eta\,]$. We use the resolvent equivalence for the projection operator \mbox{$\mathcal{P}_\mathcal{U}(x) = (\mathrm{Id} + \alpha \mathcal{N}_\mathcal{U})^{-1}(x)$}, where~$\mathrm{Id}$ is the identity operator and which holds for any \mbox{$\alpha>0$} \cite{bauschke}. Define
    \begin{equation}\label{eq:sigma}
        \sigma := u - \alpha \left( \nabla \Phi_1(u) + \Pi_{yu}^\top \nabla \Phi_2(y) + N_1^\top \lambda \right).
    \end{equation}
    From the definition of $\dot{\eta}$
    in~\eqref{eq:internalmodel-constraints-deta}, we then have
    \begin{align*}
        &  \quad u - \dot{\eta} \in \sigma
        - \alpha \mathcal{N}_{\mathcal{U}}(u - \dot{\eta})
        \\
        \Leftrightarrow & \quad (\mathrm{Id} + \alpha \mathcal{N}_\mathcal{U})(u - \dot{\eta}) \ni \sigma \\
        \Leftrightarrow & \quad u - \dot{\eta} = \mathcal{P}_\mathcal{U}(\sigma).
    \end{align*}
    Rearranging the latter yields \eqref{eq:controller-for-constrained:eta}.
\end{proof}

Note that the PPD-GFC \eqref{eq:primal-dual-gradient-flow} is recovered by \eqref{eq:genplant-constraints} and \eqref{eq:controller-for-constrained} with a P-controller $K(s) \equiv - \varepsilon I_{n_u}$.
We can therefore make use of the same results presented in Sections~\ref{sec:stability} and \ref{sec:performance} for a principled tuning of $\varepsilon$. Eq. \eqref{eq:controller-for-constrained} therefore serves as a \emph{dynamic} generalization of the projected primal dual method.

We emphasize that the right hand side of \eqref{eq:controller-for-constrained:eta} is a Lipschitz continuous vector field, cf. \cite{bianchinPrimalDual,gao}.
Eq.~\eqref{eq:controller-for-constrained} therefore results in a \emph{smooth} controller.
However, unlike the \mbox{PPD-GFC}~\eqref{eq:primal-dual-gradient-flow}, \eqref{eq:controller-for-constrained} will generally only lead to asymptotic satisfaction of the input constraint $u \in \mathcal{U}$. In the following, we therefore derive an additional control structure for hard input constraint satisfaction.

\subsection{Dyn-PPD Controller with Hard Constraint Satisfaction}

Let us consider a projected controller $u = \mathcal{P}_\mathcal{U}(v)$ with \mbox{$v = K[ \,\eta\, ]$}, for some auxiliary signal $v$. By construction, \mbox{$u(t) \in \mathcal{U}$} for all $t \geq 0$, regardless of $K$ and $\eta$. Note that we can write equivalently 
\begin{align*}
&u = \mathcal{P}_\mathcal{U}(v) = \arg \min_u \tfrac{1}{2 \alpha} \| u - v \|^2 + \mathcal{I}_{\mathcal{U}}(u) \\ & \Leftrightarrow \qquad 0 \in \tfrac{1}{\alpha}(u - v) + \mathcal{N}_{\mathcal{U}}(u),
\end{align*}
so that $u = v - \alpha p_3$, for some $p_3 \in \mathcal{N}_\mathcal{U}(u)$. Using the optimality model \eqref{eq:internalmodel-constraints}, and defining \mbox{$\tilde{v} = v - v^\star$}, \mbox{$v^\star \coloneq u^\star + \alpha p_3^\star$}, we can define the following alternative generalized plant
\begin{align}\label{eq:genplant-constraints:hard}
\setlength{\arraycolsep}{2pt}
\left[\!
\begin{array}{c}
\dot{\tilde x} \\ \dot{\tilde\eta} \\ \dot{\tilde\lambda} \\ \hline \tilde q_1 \\ \tilde q_2 \\ \tilde q_3 \\
\hdashline \tilde y_c
\end{array}\!
\right]
\!\!&=\!\!
\setlength{\arraycolsep}{2pt}
\left[\!
\begin{array}{ccc|ccc:c}
A & 0 & 0 & 0 & 0 & -\alpha B  & B \\
0 & 0 & \alpha N_1^\top & \alpha I_{n_u} & \alpha \Piyu^\top & \alpha N_4 & N_3 \\
\beta FC & 0 & 0 & 0 & 0 & -\alpha N_2 & N_2 \\ \hline
0 & 0 & 0 & 0 & 0 & -\alpha I_{n_u} & I_{n_u} \\
C & 0 & 0 & 0 & 0 & -\alpha D & D \\
0 & 0 & 0 & 0 & 0 & -\alpha I_{n_u} & I_{n_u} \\ \hdashline
0 & I_{n_u} & 0 & 0 & 0 & 0 & 0
\end{array}
\right]\!\!\!
\setlength{\arraycolsep}{2pt}
\left[\!
\begin{array}{c}
\tilde x \\ \tilde\eta \\ \tilde\lambda \\ \hline \tilde p_1 \\ \tilde p_2 \\ \tilde p_3 \\ \hdashline \tilde v
\end{array}\!
\right]\!
\notag \\
\tilde p_1 &= \nabla \Phi_1(\tilde q_1 + u^\star) - m_u \tilde q_1 - p_1^\star \notag \\
\tilde p_2 &= \nabla \Phi_2(\tilde q_2 + y^\star) - p_2^\star  \\
\tilde p_3 &= p_3 - p_3^\star, \quad p_3 \in \mathcal{N}_\mathcal{U}(\tilde{q}_3 + u^\star) \notag \\
\tilde{v} &= K [\, \tilde{y}_c \,]. \notag
\end{align}

Stabilizability and detectability holds as for~\eqref{eq:genplant-constraints}, and we can therefore synthesize a dynamic controller $K$ for \eqref{eq:genplant-constraints:hard}.
We verify that $\tilde{v} = K [\, \tilde{y}_c \,]$ admits again an explicit implementation.

\begin{proposition}[Dyn-PPD controller II]\label{prop:controller-for-constrained:hard}
    Let $K$ be a stabilizing controller for~\eqref{eq:genplant-constraints:hard}. Define the filter 
    \begin{equation}
        R(s) \coloneqq -K(s)\bigl(sI - K(s)\bigr)^{-1}.
    \end{equation}
    Then, the controller has the following explicit implementation
    \vspace{-5mm}
    \begin{subequations}\label{eq:controller-hard}
    \begin{align}
        v &= R\bigl[\, {u - \alpha(\nabla\Phi_1(u)
            + \Piyu^\top\nabla\Phi_2(y) + N_1^\top\lambda)} \,\bigr]
            \label{eq:controller-hard:v} \\
        u &= \mathcal{P}_{\mathcal{U}}(v). \label{eq:controller-hard:u}
    \end{align}
    \end{subequations}
    In particular, $R(s)$ is strictly proper and no algebraic loop arises in \eqref{eq:controller-hard}.
\end{proposition}

\begin{proof}
    As before, $\tilde{v} = K [\,\tilde y_c\,]$ is equivalent to $v = K [\,\eta\,]$. Consider \eqref{eq:sigma} and let $\hat{\eta}, \hat{\sigma}, \hat{v}$ be the Laplace transforms of $\eta, \sigma, v$, respectively. Let $p_3 \in \mathcal{N}_{\mathcal{U}}(u)$ and insert $p_3 = \tfrac{1}{\alpha} (v - u)$ into the optimality model~\eqref{eq:internalmodel-constraints}, giving \mbox{$\dot{\eta} = - \sigma + v$}. In Laplace domain, using $\hat{v} = K(s) \hat{\eta}$, this gives
    \begin{align*}
        s\hat\eta = -\hat\sigma + K(s) \hat\eta \quad \Leftrightarrow \quad \hat \eta = -(sI - K(s))^{-1} \hat \sigma.
    \end{align*}
    Thus $\hat{v} = K(s)\hat\eta = R(s)\hat\sigma$,
    yielding~\eqref{eq:controller-hard:v}.
    Strict properness of $R(s)$ follows from properness of $K(s)$, which gives
    \mbox{$\lim\limits_{s \to \infty }(sI - K(s))^{-1} = 0$} and thus, also $\lim\limits_{s \to \infty }R(s) = 0$.
\end{proof}
\vspace{1mm}

Eq. \eqref{eq:controller-hard} is generally not smooth, but fulfils hard input constraint satisfaction by construction.
Inspecting \eqref{eq:controller-for-constrained} and \eqref{eq:controller-hard} more closely reveals that both controllers can be written as
\begin{align*}
    \eqref{eq:controller-for-constrained}: \qquad v &= \mathcal{P}_\mathcal{U}(\sigma), & u &=  R [ \, v \, ] \\
    \eqref{eq:controller-hard}: \qquad v &= R [ \,\sigma \, ], & u &= \mathcal{P}_\mathcal{U} (v),
\end{align*}
demonstrating that the difference between them is the placement of the filter $R(s)$ either before or after the projection. Both controllers trade off smoothness vs. hard constraint satisfaction. For the static gain $K(s) \equiv -\varepsilon I_{n_u}$, both reduce to a first-order low-pass $R(s) = \frac{\varepsilon}{s+\varepsilon} I$. 



We note that these were just two illustrations on how to design dynamic feedback controllers solving the constrained problem \eqref{eq:oss-constrained}. Other optimality models can also can be also considered (cf. \cite{lawrence2021,simpsonporco2022}), from which other controller structures can be derived. We leave such developments for future work.

%% file: sections/numerics.tex
\section{Numerical Examples}
\label{sec:numerics}

We showcase the application of the framework developed in the preceding sections. 
All implementations are publicly available\footnote{Code available at \mbox{https://github.com/col-tasas/2026-iqc-feedback-opt}}. 
For implementation of the IQC analysis and synthesis, we use the open-source toolbox IQClab~\cite{IQClab}, built on the Matlab Robust Control Toolbox.

\subsection{Academic Example}

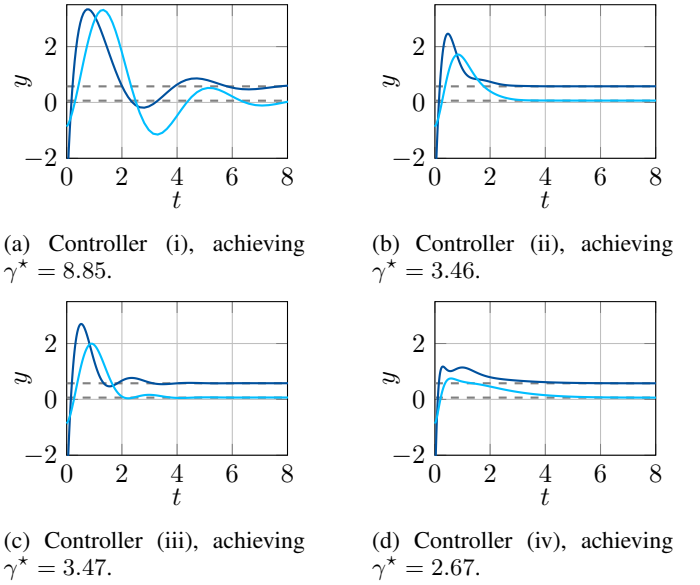
\begin{figure}[t]
  \centering
  \begin{subfigure}[t]{0.45\columnwidth}
    \centering
    \input{plots/transient_static.tex}
    \caption{Controller (i), achieving $\gamma^\star = 8.85$.}
  \end{subfigure}
  \hfill
  \begin{subfigure}[t]{0.45\columnwidth}
    \centering
    \input{plots/transient_dyn_eta.tex}
    \caption{Controller (ii),  achieving \mbox{$\gamma^\star = 3.46$}.}
  \end{subfigure}

  \vspace{2mm}

  \begin{subfigure}[t]{0.45\columnwidth}
    \centering
    \input{plots/transient_dyn_eta_reduced.tex}
    \caption{Controller (iii), achieving \mbox{$\gamma^\star = 3.47$}.}
  \end{subfigure}
  \hfill
  \begin{subfigure}[t]{0.45\columnwidth}
    \centering
    \input{plots/transient_dyn_y_eta.tex}
    \caption{Controller (iv),  achieving $\gamma^\star = 2.67$.}
  \end{subfigure}
  \caption{Closed loop trajectories of system \eqref{eq:example-sys} and the controller (i)-(iv). Gray dashed lines mark the optimal $y^\star$ solving the OSS problem. The achieved $\Ltwo$-gain $\gamma^\star$ is reported in the caption.}
  \label{fig:synthesis-transient}
  \vspace{-4mm}
\end{figure}

We consider the following example system 
\begin{align}\label{eq:example-sys}
    \begin{aligned}
    A &= \bmat{-1 & -4 & -1 & 3 \\ 1 & -4 & -1 & -3 \\ -1 & 4 & -1 & -9 \\ 0 & 0 & 0 & -4}, & 
    B &= \bmat{0 \\ 1 \\ 0 \\ 1},\,\, &
    B_w &= \bmat{1 \\ 0 \\ 0 \\ 0}, \\
    C &= \bmat{1 & -1 & 0 & -4 \\ 1 & 0 & 2 & 0}, &
    D &= \bmat{0 \\ 0}, & D_w &= \bmat{0 \\ 0},
    \end{aligned}
\end{align}
which was introduced in \cite{lawrence2018}.

\subsubsection{Unconstrained OSS problem}
We first seek to solve the unconstrained OSS problem~\eqref{eq:oss}. We compare four different IQC-based feedback optimization controller designs. The controller structures we compare are:
\begin{itemize}
  \item[(i)] A GFC~\eqref{eq:gradient-flow}, where $\varepsilon$ is tuned with \eqref{eq:LMI-perf} to achieve a minimal $\Ltwo$-gain.
  \item[(ii)] A full-order dynamic controller $K(s)$ synthesized for the generalized plant \eqref{eq:genplant-perf}.
  \item[(iii)] A full-order dynamic controller $K(s)$ synthesized for a reduced plant $G_r(s)$. The reduced plant is obtained by a DC-gain preserving model-order reduction to order 2, and an additional LTI dynamic uncertainty is modelled within the generalized plant. 
  \item[(iv)] A full-order dynamic controller $K(s)$ synthesized for the generalized plant \eqref{eq:genplant-perf}, but with an enlarged measurement $y_c = \mathrm{col}(y, \eta)$.
\end{itemize}
All controllers are synthesized to minimize the $\Ltwo$-gain from $\hat w$ to $z=\mathrm{col}(\eta, 10 u)$. The synthesis is performed without exact knowledge of $\Phi_1$ and $\Phi_2$; only the convexity and smoothness moduli $m_u, L_u, L_y$ are assumed to be known. All controllers synthesized are therefore guaranteed to stabilize the closed loop for $\emph{any}$ strongly convex and smooth functions in the respective sectors. For simulation, we choose quadratic cost functions $\Phi_1(u) = \tfrac{1}{2}u^\top Q_u u$, $\Phi_2(y) = \tfrac{1}{2}y^\top Q_y y$ with $Q_u, Q_y$ chosen to have all eigenvalues in the intervals $[m_u, L_u]$ and $[0, L_y]$, respectively.


The resulting output trajectories are displayed in Fig.~\ref{fig:synthesis-transient}. We report the true $\Ltwo$-gain in the captions, computed as \mbox{$\gamma^\star = \| \mathrm{diag}(Q_u, Q_y) \star P_{\mathrm{perf}} \star K \|_\infty$}.
Observe that all plant outputs converge to the optimal $y^\star$ that solves the OSS problem. Moreover, compared to the GFC (i), all dynamic controllers offer a faster settling time, reduced overshoot and smaller $\Ltwo$-gains~$\gamma^\star$. The best performance is obtained by the full-order dynamic controller utilizing the extended measurement vector (iv), while the dynamic controller synthesized with the uncertain reduced-order system (iii) introduces only a marginal performance degradation relative to its full-knowledge counterpart (ii). 

\subsubsection{Constrained OSS problem}

We now consider the constrained OSS problem~\eqref{eq:oss-constrained}. We define the input constraint set $\mathcal{U} = \{ u \in \mathbb{R}\; | \; 0.2 \leq u \leq 1 \}$ and impose the linear equality constraint $E = 1$, \mbox{$F = \bmat{0 & -1}$}, enforcing $u = y_2$ at steady state. We synthesize both Dyn-PPD controllers~\eqref{eq:controller-for-constrained} and~\eqref{eq:controller-hard} from Section~\ref{sec:constraints}, with the $\Ltwo$-gain from $\hat w$ to $z=\mathrm{col}(\eta, \lambda, 10 u)$ as performance criterion. We choose $\alpha = 0.05$, $\beta=10$ in the optimality models.

The resulting trajectories of $u$ and $y$ are displayed in Fig.~\ref{fig:synthesis-constrained}. In both cases, the trajectories converge to the linear constraint manifold $u=y_2$. The Dyn-PPD controller~I leads to smooth control inputs, as indicated in the zoom within Fig.~\ref{fig:synthesis-constrained:a}, but leaves the input constraint set during the initial transient. In contrast, controller~II guarantees hard input constraint satisfaction by design, but leads to non-smooth control inputs, see Fig~\ref{fig:synthesis-constrained:b}. 
\begin{figure}[t]
  \centering
  \begin{subfigure}[t]{0.45\columnwidth}
    \centering
    \input{plots/transient_con_soft.tex}
    \caption{Dyn-PPD controller I \eqref{eq:controller-for-constrained} with smooth control inputs.}
    \label{fig:synthesis-constrained:a}
  \end{subfigure}
  \hfill
  \begin{subfigure}[t]{0.45\columnwidth}
    \centering
    \input{plots/transient_con_hard.tex}
    \caption{Dyn-PPD controller II \eqref{eq:controller-hard} with hard input constraint satisfaction.}
    \label{fig:synthesis-constrained:b}
  \end{subfigure}
  \caption{Closed loop trajectories of system \eqref{eq:example-sys} and synthesized Dyn-PPD controllers. Gray dashed lines mark the optimal \mbox{$(u^\star$, $y^\star)$} solving the constrained OSS problem.}
  \label{fig:synthesis-constrained}
\end{figure}
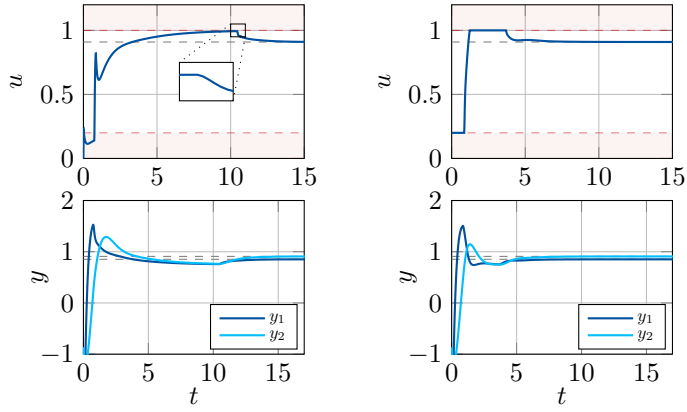
\vspace{-3mm}

\subsection{Satellite Rendezvous}

We adopt the satellite rendezvous example from \cite{chuySatellite}. 
The relative motion between a chaser and target satellite is described by the \mbox{Clohessy--Wiltshire} equations
\begin{align}
  \begin{aligned}\label{eq:satellite}
  \ddot{p}_x &= 3 \omega^2 p_x + 2 \omega \dot{p}_y + \tfrac{1}{m_c} F_x \\
  \ddot{p}_y &= - 2 \omega \dot{p}_x + \tfrac{1}{m_c} F_y + w \\
  \ddot{p}_z &= - \omega^2 p_z + \tfrac{1}{m_c} F_z,
  \end{aligned}
\end{align}
where $(p_x,p_y,p_z)$ is the relative position, $(F_x,F_y,F_z)$ are thruster forces, and $m_c, \omega$ are chaser mass and orbital velocity, respectively chosen as in~\cite{chuySatellite}. Moreover, $w$ models an unknown along-track disturbance. The control input is the thrust \mbox{$u=\mathrm{col}(F_x, F_y, F_z) \in \mathbb{R}^3$} and the measured output is the relative position \mbox{$y = \mathrm{col}(p_x, p_y, p_z) \in \mathbb{R}^3$}. The dynamics~\eqref{eq:satellite} can be written as~\eqref{eq:plant} with order $n_x=6$, where $A$ has all its eigenvalues on the imaginary axis.

We define the OSS problem with \mbox{$\Phi_1(u) = 0.01 \| u \|^2$} and \mbox{$\Phi_2(y) = \tfrac{1}{2} \| y - y_{\mathrm{ref}} \|^2$}, for a given reference position \mbox{$y_\mathrm{ref} \in \mathbb{R}^3$}. We impose an input constraint set \mbox{$\mathcal{U} = \{ u \in \mathbb{R}^3 \; | \; \max\left( |u_1|, |u_2|, |u_3| \right) \leq u_{\max} \}$} for some maximum force $u_{\max}$. We synthesize the \mbox{Dyn-PPD II} controller~\eqref{eq:controller-hard}, where the optimality model \eqref{eq:internalmodel-constraints}, the generalized plant \eqref{eq:genplant-constraints:hard}, and the controller implementation~\eqref{eq:controller-hard:v} is minimally adjusted in line with Appendix~\ref{appendix:generalization}. We minimize the $\Ltwo$-gain from $w$ to \mbox{$z=\mathrm{col}(y,10 \eta,u)$, with $\alpha=10^{-4}$}.

The resulting closed-loop trajectories are shown in Fig.~\ref{fig:satellite} for initial position \mbox{$y(0) = \mathrm{col}(50,-140,270)$}, reference \mbox{$y_{\mathrm{ref}} = \mathrm{col}(100,100,100)$}, and actuator limit $u_{\mathrm{max}} = 20$. At $t=1500$, the disturbance steps from $w=0$ to $w=1$. We observe that both input and output converge to the optimal $(u^\star,y^\star)$, and after the disturbance step the controller re-optimizes quickly. Moreover, the input constraints are satisfied at all times. 

\begin{figure}[t]
  \centering
  \input{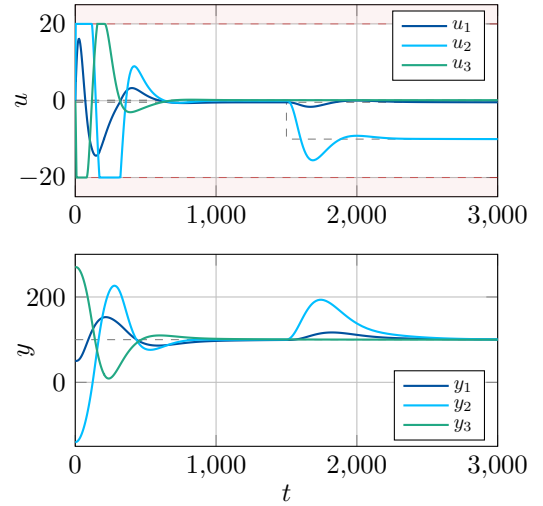}
  \caption{Closed loop trajectories of the satellite rendezvous dynamics \eqref{eq:satellite} controlled with a synthesized Dyn-PPD controller. Gray dashed lines mark the optimal \mbox{$(u^\star$, $y^\star)$} solving the constrained OSS problem.}
  \label{fig:satellite}
\end{figure}
\vspace{-3mm}

%% file: plots/transient_static.tex
%
%

\definecolor{mycolor2}{rgb}{0.85000,0.32500,0.09800}%
\definecolor{mycolor3}{rgb}{0.85000,0.42500,0.09800}%

\definecolor{mycolor1}{RGB}{0,81,158}%
\definecolor{mycolor2}{RGB}{0, 190, 255}
\definecolor{usGray}{RGB}{139,143,148}%
\begin{tikzpicture}

\begin{axis}[%
width=1.15in,
height=0.8in,
at={(0.758in,0.481in)},
scale only axis,
xmin=0,
xmax=8,
ymin=-2,
ymax=3.5,
axis background/.style={fill=white},
xmajorgrids,
ymajorgrids,
ylabel style = {yshift=-3.5mm},
ylabel={$y$},
xlabel={$t$},
xlabel style = {yshift=1.5mm},
]

\addplot [color=gray, dashed, thick]
  table[row sep=crcr]{%
0	0.576082308650491\\
20	0.576082308650491\\
};\label{tik:gray}
\addplot [color=gray, dashed, thick]
  table[row sep=crcr]{%
0	0.0627947023452091\\
20	0.0627947023452091\\
};

\addplot [color=mycolor1, thick]
  table[row sep=crcr]{%
0	-3.48914335047704\\
3.78079857924201e-06	-3.48904153555318\\
2.2684791475452e-05	-3.48853247988917\\
0.000117204755956502	-3.48598767539267\\
0.000589804578361753	-3.47327549129846\\
0.00295280369038801	-3.41000963275375\\
0.0147677992505193	-3.10094608865347\\
0.0413633872623466	-2.44773717717713\\
0.0742917511511331	-1.71479997880347\\
0.11237614860424	-0.961851564167096\\
0.15477140400508	-0.230376107230957\\
0.198311419245539	0.417369139808325\\
0.243189956739887	0.988475396483819\\
0.289463134393057	1.48723213110467\\
0.337198426326258	1.91780417491493\\
0.386469562228858	2.28413960877154\\
0.437358415343939	2.5899769256465\\
0.490352448441863	2.840481004474\\
0.546299456015065	3.03993191860064\\
0.605339471917114	3.1882529972485\\
0.667654440010817	3.28581789002488\\
0.733447867940254	3.33329383086689\\
0.802964943807988	3.33161311291477\\
0.876517153206322	3.28189610595556\\
0.954518313794909	3.18534019408117\\
1.03754139608873	3.04306869002475\\
1.12641636914043	2.85590868373376\\
1.22241836893893	2.62400133967538\\
1.32768981366033	2.34593332265186\\
1.44382668203155	2.02348732525755\\
1.55417563181538	1.71323670594833\\
1.66187832460369	1.41594238247153\\
1.76915139189966	1.13267235416652\\
1.8772466317369	0.866426162353266\\
1.98711582675171	0.62081322834588\\
2.10133708899577	0.396630122304362\\
2.22245785170716	0.197204393037171\\
2.35196704576577	0.0300264944435735\\
2.49218149518855	-0.0966285278492118\\
2.64724475103649	-0.173108410319686\\
2.81045174226079	-0.187676245753346\\
2.97200850550562	-0.145167134462254\\
3.13469125538284	-0.0572621493405592\\
3.2992132172035	0.0640197962484059\\
3.46574440554733	0.205692475220956\\
3.63154304325345	0.352180349405578\\
3.79734691784755	0.492342130667334\\
3.96590237440201	0.618623968984279\\
4.13993785224592	0.724394311386169\\
4.32299743749447	0.803765860537824\\
4.52105095006489	0.851171879551459\\
4.71890431313873	0.861007407931709\\
4.91386150386258	0.839704466871603\\
5.11074621552524	0.794872482306574\\
5.30943062540315	0.734928829712086\\
5.51277391696738	0.667693626896879\\
5.7143056651176	0.603519603692425\\
5.91879176741466	0.547599253588991\\
6.13096766587537	0.504116385786852\\
6.35717806140022	0.476569445108471\\
6.60239817023005	0.468188840208489\\
6.8397411860503	0.477970712075321\\
7.07795942119103	0.500185724760255\\
7.3192631342338	0.529237858789867\\
7.56519404424937	0.559597997920649\\
7.81081600967929	0.585774885108607\\
8.06458608681818	0.605100266166718\\
8.33636063164688	0.615542145709161\\
8.63320695500445	0.615606473263261\\
8.92067697569078	0.607251249374014\\
9.21244296644705	0.59429339750297\\
9.51056907475738	0.580483040878203\\
9.80827199843748	0.569328028264273\\
10.1177762841238	0.562439578398569\\
10.4577254037584	0.560759876132546\\
10.8102610650617	0.564091265691471\\
11.1614659534551	0.570011686419788\\
11.5214543407876	0.576099299027406\\
11.8832553358441	0.580322586220046\\
12.2665157757015	0.581913303884681\\
12.6665157757015	0.580866013153158\\
13.0665157757015	0.578368479035147\\
13.4665157757015	0.575847943962553\\
13.8665157757015	0.574261416174166\\
14.2665157757015	0.573893021138583\\
14.6665157757015	0.574472772537896\\
15.0665157757015	0.575464181867481\\
15.4665157757015	0.576355734270881\\
15.8665157757015	0.576845156943508\\
16.2665157757015	0.576883052835287\\
16.6665157757015	0.57660665477551\\
17.0665157757015	0.576225518703702\\
17.4665157757015	0.575918332870898\\
17.8665157757015	0.575775977147047\\
18.2665157757015	0.575796758827996\\
18.6665157757015	0.575918111944167\\
19.0665157757015	0.576060586432542\\
19.4665157757014	0.576163436114908\\
19.8665157757014	0.576201221181281\\
20	0.57619995286626\\
};
\addplot [color=mycolor2, thick]
  table[row sep=crcr]{%
0	-0.862605362039618\\
3.78079857924201e-06	-0.862600844636498\\
2.2684791475452e-05	-0.862578254108283\\
0.000117204755956502	-0.862465213663721\\
0.000589804578361753	-0.861897817833886\\
0.00295280369038801	-0.859006182982677\\
0.0147677992505193	-0.843204335175919\\
0.0413633872623466	-0.799810455943624\\
0.0742917511511331	-0.732232646866309\\
0.11237614860424	-0.636951971553671\\
0.15477140400508	-0.511901415652399\\
0.198311419245539	-0.365477650831518\\
0.243189956739887	-0.198341445835198\\
0.289463134393057	-0.0116875672349437\\
0.337198426326258	0.19311261744321\\
0.386469562228858	0.414432783951765\\
0.437358415343939	0.650337469585018\\
0.490352448441863	0.900404118154409\\
0.546299456015065	1.16541931938817\\
0.605339471917114	1.44193114356014\\
0.667654440010817	1.72554725632081\\
0.733447867940254	2.01070275481789\\
0.802964943807988	2.2906228174894\\
0.876517153206322	2.55729086993575\\
0.954518313794909	2.80141390154871\\
1.03754139608873	3.0123468099795\\
1.12641636914043	3.17789690968016\\
1.22241836893893	3.28382726922348\\
1.32768981366033	3.31255909475748\\
1.44382668203155	3.24200423919987\\
1.55417563181538	3.08352082234365\\
1.66187832460369	2.85358822583986\\
1.76915139189966	2.56313722226774\\
1.8772466317369	2.22236060355076\\
1.98711582675171	1.84164736978744\\
2.10133708899577	1.42568705100152\\
2.22245785170716	0.980481306690445\\
2.35196704576577	0.519910330173333\\
2.49218149518855	0.060952149392536\\
2.64724475103649	-0.376030309925247\\
2.81045174226079	-0.736874194125591\\
2.97200850550562	-0.983503987811808\\
3.13469125538284	-1.11950876859982\\
3.2992132172035	-1.14970659126383\\
3.46574440554733	-1.08487277197608\\
3.63154304325345	-0.944485820354392\\
3.79734691784755	-0.750628732686026\\
3.96590237440201	-0.522222057912574\\
4.13993785224592	-0.277086178806265\\
4.32299743749447	-0.0325220153961272\\
4.52105095006489	0.194363769634682\\
4.71890431313873	0.365764121252291\\
4.91386150386258	0.472959173938969\\
5.11074621552524	0.519191201770072\\
5.30943062540315	0.509330441714436\\
5.51277391696738	0.452417437574876\\
5.7143056651176	0.364106864888811\\
5.91879176741466	0.257977957546997\\
6.13096766587537	0.146190908172492\\
6.35717806140022	0.0402950096134234\\
6.60239817023005	-0.0461696640214324\\
6.8397411860503	-0.0953940566827418\\
7.07795942119103	-0.110730960889753\\
7.3192631342338	-0.0967019884544902\\
7.56519404424937	-0.0607123054529568\\
7.81081600967929	-0.0133636412496765\\
8.06458608681818	0.0369649067258147\\
8.33636063164688	0.0827220340124644\\
8.63320695500445	0.115608279340611\\
8.92067697569078	0.128069557718839\\
9.21244296644705	0.123524571855041\\
9.51056907475738	0.10664854639753\\
9.80827199843748	0.0844283476731897\\
10.1177762841238	0.0624602291576259\\
10.4577254037584	0.0453933472098074\\
10.8102610650617	0.0381853510378471\\
11.1614659534551	0.0405654572473235\\
11.5214543407876	0.0491419815857412\\
11.8832553358441	0.0594049373558407\\
12.2665157757015	0.0679336781802204\\
12.6665157757015	0.0719650140007513\\
13.0665157757015	0.0711899288371765\\
13.4665157757015	0.0675130836015034\\
13.8665157757015	0.0632929178101058\\
14.2665157757015	0.060301147079814\\
14.6665157757015	0.0592523642815301\\
15.0665157757015	0.0598896146336116\\
15.4665157757015	0.0614071637451588\\
15.8665157757015	0.0629380928952781\\
16.2665157757015	0.0639013300128681\\
16.6665157757015	0.0641245815830338\\
17.0665157757015	0.0637720607391818\\
17.4665157757015	0.063169729455935\\
17.8665157757015	0.0626282606991522\\
18.2665157757015	0.0623311620358491\\
18.6665157757015	0.062308331608329\\
19.0665157757015	0.0624763399183381\\
19.4665157757014	0.0627078772147422\\
19.8665157757014	0.0628944248027761\\
20	0.0629353795131129\\
};

\end{axis}
\end{tikzpicture}%

%% file: plots/transient_dyn_eta.tex
%
%
\definecolor{mycolor1}{RGB}{0,81,158}%
\definecolor{mycolor2}{RGB}{0, 190, 255}
\definecolor{usGray}{RGB}{139,143,148}%
\begin{tikzpicture}

\begin{axis}[%
width=1.15in,
height=0.8in,
at={(0.758in,0.481in)},
scale only axis,
xmin=0,
xmax=8,
ymin=-2,
ymax=3.5,
axis background/.style={fill=white},
xmajorgrids,
ymajorgrids,
ylabel style = {yshift=-3.5mm},
ylabel={$y$},
xlabel={$t$},
xlabel style = {yshift=1.5mm},
]

\addplot [color=gray, dashed, thick]
  table[row sep=crcr]{%
0	0.576082308650491\\
20	0.576082308650491\\
};
\addplot [color=gray, dashed, thick]
  table[row sep=crcr]{%
0	0.0627947023452091\\
20	0.0627947023452091\\
};
\addplot [color=mycolor1, thick]
  table[row sep=crcr]{%
0	-3.48914335047704\\
1.90494919136049e-07	-3.4891382205196\\
3.80989838272098e-07	-3.48913309056531\\
5.71484757408146e-07	-3.48912796061423\\
1.40362506689516e-06	-3.4891055514475\\
2.23576537638218e-06	-3.48908314234077\\
3.06790568586919e-06	-3.4890607332949\\
3.90004599535621e-06	-3.48903832431016\\
6.42033919822868e-06	-3.48897045486378\\
8.94063240110115e-06	-3.48890258597585\\
1.14609256039736e-05	-3.48883471764434\\
1.39812188068461e-05	-3.48876684986713\\
2.31122755165313e-05	-3.48852096859265\\
3.22433322262165e-05	-3.48827509446702\\
4.13743889359017e-05	-3.48802922739036\\
5.05054456455869e-05	-3.48778336726338\\
5.96365023552721e-05	-3.48753751398734\\
8.87080004601316e-05	-3.48675480975829\\
0.000117779498564991	-3.48597217083059\\
0.000146850996669851	-3.48518959408696\\
0.00017592249477471	-3.48440707644747\\
0.00020499399287957	-3.48362461486846\\
0.000257350826916829	-3.4822155585062\\
0.000309707660954089	-3.48080665694453\\
0.000362064494991349	-3.47939789325574\\
0.000414421329028608	-3.47798925086544\\
0.000466778163065868	-3.47658071354659\\
0.000519134997103128	-3.47517226541476\\
0.00062951347348196	-3.47220320332807\\
0.000739891949860791	-3.46923432805506\\
0.000850270426239623	-3.46626550504113\\
0.000960648902618454	-3.46329660582866\\
0.00107102737899729	-3.46032750783765\\
0.00118140585537612	-3.45735809416292\\
0.00134948275112732	-3.45283561532057\\
0.00151755964687852	-3.44831177739673\\
0.00168563654262973	-3.44378623716304\\
0.00185371343838093	-3.4392586768374\\
0.00202179033413213	-3.43472880272141\\
0.00218986722988334	-3.4301963438917\\
0.00235794412563454	-3.42566105095238\\
0.00265480285016022	-3.41764316685815\\
0.0029516615746859	-3.40961461952201\\
0.00324852029921158	-3.40157445979847\\
0.00354537902373726	-3.39352189154711\\
0.00384223774826295	-3.38545625683264\\
0.00413909647278863	-3.37737702234174\\
0.00443595519731431	-3.36928376691018\\
0.00481404091763571	-3.35895522445322\\
0.00519212663795712	-3.34860297420967\\
0.00557021235827852	-3.3382267707772\\
0.00594829807859993	-3.32782654820978\\
0.00632638379892133	-3.31740239578305\\
0.00670446951924273	-3.30695453632281\\
0.00708255523956414	-3.29648330680996\\
0.00755285141043324	-3.28342632170042\\
0.00802314758130234	-3.2703348637362\\
0.00849344375217143	-3.25721009926611\\
0.00896373992304053	-3.24405331873061\\
0.00943403609390963	-3.23086590870307\\
0.00990433226477873	-3.21764932780546\\
0.0103746284356478	-3.20440508585709\\
0.0109599759457901	-3.18788452134206\\
0.0115453234559323	-3.17132651295408\\
0.0121306709660745	-3.15473408389117\\
0.0127160184762168	-3.13811023809537\\
0.013301365986359	-3.12145793844372\\
0.0138867134965012	-3.10478008939355\\
0.0144720610066435	-3.08807952297232\\
0.0152109749913156	-3.06696936787682\\
0.0159498889759876	-3.04583270941736\\
0.0166888029606597	-3.0246746265546\\
0.0174277169453318	-3.00349994545301\\
0.0181666309300039	-2.98231323809372\\
0.018905544914676	-2.96111882424886\\
0.0196444588993481	-2.93992077526105\\
0.0206121406972682	-2.91216066924913\\
0.0215798224951882	-2.88440886936231\\
0.0225475042931083	-2.85667272140892\\
0.0235151860910284	-2.8289589588595\\
0.0244828678889485	-2.80127374534549\\
0.0254505496868686	-2.77362271606843\\
0.0264182314847887	-2.74601101636776\\
0.0278471252801151	-2.7053209264141\\
0.0292760190754416	-2.6647402280356\\
0.0307049128707681	-2.62428064038902\\
0.0321338066660945	-2.5839523518332\\
0.033562700461421	-2.54376421316156\\
0.0349915942567475	-2.50372390158885\\
0.0365840167244509	-2.45928346172088\\
0.0381764391921544	-2.41504273573561\\
0.0397688616598578	-2.37100844359413\\
0.0413612841275613	-2.32718628065568\\
0.0429537065952647	-2.28358105529498\\
0.0445461290629681	-2.24019681780254\\
0.0461385515306716	-2.19703697421605\\
0.0482131130874922	-2.14115082976794\\
0.0502876746443129	-2.08565545733914\\
0.0523622362011335	-2.03055486965141\\
0.0544367977579541	-1.97585216146273\\
0.0565113593147748	-1.92154967066953\\
0.0585859208715954	-1.86764911318865\\
0.0606604824284161	-1.81415169736893\\
0.0631477595675582	-1.75054395894324\\
0.0656350367067003	-1.68751757424484\\
0.0681223138458424	-1.62507271579427\\
0.0706095909849846	-1.56320911367265\\
0.0730968681241267	-1.50192614169868\\
0.0755841452632688	-1.44122288690693\\
0.0780714224024109	-1.38109820832085\\
0.0813020082716092	-1.30386718715457\\
0.0845325941408074	-1.22760671499274\\
0.0877631800100057	-1.15231321011813\\
0.0909937658792039	-1.07798284124073\\
0.0942243517484022	-1.00461157357985\\
0.0974549376176004	-0.932195204836231\\
0.100685523486799	-0.860729397431077\\
0.10514741504118	-0.763579897441353\\
0.109609306595561	-0.668223125402594\\
0.114071198149942	-0.574646776913388\\
0.118533089704324	-0.482838293778959\\
0.122994981258705	-0.392784882231362\\
0.127456872813086	-0.30447352957636\\
0.131918764367467	-0.217891024112563\\
0.138975857326099	-0.0844458513920165\\
0.146032950284731	0.0447619152475824\\
0.153090043243363	0.169787205231275\\
0.160147136201996	0.29068592707584\\
0.167204229160628	0.407514900588518\\
0.17426132211926	0.52033179824524\\
0.18490283178396	0.682993526045554\\
0.195544341448661	0.83686940032728\\
0.206185851113361	0.982166634283427\\
0.216827360778062	1.11909555288541\\
0.227468870442763	1.24786919014629\\
0.238110380107463	1.36870290622218\\
0.248751889772164	1.48181401701647\\
0.264264601985971	1.63332026358769\\
0.279777314199778	1.76956343683851\\
0.295290026413586	1.89122912527476\\
0.310802738627393	1.99900384700049\\
0.3263154508412	2.09357291774344\\
0.341828163055008	2.1756183999908\\
0.357340875268815	2.24581714595503\\
0.373617632703406	2.30746886161346\\
0.389894390137998	2.3575787554332\\
0.406171147572589	2.39689744357995\\
0.42244790500718	2.42616097376545\\
0.438724662441772	2.44608907179932\\
0.455001419876363	2.45738354703871\\
0.471278177310954	2.46072686000305\\
0.491045458137534	2.4550473993433\\
0.510812738964114	2.43975784765235\\
0.530580019790694	2.41595039254861\\
0.550347300617274	2.38466527469356\\
0.570114581443854	2.34688913404768\\
0.589881862270433	2.30355376667536\\
0.609649143097013	2.25553530397734\\
0.629416423923593	2.20365379453238\\
0.649183704750173	2.14867314877119\\
0.668950985576753	2.09130141216412\\
0.688718266403333	2.0321913473451\\
0.708718084154522	1.9712279899019\\
0.728717901905712	1.90966181221128\\
0.748717719656901	1.8480026968692\\
0.768717537408091	1.78670785867777\\
0.78871735515928	1.726183724253\\
0.80871717291047	1.66678798325053\\
0.828716990661659	1.60883178583477\\
0.850907308267992	1.54653494196815\\
0.873097625874324	1.48664240340919\\
0.895287943480656	1.42939951836414\\
0.917478261086989	1.37499669487382\\
0.939668578693321	1.32357371060053\\
0.961858896299653	1.27522402008666\\
0.984049213905985	1.22999902362441\\
1.00869106575997	1.18345453698551\\
1.03333291761395	1.14074931782293\\
1.05797476946794	1.10180890834999\\
1.08261662132192	1.06652068135127\\
1.1072584731759	1.03473999402236\\
1.13190032502989	1.0062959933074\\
1.15654217688387	0.980997028814215\\
1.18283314760554	0.957238670415711\\
1.20912411832722	0.936558320704616\\
1.23541508904889	0.918681535293652\\
1.26170605977057	0.903330329229272\\
1.28799703049224	0.890227783470935\\
1.31428800121392	0.879102058530613\\
1.34057897193559	0.869689809873418\\
1.37287620644639	0.860102471472964\\
1.40517344095719	0.852282984254774\\
1.43747067546799	0.845829050992621\\
1.46976790997879	0.840377855844124\\
1.50206514448959	0.835608189775558\\
1.53436237900039	0.831241401740163\\
1.56665961351118	0.827041221492158\\
1.6002697744721	0.822637721051683\\
1.63387993543302	0.81802908129197\\
1.66749009639394	0.81308381755953\\
1.70110025735486	0.807710024010146\\
1.73471041831577	0.801851376297928\\
1.76832057927669	0.795482869579632\\
1.80193074023761	0.788606454098751\\
1.83554090119853	0.781246689162184\\
1.86915106215945	0.773446499350447\\
1.90315811876965	0.765164483647755\\
1.93716517537985	0.756560366766133\\
1.97117223199006	0.747711590092384\\
2.00517928860026	0.73869960755423\\
2.03918634521047	0.729606928833373\\
2.07319340182067	0.720514562700866\\
2.10720045843087	0.711499895320543\\
2.15155823425687	0.699976840070827\\
2.19591601008287	0.688850619658495\\
2.24027378590886	0.678241228059834\\
2.28463156173486	0.668244390780304\\
2.32898933756086	0.658930985929163\\
2.37334711338685	0.650347539544501\\
2.41770488921285	0.64251773742233\\
2.46764913538314	0.634607077535375\\
2.51759338155343	0.627633133754436\\
2.56753762772372	0.621550308783045\\
2.61748187389401	0.616295109536756\\
2.66742612006429	0.611791583054045\\
2.71737036623458	0.607956401107154\\
2.76731461240487	0.604703328867186\\
2.81986071624071	0.60181532001947\\
2.87240682007654	0.599382430093577\\
2.92495292391237	0.597317795436823\\
2.97749902774821	0.595544272590131\\
3.03004513158404	0.593995671974684\\
3.08259123541987	0.592617260103309\\
3.13513733925571	0.591365564210584\\
3.20243296799289	0.58989581736498\\
3.26972859673007	0.588533141199518\\
3.33702422546725	0.587250899584672\\
3.40431985420443	0.58603820097753\\
3.4716154829416	0.584895052215018\\
3.53891111167878	0.583827975909683\\
3.60620674041596	0.582846348494439\\
3.68249541704078	0.581848910742543\\
3.75878409366559	0.580984040019254\\
3.8350727702904	0.580256544376634\\
3.91136144691521	0.579665427810705\\
3.98765012354003	0.579204160738141\\
4.06393880016484	0.578861244741391\\
4.14022747678965	0.57862135814461\\
4.23092996141938	0.57844621400165\\
4.32163244604912	0.578362743907165\\
4.41233493067885	0.578343478216576\\
4.50303741530858	0.578364420955604\\
4.59373989993832	0.578406192502787\\
4.68444238456805	0.578454785998902\\
4.77514486919778	0.578501265239061\\
4.88665703125435	0.578548214752662\\
4.99816919331091	0.578580873062233\\
5.10968135536747	0.578600000689536\\
5.22119351742403	0.578608758723229\\
5.33270567948059	0.578610892879084\\
5.44421784153716	0.578609434196396\\
5.55573000359372	0.578606373908549\\
5.70470801839516	0.578601921354251\\
5.8536860331966	0.578598047551446\\
6.00266404799804	0.578594641457905\\
6.15164206279947	0.578591049336092\\
6.30062007760091	0.578586941162957\\
6.44959809240235	0.578582637925044\\
6.59857610720379	0.578578665055301\\
6.81931619337614	0.578573868017818\\
7.04005627954849	0.578570787307957\\
7.26079636572084	0.578569200844778\\
7.48153645189319	0.578568516424286\\
7.70227653806554	0.578568208297976\\
7.92301662423789	0.578568035804808\\
8.22301662423789	0.57856799041242\\
8.52301662423789	0.578568141738282\\
8.82301662423789	0.578568206375227\\
9.12301662423789	0.578568056855325\\
9.42301662423789	0.578567953192745\\
9.72301662423789	0.578567997885579\\
10.0230166242379	0.578568109809178\\
10.3230166242379	0.578568101062876\\
10.6230166242379	0.578568025349278\\
10.9230166242379	0.57856803680217\\
11.2230166242379	0.578568043763787\\
11.5230166242379	0.578568048417004\\
11.8230166242379	0.57856804750825\\
12.1230166242379	0.578568045954864\\
12.4230166242379	0.578568046011634\\
12.7230166242379	0.578568046542826\\
13.0230166242379	0.578568046546289\\
13.3230166242379	0.578568046259311\\
13.6230166242379	0.578568046124431\\
13.9230166242379	0.578568046154946\\
14.2230166242379	0.578568046186546\\
14.5230166242379	0.57856804617179\\
14.8230166242379	0.578568046154543\\
15	0.578568046156311\\
};
\addplot [color=mycolor2, thick]
  table[row sep=crcr]{%
0	-0.862605362039618\\
1.90494919136049e-07	-0.862605134436397\\
3.80989838272098e-07	-0.862604906832596\\
5.71484757408146e-07	-0.862604679228203\\
1.40362506689516e-06	-0.862603684975834\\
2.23576537638218e-06	-0.86260269071237\\
3.06790568586919e-06	-0.862601696437634\\
3.90004599535621e-06	-0.862600702151567\\
6.42033919822868e-06	-0.862597690700266\\
8.94063240110115e-06	-0.862594679144908\\
1.14609256039736e-05	-0.862591667485491\\
1.39812188068461e-05	-0.862588655722015\\
2.31122755165313e-05	-0.862577743190367\\
3.22433322262165e-05	-0.862566829292896\\
4.13743889359017e-05	-0.862555914029642\\
5.05054456455869e-05	-0.862544997400638\\
5.96365023552721e-05	-0.862534079405919\\
8.87080004601316e-05	-0.862499309551189\\
0.000117779498564991	-0.862464525853915\\
0.000146850996669851	-0.862429728314799\\
0.00017592249477471	-0.862394916934335\\
0.00020499399287957	-0.862360091712818\\
0.000257350826916829	-0.862297337697728\\
0.000309707660954089	-0.862234538788655\\
0.000362064494991349	-0.862171694982657\\
0.000414421329028608	-0.86210880627488\\
0.000466778163065868	-0.862045872658585\\
0.000519134997103128	-0.861982894125193\\
0.00062951347348196	-0.861849975810465\\
0.000739891949860791	-0.861716857690935\\
0.000850270426239623	-0.861583539610324\\
0.000960648902618454	-0.861450021380517\\
0.00107102737899729	-0.861316302783042\\
0.00118140585537612	-0.861182383570428\\
0.00134948275112732	-0.860978074475323\\
0.00151755964687852	-0.860773298478774\\
0.00168563654262973	-0.860568054361966\\
0.00185371343838093	-0.860362340784269\\
0.00202179033413213	-0.860156156292308\\
0.00218986722988334	-0.859949499328651\\
0.00235794412563454	-0.859742368240056\\
0.00265480285016022	-0.859375368091445\\
0.0029516615746859	-0.859006873248267\\
0.00324852029921158	-0.858636872851135\\
0.00354537902373726	-0.858265355475108\\
0.00384223774826295	-0.857892309224115\\
0.00413909647278863	-0.85751772181694\\
0.00443595519731431	-0.857141580664944\\
0.00481404091763571	-0.856660249220299\\
0.00519212663795712	-0.856176349711658\\
0.00557021235827852	-0.85568985496864\\
0.00594829807859993	-0.855200737701845\\
0.00632638379892133	-0.854708970641293\\
0.00670446951924273	-0.854214526657267\\
0.00708255523956414	-0.853717378863923\\
0.00755285141043324	-0.853095169447969\\
0.00802314758130234	-0.85246868527027\\
0.00849344375217143	-0.851837877269311\\
0.00896373992304053	-0.851202697678198\\
0.00943403609390963	-0.850563100137366\\
0.00990433226477873	-0.849919039782634\\
0.0103746284356478	-0.849270473309774\\
0.0109599759457901	-0.848456886828408\\
0.0115453234559323	-0.847636178475489\\
0.0121306709660745	-0.846808275323743\\
0.0127160184762168	-0.845973108379196\\
0.013301365986359	-0.845130612540569\\
0.0138867134965012	-0.844280726536447\\
0.0144720610066435	-0.843423392841935\\
0.0152109749913156	-0.842330418674645\\
0.0159498889759876	-0.841225390427701\\
0.0166888029606597	-0.840108216306788\\
0.0174277169453318	-0.838978812692371\\
0.0181666309300039	-0.837837103725695\\
0.018905544914676	-0.836683020897885\\
0.0196444588993481	-0.835516502643442\\
0.0206121406972682	-0.833969932851734\\
0.0215798224951882	-0.832401834149836\\
0.0225475042931083	-0.830812115445968\\
0.0235151860910284	-0.829200700978179\\
0.0244828678889485	-0.827567529055881\\
0.0254505496868686	-0.82591255088851\\
0.0264182314847887	-0.824235729499133\\
0.0278471252801151	-0.821719715446187\\
0.0292760190754416	-0.819155973746686\\
0.0307049128707681	-0.816544500352931\\
0.0321338066660945	-0.813885326921808\\
0.033562700461421	-0.811178516349054\\
0.0349915942567475	-0.808424159047046\\
0.0365840167244509	-0.805298702177903\\
0.0381764391921544	-0.802114525007254\\
0.0397688616598578	-0.798871844338564\\
0.0413612841275613	-0.795570901082089\\
0.0429537065952647	-0.792211957051031\\
0.0445461290629681	-0.788795291984903\\
0.0461385515306716	-0.785321200961346\\
0.0482131130874922	-0.780709628754473\\
0.0502876746443129	-0.776001827339182\\
0.0523622362011335	-0.771198534451007\\
0.0544367977579541	-0.766300510712292\\
0.0565113593147748	-0.761308535769943\\
0.0585859208715954	-0.756223405054742\\
0.0606604824284161	-0.75104592715315\\
0.0631477595675582	-0.744717856776986\\
0.0656350367067003	-0.738259649738019\\
0.0681223138458424	-0.731672751715798\\
0.0706095909849846	-0.724958621376846\\
0.0730968681241267	-0.718118728130791\\
0.0755841452632688	-0.711154550274603\\
0.0780714224024109	-0.704067573494423\\
0.0813020082716092	-0.69468181456365\\
0.0845325941408074	-0.68509469348949\\
0.0877631800100057	-0.675309504140685\\
0.0909937658792039	-0.66532955022023\\
0.0942243517484022	-0.655158143411152\\
0.0974549376176004	-0.644798601753676\\
0.100685523486799	-0.634254248267665\\
0.10514741504118	-0.619393291197495\\
0.109609306595561	-0.604194916689355\\
0.114071198149942	-0.588667903985223\\
0.118533089704324	-0.572821034317351\\
0.122994981258705	-0.556663087158364\\
0.127456872813086	-0.540202836579396\\
0.131918764367467	-0.523449047795408\\
0.138975857326099	-0.496372672990109\\
0.146032950284731	-0.468618431009129\\
0.153090043243363	-0.44022074926938\\
0.160147136201996	-0.411213899978309\\
0.167204229160628	-0.381631972071701\\
0.17426132211926	-0.351508845098655\\
0.18490283178396	-0.305136656058911\\
0.195544341448661	-0.257724531677064\\
0.206185851113361	-0.209385142745299\\
0.216827360778062	-0.160229443598852\\
0.227468870442763	-0.110366505232193\\
0.238110380107463	-0.0599033568535172\\
0.248751889772164	-0.00894483692515147\\
0.264264601985971	0.066016517025127\\
0.279777314199778	0.141505102020995\\
0.295290026413586	0.217224660816506\\
0.310802738627393	0.292891439614922\\
0.3263154508412	0.36823475641024\\
0.341828163055008	0.442997484802094\\
0.357340875268815	0.51693645091494\\
0.373617632703406	0.593381552010951\\
0.389894390137998	0.668420591075945\\
0.406171147572589	0.741824048222427\\
0.42244790500718	0.813380247214782\\
0.438724662441772	0.882895356108992\\
0.455001419876363	0.950193301477798\\
0.471278177310954	1.01511560044325\\
0.491045458137534	1.09056388585693\\
0.510812738964114	1.16208209230755\\
0.530580019790694	1.22948715890569\\
0.550347300617274	1.29263085234416\\
0.570114581443854	1.35139840833936\\
0.589881862270433	1.40570704409188\\
0.609649143097013	1.45550435190583\\
0.629416423923593	1.50076659054246\\
0.649183704750173	1.54149689827957\\
0.668950985576753	1.57772345027447\\
0.688718266403333	1.60949757736048\\
0.708718084154522	1.63718840385168\\
0.728717901905712	1.66049107917396\\
0.748717719656901	1.67951914982789\\
0.768717537408091	1.69440240254171\\
0.78871735515928	1.70528489958045\\
0.80871717291047	1.71232306068505\\
0.828716990661659	1.71568380199197\\
0.850907308267992	1.71532226904658\\
0.873097625874324	1.71089982312646\\
0.895287943480656	1.70267410402079\\
0.917478261086989	1.69090801008124\\
0.939668578693321	1.67586729001928\\
0.961858896299653	1.65781831169447\\
0.984049213905985	1.63702601510216\\
1.00869106575997	1.61103971836219\\
1.03333291761395	1.5823440404543\\
1.05797476946794	1.55127996521167\\
1.08261662132192	1.51817630015548\\
1.1072584731759	1.48334785540983\\
1.13190032502989	1.44709393394188\\
1.15654217688387	1.40969712798264\\
1.18283314760554	1.36883577185768\\
1.20912411832722	1.32727369294693\\
1.23541508904889	1.285285548637\\
1.26170605977057	1.24312160038886\\
1.28799703049224	1.2010078121968\\
1.31428800121392	1.15914621183291\\
1.34057897193559	1.11771549008187\\
1.37287620644639	1.0676377804865\\
1.40517344095719	1.01868983150998\\
1.43747067546799	0.971070903503555\\
1.46976790997879	0.924937842096255\\
1.50206514448959	0.880408823604084\\
1.53436237900039	0.837567214892545\\
1.56665961351118	0.796465493245561\\
1.6002697744721	0.755567524890609\\
1.63387993543302	0.71658338702191\\
1.66749009639394	0.679492523590009\\
1.70110025735486	0.644255830018333\\
1.73471041831577	0.610819529746629\\
1.76832057927669	0.579118713782549\\
1.80193074023761	0.549080518270864\\
1.83554090119853	0.52062692727879\\
1.86915106215945	0.493677200481327\\
1.90315811876965	0.46785666032196\\
1.93716517537985	0.443409150081984\\
1.97117223199006	0.420253988767283\\
2.00517928860026	0.398314157001874\\
2.03918634521047	0.377517115630853\\
2.07319340182067	0.357795352834935\\
2.10720045843087	0.33908668933685\\
2.15155823425687	0.316112926978803\\
2.19591601008287	0.29465447088833\\
2.24027378590886	0.274613681118164\\
2.28463156173486	0.25590615709543\\
2.32898933756086	0.238459048011262\\
2.37334711338685	0.222209232578139\\
2.41770488921285	0.207101435998178\\
2.46764913538314	0.19139648481149\\
2.51759338155343	0.177012957905572\\
2.56753762772372	0.163890997952108\\
2.61748187389401	0.151972660742114\\
2.66742612006429	0.141200191321418\\
2.71737036623458	0.131514809613582\\
2.76731461240487	0.122855980332155\\
2.81986071624071	0.114785473550612\\
2.87240682007654	0.107706719931765\\
2.92495292391237	0.101542526788978\\
2.97749902774821	0.0962145440916169\\
3.03004513158404	0.0916442171294474\\
3.08259123541987	0.0877538292093432\\
3.13513733925571	0.0844675668018242\\
3.20243296799289	0.081025112180515\\
3.26972859673007	0.0783104758841571\\
3.33702422546725	0.0761925163751451\\
3.40431985420443	0.074554431373353\\
3.4716154829416	0.0732946337598008\\
3.53891111167878	0.0723268877178862\\
3.60620674041596	0.071579806311715\\
3.68249541704078	0.0709275592626257\\
3.75878409366559	0.0704206575720194\\
3.8350727702904	0.0700119755989457\\
3.91136144691521	0.0696686098112842\\
3.98765012354003	0.069368788442203\\
4.06393880016484	0.0690990832913601\\
4.14022747678965	0.0688519343057665\\
4.23092996141938	0.0685824430864903\\
4.32163244604912	0.0683372094003722\\
4.41233493067885	0.0681156600310349\\
4.50303741530858	0.067917630018532\\
4.59373989993832	0.0677425575194324\\
4.68444238456805	0.0675891791398069\\
4.77514486919778	0.0674556203230186\\
4.88665703125435	0.0673153218544718\\
4.99816919331091	0.0671973769254013\\
5.10968135536747	0.0670978711124439\\
5.22119351742403	0.0670136178157774\\
5.33270567948059	0.0669422645932061\\
5.44421784153716	0.066882153663018\\
5.55573000359372	0.0668320588253453\\
5.70470801839516	0.0667790329781575\\
5.8536860331966	0.0667398807267627\\
6.00266404799804	0.0667123343304947\\
6.15164206279947	0.0666939862549132\\
6.30062007760091	0.0666823976518485\\
6.44959809240235	0.0666753709839476\\
6.59857610720379	0.0666712048651492\\
6.81931619337614	0.0666678873321604\\
7.04005627954849	0.066666096085183\\
7.26079636572084	0.0666649142863136\\
7.48153645189319	0.0666640636582487\\
7.70227653806554	0.0666634289042261\\
7.92301662423789	0.0666629034912649\\
8.22301662423789	0.0666622926677949\\
8.52301662423789	0.0666618489355661\\
8.82301662423789	0.0666616146812498\\
9.12301662423789	0.0666615027312254\\
9.42301662423789	0.0666614184134309\\
9.72301662423789	0.0666613645382679\\
10.0230166242379	0.0666613760737214\\
10.3230166242379	0.0666614134648347\\
10.6230166242379	0.0666614239061465\\
10.9230166242379	0.066661418767179\\
11.2230166242379	0.0666614138918398\\
11.5230166242379	0.0666614115396289\\
11.8230166242379	0.0666614095127276\\
12.1230166242379	0.06666140736088\\
12.4230166242379	0.0666614058941768\\
12.7230166242379	0.0666614053213284\\
13.0230166242379	0.0666614051851804\\
13.3230166242379	0.066661405118281\\
13.6230166242379	0.0666614050687464\\
13.9230166242379	0.0666614050724261\\
14.2230166242379	0.0666614051079382\\
14.5230166242379	0.0666614051348597\\
14.8230166242379	0.0666614051439556\\
15	0.0666614051461467\\
};
\end{axis}
\end{tikzpicture}%

%% file: plots/transient_dyn_eta_reduced.tex
%
%
\definecolor{mycolor1}{RGB}{0,81,158}%
\definecolor{mycolor2}{RGB}{0, 190, 255}
\definecolor{usGray}{RGB}{139,143,148}%
\begin{tikzpicture}

\begin{axis}[%
width=1.15in,
height=0.8in,
at={(0.758in,0.481in)},
scale only axis,
xmin=0,
xmax=8,
ymin=-2,
ymax=3.5,
axis background/.style={fill=white},
xmajorgrids,
ymajorgrids,
ylabel style = {yshift=-3.5mm},
ylabel={$y$},
xlabel={$t$},
xlabel style = {yshift=1.5mm},
]

\addplot [color=gray, dashed, thick]
  table[row sep=crcr]{%
0	0.576082308650491\\
20	0.576082308650491\\
};
\addplot [color=gray, dashed, thick]
  table[row sep=crcr]{%
0	0.0627947023452091\\
20	0.0627947023452091\\
};
\addplot [color=mycolor1, thick]
  table[row sep=crcr]{%
0	-3.48914335047704\\
1.28615122312548e-07	-3.48913988691896\\
2.57230244625095e-07	-3.48913642336231\\
3.85845366937643e-07	-3.48913295980712\\
9.48078129744786e-07	-3.48911781911363\\
1.51031089255193e-06	-3.48910267844751\\
2.07254365535907e-06	-3.48908753780915\\
2.63477641816621e-06	-3.48907239719866\\
4.35933951202973e-06	-3.48902595586293\\
6.08390260589324e-06	-3.48897951478822\\
7.80846569975676e-06	-3.48893307397344\\
9.53302879362027e-06	-3.48888663341745\\
1.53912396184108e-05	-3.48872888026756\\
2.12494504432013e-05	-3.48857113004717\\
2.71076612679918e-05	-3.48841338271232\\
3.29658720927823e-05	-3.4882556382194\\
3.88240829175729e-05	-3.48809789652513\\
5.47446266771513e-05	-3.48766922422233\\
7.06651704367298e-05	-3.48724057141716\\
8.65857141963082e-05	-3.48681193726738\\
0.000102506257955887	-3.48638332094293\\
0.000118426801715465	-3.48595472162601\\
0.000134347345475044	-3.48552613851113\\
0.000189779887373647	-3.48403400245621\\
0.000245212429272251	-3.48254202093974\\
0.000300644971170855	-3.48105016323868\\
0.000356077513069458	-3.47955840021515\\
0.000411510054968062	-3.47806670424035\\
0.000466942596866666	-3.47657504912594\\
0.000532203611459508	-3.47481893471714\\
0.000597464626052351	-3.4730628044465\\
0.000662725640645194	-3.47130662252427\\
0.000727986655238036	-3.46955035541473\\
0.000793247669830879	-3.46779397171586\\
0.000858508684423721	-3.4660374420418\\
0.000923769699016564	-3.464280738911\\
0.00103873397645211	-3.46118561866387\\
0.00115369825388766	-3.45808975357046\\
0.0012686625313232	-3.45499303067767\\
0.00138362680875875	-3.45189535048571\\
0.00149859108619429	-3.44879662568462\\
0.00161355536362984	-3.44569677999948\\
0.00172851964106539	-3.44259574713463\\
0.00187553903271815	-3.43862824326918\\
0.00202255842437091	-3.43465860436387\\
0.00216957781602367	-3.43068674900766\\
0.00231659720767643	-3.42671261264054\\
0.00246361659932919	-3.42273614551624\\
0.00261063599098195	-3.41875731088656\\
0.00275765538263471	-3.41477608337864\\
0.00294435013522383	-3.4097169907331\\
0.00313104488781295	-3.40465400359182\\
0.00331773964040207	-3.39958712926241\\
0.00350443439299119	-3.39451639101811\\
0.0036911291455803	-3.38944182553898\\
0.00387782389816942	-3.3843634807009\\
0.00406451865075854	-3.37928141364609\\
0.00430327930002514	-3.37277672693328\\
0.00454203994929175	-3.36626620889482\\
0.00478080059855835	-3.359750022315\\
0.00501956124782495	-3.35322833957223\\
0.00525832189709155	-3.34670134014747\\
0.00549708254635816	-3.3401692085581\\
0.00573584319562476	-3.33363213259729\\
0.00604854021380708	-3.32506357634591\\
0.0063612372319894	-3.31648729240569\\
0.00667393425017172	-3.3079037072677\\
0.00698663126835404	-3.29931324380158\\
0.00729932828653637	-3.29071631960092\\
0.00761202530471869	-3.28211334571949\\
0.00792472232290101	-3.27350472561742\\
0.00835348610700806	-3.26169230251737\\
0.0087822498911151	-3.24987099713631\\
0.00921101367522215	-3.23804176885584\\
0.0096397774593292	-3.226205544314\\
0.0100685412434362	-3.21436321774773\\
0.0104973050275433	-3.20251565151402\\
0.0109260688116503	-3.19066367656223\\
0.011573730047872	-3.17275430192017\\
0.0122213912840937	-3.15483931597395\\
0.0128690525203154	-3.13692120009617\\
0.0135167137565371	-3.11900229575404\\
0.0141643749927588	-3.10108481148962\\
0.0148120362289805	-3.08317082966669\\
0.0158195685650938	-3.05531454995912\\
0.0168271009012072	-3.02747832737779\\
0.0178346332373206	-2.99966841568214\\
0.0188421655734339	-2.97189051820851\\
0.0198496979095473	-2.94414983540772\\
0.0208572302456606	-2.91645110726834\\
0.0220008301472745	-2.88506794082387\\
0.0231444300488884	-2.85375009079386\\
0.0242880299505023	-2.82250269907003\\
0.0254316298521162	-2.79133038471993\\
0.0265752297537301	-2.76023729286117\\
0.027718829655344	-2.72922714141304\\
0.0288624295569579	-2.69830326382292\\
0.0304835349977048	-2.65462083822275\\
0.0321046404384517	-2.6111251681712\\
0.0337257458791986	-2.56782260393163\\
0.0353468513199455	-2.52471854496318\\
0.0369679567606924	-2.48181756572091\\
0.0385890622014393	-2.43912352614222\\
0.0402101676421862	-2.3966396689264\\
0.0420399985456778	-2.34894169528252\\
0.0438698294491695	-2.30151814629754\\
0.0456996603526611	-2.2543715801982\\
0.0475294912561527	-2.20750401441053\\
0.0493593221596444	-2.16091700597072\\
0.051189153063136	-2.11461172038557\\
0.0530189839666277	-2.06858899125597\\
0.0552171038213478	-2.01367764162156\\
0.0574152236760679	-1.959175327021\\
0.059613343530788	-1.90508214831196\\
0.0618114633855081	-1.85139787212273\\
0.0640095832402282	-1.7981219892908\\
0.0662077030949483	-1.74525376370433\\
0.0684058229496684	-1.69279227416281\\
0.0711159409684139	-1.62866944352366\\
0.0738260589871594	-1.56516099581198\\
0.0765361770059048	-1.50226444922117\\
0.0792462950246503	-1.43997715093935\\
0.0819564130433958	-1.37829631405668\\
0.0846665310621412	-1.31721904749891\\
0.0873766490808867	-1.25674238149194\\
0.0908656754953445	-1.17976370515236\\
0.0943547019098024	-1.10376888594532\\
0.0978437283242602	-1.02875127803764\\
0.101332754738718	-0.954704184510877\\
0.104821781153176	-0.881620876212344\\
0.108310807567634	-0.809494606501236\\
0.111799833982092	-0.738318624299765\\
0.116704214488559	-0.639863251678736\\
0.121608594995026	-0.543253505134063\\
0.126512975501493	-0.448470753694346\\
0.13141735600796	-0.355496446294065\\
0.136321736514428	-0.264312116775021\\
0.141226117020895	-0.174899387920899\\
0.146130497527362	-0.0872399769992513\\
0.154167762509756	0.0526722527647503\\
0.16220502749215	0.188003954813752\\
0.170242292474544	0.318833891086842\\
0.178279557456937	0.445240131639722\\
0.186316822439331	0.567300057422469\\
0.194354087421725	0.685090365998769\\
0.208509335515431	0.88239265261326\\
0.222664583609137	1.06709715093659\\
0.236819831702842	1.23960803464468\\
0.250975079796548	1.40032341208012\\
0.265130327890254	1.54963542753718\\
0.279285575983959	1.6879303545709\\
0.296335646065869	1.84041944572626\\
0.313385716147779	1.97813023320672\\
0.330435786229689	2.1017047865066\\
0.347485856311599	2.21177338116642\\
0.364535926393509	2.30895455631454\\
0.381585996475419	2.39385515202978\\
0.398636066557329	2.46707031443608\\
0.422370035404442	2.55062841116006\\
0.446104004251554	2.6142079033575\\
0.469837973098666	2.65929983508395\\
0.493571941945779	2.68735013384075\\
0.517305910792891	2.69975901676686\\
0.541039879640004	2.69788034671925\\
0.564773848487116	2.68302095714507\\
0.591475198996247	2.65235311570937\\
0.618176549505377	2.6085676773939\\
0.644877900014508	2.55330560062196\\
0.671579250523639	2.48812819045896\\
0.69828060103277	2.41451621932605\\
0.7249819515419	2.33386923210688\\
0.751683302051031	2.24750508008538\\
0.784632820449226	2.13488054524767\\
0.817582338847421	2.01757735187805\\
0.850531857245616	1.89755710342868\\
0.883481375643812	1.77659442758659\\
0.916430894042007	1.65627973438345\\
0.949380412440202	1.5380229034004\\
0.982329930838397	1.42305799562312\\
1.02051401380824	1.2953398659131\\
1.05869809677809	1.17484693942931\\
1.09688217974794	1.06268070667745\\
1.13506626271778	0.959679477297217\\
1.17325034568763	0.866436346769465\\
1.21143442865748	0.783318235424123\\
1.24961851162733	0.710485936761563\\
1.29405047445646	0.638642752854087\\
1.3384824372856	0.580370841374645\\
1.38291440011473	0.53509719242483\\
1.42734636294387	0.502035982891602\\
1.471778325773	0.480231385180631\\
1.51621028860214	0.468598572270845\\
1.56064225143127	0.465962386414658\\
1.60507421426041	0.471093212908612\\
1.64950617708954	0.482739719097042\\
1.69393813991868	0.499658217609456\\
1.73837010274782	0.52063847892944\\
1.78280206557695	0.544525883147803\\
1.82723402840609	0.570239878993156\\
1.87166599123522	0.59678879807503\\
1.91609795406436	0.623281139413231\\
1.96052991689349	0.648933490741974\\
2.00496187972263	0.673075294711592\\
2.04939384255176	0.6951507061397\\
2.09513875576048	0.71525402288933\\
2.14088366896919	0.732339629046864\\
2.18662858217791	0.746166693725907\\
2.23237349538662	0.756604847503589\\
2.27811840859533	0.763624765276043\\
2.32386332180405	0.767287603260642\\
2.36960823501276	0.767733439658424\\
2.42131349170287	0.764626439540102\\
2.47301874839297	0.758056383448936\\
2.52472400508308	0.748470915519199\\
2.57642926177318	0.736362894313434\\
2.62813451846329	0.722250156505981\\
2.67983977515339	0.706656972728624\\
2.7315450318435	0.690097676872401\\
2.7832502885336	0.673062666240757\\
2.83495554522371	0.656006736033444\\
2.88666080191381	0.639339750722301\\
2.93836605860391	0.623419723529997\\
2.99007131529402	0.608548266612183\\
3.04213549719279	0.594879088484485\\
3.09419967909155	0.582706501022971\\
3.14626386099032	0.572158095604881\\
3.19832804288909	0.563305065162269\\
3.25039222478785	0.55616633788132\\
3.30245640668662	0.550713645166329\\
3.35452058858538	0.546877476491009\\
3.41956112325278	0.544195113181507\\
3.48460165792018	0.543608664438117\\
3.54964219258758	0.54480826032166\\
3.61468272725498	0.547456046681059\\
3.67972326192237	0.551204276031877\\
3.74476379658977	0.555711333961879\\
3.80980433125717	0.560655201763729\\
3.87484486592457	0.56574410491885\\
3.93988540059197	0.570724370024883\\
4.00492593525937	0.575385593280848\\
4.06996646992676	0.579563238895793\\
4.13500700459416	0.583138883294671\\
4.2027738289708	0.586144407484323\\
4.27054065334744	0.588379719917327\\
4.33830747772408	0.58985090271234\\
4.40607430210072	0.590599480629329\\
4.47384112647735	0.590695121533581\\
4.54160795085399	0.590228258236689\\
4.60937477523063	0.589302664010211\\
4.68354868502109	0.587893987860246\\
4.75772259481154	0.586211664672572\\
4.831896504602	0.584391548052909\\
4.90607041439245	0.582555318113783\\
4.98024432418291	0.580806080985843\\
5.05441823397336	0.579225354679408\\
5.12859214376381	0.577872322379568\\
5.2094628947603	0.576699865365944\\
5.29033364575678	0.575862116537422\\
5.37120439675327	0.57535366050767\\
5.45207514774975	0.575148795908043\\
5.53294589874623	0.575206814527894\\
5.61381664974272	0.575477002432445\\
5.6946874007392	0.575903605715471\\
5.78423695233517	0.576490685861189\\
5.87378650393113	0.577129901466554\\
5.9633360555271	0.577760340428476\\
6.05288560712306	0.578333618103193\\
6.14243515871903	0.578815328712524\\
6.231984710315	0.579184757803413\\
6.32153426191096	0.579434049961272\\
6.43552860544557	0.579583655918971\\
6.54952294898017	0.579571350562783\\
6.66351729251478	0.579437947843443\\
6.77751163604938	0.579229200709695\\
6.89150597958398	0.578988878304726\\
7.00550032311859	0.578754159050641\\
7.11949466665319	0.578552798831785\\
7.24203666586693	0.578392784321236\\
7.36457866508066	0.578298948631222\\
7.4871206642944	0.578267805890095\\
7.60966266350814	0.578287623554306\\
7.73220466272187	0.578342332383785\\
7.85474666193561	0.578415073897494\\
7.98431875509124	0.578495067046873\\
8.11389084824686	0.578564776502465\\
8.24346294140249	0.578616108825027\\
8.37303503455812	0.578646049703382\\
8.50260712771375	0.578655630057952\\
8.63217922086938	0.578648778791012\\
8.81321192723325	0.578622194587718\\
8.99424463359713	0.578588749506576\\
9.17527733996101	0.578560421268248\\
9.35631004632488	0.578544016245433\\
9.53734275268876	0.578540562486334\\
9.71837545905263	0.578546752305666\\
9.94135056481037	0.578559996020987\\
10.1643256705681	0.578571878828614\\
10.3873007763258	0.578577638087539\\
10.6102758820836	0.578576853451773\\
10.8332509878413	0.578572216940451\\
11.1332509878413	0.578565690520393\\
11.4332509878413	0.578563119489522\\
11.7332509878413	0.578565103247465\\
12.0332509878413	0.578568851443027\\
12.3332509878413	0.578571080240309\\
12.6332509878413	0.578570499183214\\
12.9332509878413	0.578568605612606\\
13.2332509878413	0.578566986821809\\
13.5332509878413	0.578566463155898\\
13.8332509878413	0.578566942235331\\
14.1332509878413	0.578567751370251\\
14.4332509878413	0.578568194280725\\
14.7332509878413	0.578568333925387\\
15	0.578568318129197\\
};
\addplot [color=mycolor2, thick]
  table[row sep=crcr]{%
0	-0.862605362039618\\
1.28615122312548e-07	-0.862605208370453\\
2.57230244625095e-07	-0.862605054701024\\
3.85845366937643e-07	-0.862604901031325\\
9.48078129744786e-07	-0.862604229271127\\
1.51031089255193e-06	-0.862603557505864\\
2.07254365535907e-06	-0.862602885735456\\
2.63477641816621e-06	-0.862602213959875\\
4.35933951202973e-06	-0.862600153358615\\
6.08390260589324e-06	-0.862598092708632\\
7.80846569975676e-06	-0.862596032009925\\
9.53302879362027e-06	-0.862593971262494\\
1.53912396184108e-05	-0.862586970694962\\
2.12494504432013e-05	-0.862579969565226\\
2.71076612679918e-05	-0.862572967873297\\
3.29658720927823e-05	-0.862565965619185\\
3.88240829175729e-05	-0.862558962802896\\
5.47446266771513e-05	-0.86253992878673\\
7.06651704367298e-05	-0.862520890618757\\
8.65857141963082e-05	-0.86250184829909\\
0.000102506257955887	-0.86248280182781\\
0.000118426801715465	-0.862463751204971\\
0.000134347345475044	-0.862444696430599\\
0.000189779887373647	-0.862378318654785\\
0.000245212429272251	-0.862311890546493\\
0.000300644971170855	-0.862245412099198\\
0.000356077513069458	-0.862178883302727\\
0.000411510054968062	-0.862112304143455\\
0.000466942596866666	-0.862045674604484\\
0.000532203611459508	-0.86196716671027\\
0.000597464626052351	-0.861888588921701\\
0.000662725640645194	-0.861809941195642\\
0.000727986655238036	-0.861731223483987\\
0.000793247669830879	-0.861652435733974\\
0.000858508684423721	-0.861573577888479\\
0.000923769699016564	-0.86149464988631\\
0.00103873397645211	-0.861355438892465\\
0.00115369825388766	-0.861216009599031\\
0.0012686625313232	-0.86107636159419\\
0.00138362680875875	-0.860936494439356\\
0.00149859108619429	-0.860796407672399\\
0.00161355536362984	-0.86065610081059\\
0.00172851964106539	-0.86051557335327\\
0.00187553903271815	-0.86033554082768\\
0.00202255842437091	-0.86015514558458\\
0.00216957781602367	-0.859974386484139\\
0.00231659720767643	-0.859793262363162\\
0.00246361659932919	-0.859611772040181\\
0.00261063599098195	-0.859429914319979\\
0.00275765538263471	-0.859247687997556\\
0.00294435013522383	-0.859015752065496\\
0.00313104488781295	-0.858783217300285\\
0.00331773964040207	-0.858550081214309\\
0.00350443439299119	-0.858316341326818\\
0.0036911291455803	-0.85808199516992\\
0.00387782389816942	-0.857847040293665\\
0.00406451865075854	-0.857611474270292\\
0.00430327930002514	-0.857309318615332\\
0.00454203994929175	-0.857006154522313\\
0.00478080059855835	-0.856701977119296\\
0.00501956124782495	-0.856396781623835\\
0.00525832189709155	-0.85609056334729\\
0.00549708254635816	-0.855783317698078\\
0.00573584319562476	-0.855475040183918\\
0.00604854021380708	-0.855069731243114\\
0.0063612372319894	-0.854662635293072\\
0.00667393425017172	-0.854253743029484\\
0.00698663126835404	-0.853843045444293\\
0.00729932828653637	-0.853430533822527\\
0.00761202530471869	-0.853016199738326\\
0.00792472232290101	-0.852600035050163\\
0.00835348610700806	-0.852026408252941\\
0.0087822498911151	-0.851449305490574\\
0.00921101367522215	-0.850868708258074\\
0.0096397774593292	-0.850284598953942\\
0.0100685412434362	-0.849696960849321\\
0.0104973050275433	-0.849105778057765\\
0.0109260688116503	-0.848511035505453\\
0.011573730047872	-0.847605876858945\\
0.0122213912840937	-0.846692517431237\\
0.0128690525203154	-0.845770914926859\\
0.0135167137565371	-0.844841030548542\\
0.0141643749927588	-0.843902828803494\\
0.0148120362289805	-0.842956277325619\\
0.0158195685650938	-0.841467110815884\\
0.0168271009012072	-0.839957569676427\\
0.0178346332373206	-0.838427571398246\\
0.0188421655734339	-0.836877047065274\\
0.0198496979095473	-0.835305940195854\\
0.0208572302456606	-0.833714205719891\\
0.0220008301472745	-0.831882469259775\\
0.0231444300488884	-0.830024077371527\\
0.0242880299505023	-0.828139008474769\\
0.0254316298521162	-0.826227253915168\\
0.0265752297537301	-0.824288816748357\\
0.027718829655344	-0.822323710581077\\
0.0288624295569579	-0.820331958515506\\
0.0304835349977048	-0.817462964350333\\
0.0321046404384517	-0.814540608314812\\
0.0337257458791986	-0.811565033870644\\
0.0353468513199455	-0.80853640798017\\
0.0369679567606924	-0.805454917940317\\
0.0385890622014393	-0.802320768600134\\
0.0402101676421862	-0.799134179960163\\
0.0420399985456778	-0.795474592906061\\
0.0438698294491695	-0.791748843333034\\
0.0456996603526611	-0.787957299243092\\
0.0475294912561527	-0.78410034186828\\
0.0493593221596444	-0.780178363643007\\
0.051189153063136	-0.776191766464988\\
0.0530189839666277	-0.772140960220188\\
0.0552171038213478	-0.767190544345571\\
0.0574152236760679	-0.762148806142083\\
0.059613343530788	-0.757016488128485\\
0.0618114633855081	-0.751794340657464\\
0.0640095832402282	-0.746483120415933\\
0.0662077030949483	-0.741083589160981\\
0.0684058229496684	-0.735596512677508\\
0.0711159409684139	-0.728711989752476\\
0.0738260589871594	-0.721697004231373\\
0.0765361770059048	-0.714553006536067\\
0.0792462950246503	-0.707281450026229\\
0.0819564130433958	-0.699883789997711\\
0.0846665310621412	-0.69236148285616\\
0.0873766490808867	-0.684715985459361\\
0.0908656754953445	-0.674694077926976\\
0.0943547019098024	-0.664473514404048\\
0.0978437283242602	-0.654057399775458\\
0.101332754738718	-0.643448835731911\\
0.104821781153176	-0.63265092009756\\
0.108310807567634	-0.621666746261253\\
0.111799833982092	-0.610499402717085\\
0.116704214488559	-0.594498148065353\\
0.121608594995026	-0.578149608665149\\
0.126512975501493	-0.561462306188324\\
0.13141735600796	-0.544444738101089\\
0.136321736514428	-0.527105376562376\\
0.141226117020895	-0.509452667338824\\
0.146130497527362	-0.491495028773097\\
0.154167762509756	-0.461428841367759\\
0.16220502749215	-0.430602997575347\\
0.170242292474544	-0.399054029647183\\
0.178279557456937	-0.366818235706594\\
0.186316822439331	-0.333931668411295\\
0.194354087421725	-0.300430123954032\\
0.208509335515431	-0.240039939001205\\
0.222664583609137	-0.178044190333432\\
0.236819831702842	-0.114631956737703\\
0.250975079796548	-0.0499892340797854\\
0.265130327890254	0.0157012595352262\\
0.279285575983959	0.0822602700877471\\
0.296335646065869	0.163335275080129\\
0.313385716147779	0.245116321875469\\
0.330435786229689	0.327312489671583\\
0.347485856311599	0.40964181405769\\
0.364535926393509	0.491831690287228\\
0.381585996475419	0.573619262891473\\
0.398636066557329	0.654751791435528\\
0.422370035404442	0.766144872665382\\
0.446104004251554	0.875182096022935\\
0.469837973098666	0.981286500148746\\
0.493571941945779	1.08392399644955\\
0.517305910792891	1.18260427282314\\
0.541039879640004	1.27688145966866\\
0.564773848487116	1.36635454828117\\
0.591475198996247	1.46083003560718\\
0.618176549505377	1.54833682739463\\
0.644877900014508	1.6285074477288\\
0.671579250523639	1.70104478561186\\
0.69828060103277	1.76572061333391\\
0.7249819515419	1.82237370712365\\
0.751683302051031	1.87090758427644\\
0.784632820449226	1.91956117538015\\
0.817582338847421	1.95586158783618\\
0.850531857245616	1.97999570408323\\
0.883481375643812	1.99226824053976\\
0.916430894042007	1.99309099572827\\
0.949380412440202	1.98297164543825\\
0.982329930838397	1.96250218214513\\
1.02051401380824	1.92671222608846\\
1.05869809677809	1.87905181094764\\
1.09688217974794	1.82074065904254\\
1.13506626271778	1.75305963928162\\
1.17325034568763	1.67732781951055\\
1.21143442865748	1.59488070576488\\
1.24961851162733	1.50704987372096\\
1.29405047445646	1.39980625904566\\
1.3384824372856	1.28903638470587\\
1.38291440011473	1.17662839939593\\
1.42734636294387	1.06433938147684\\
1.471778325773	0.953773603685806\\
1.51621028860214	0.84636567853176\\
1.56064225143127	0.743368587897991\\
1.60507421426041	0.645846455340981\\
1.64950617708954	0.554671740707596\\
1.69393813991868	0.470526455357579\\
1.73837010274782	0.393906988919609\\
1.78280206557695	0.325132129255434\\
1.82723402840609	0.264353824048093\\
1.87166599123522	0.211570202569646\\
1.91609795406436	0.166640372081267\\
1.96052991689349	0.129300520174426\\
2.00496187972263	0.0991808775742724\\
2.04939384255176	0.0758231200574162\\
2.09513875576048	0.0582806713344333\\
2.14088366896919	0.0467068858608002\\
2.18662858217791	0.0404229498454454\\
2.23237349538662	0.0387265252342761\\
2.27811840859533	0.0409083981508447\\
2.32386332180405	0.0462675705029492\\
2.36960823501276	0.0541246468475141\\
2.42131349170287	0.065202481866498\\
2.47301874839297	0.077780526610411\\
2.52472400508308	0.0910761370031361\\
2.57642926177318	0.104400585177333\\
2.62813451846329	0.117165489694417\\
2.67983977515339	0.128885890283692\\
2.7315450318435	0.139180166073968\\
2.7832502885336	0.147767129677379\\
2.83495554522371	0.154460747614515\\
2.88666080191381	0.15916294290767\\
2.93836605860391	0.161854892694327\\
2.99007131529402	0.162587216347158\\
3.04213549719279	0.161455576622268\\
3.09419967909155	0.158609184748039\\
3.14626386099032	0.154247260011525\\
3.19832804288909	0.148596281548206\\
3.25039222478785	0.141899678046441\\
3.30245640668662	0.134408309220611\\
3.35452058858538	0.126371896267875\\
3.41956112325278	0.115932540190471\\
3.48460165792018	0.105453965297253\\
3.54964219258758	0.0953208882932284\\
3.61468272725498	0.0858565092752688\\
3.67972326192237	0.0773164532172421\\
3.74476379658977	0.0698864867473228\\
3.80980433125717	0.0636836778863064\\
3.87484486592457	0.0587605242993838\\
3.93988540059197	0.0551114012834839\\
4.00492593525937	0.0526806504587043\\
4.06996646992676	0.0513717056026737\\
4.13500700459416	0.0510567259851573\\
4.2027738289708	0.0516244940809107\\
4.27054065334744	0.0529257816321183\\
4.33830747772408	0.054773248052264\\
4.40607430210072	0.0569844306955174\\
4.47384112647735	0.0593892189702965\\
4.54160795085399	0.0618356770258615\\
4.60937477523063	0.0641942106301601\\
4.68354868502109	0.0665517589772779\\
4.75772259481154	0.068577160721083\\
4.831896504602	0.0702050884495478\\
4.90607041439245	0.0714040278704804\\
4.98024432418291	0.0721723167847479\\
5.05441823397336	0.0725332626700474\\
5.12859214376381	0.0725296811503693\\
5.2094628947603	0.0721771619208003\\
5.29033364575678	0.0715432990491829\\
5.37120439675327	0.0707165687746798\\
5.45207514774975	0.0697828311367683\\
5.53294589874623	0.0688201358709701\\
5.61381664974272	0.0678949764616359\\
5.6946874007392	0.067059913357256\\
5.78423695233517	0.06628533822109\\
5.87378650393113	0.0656968156461764\\
5.9633360555271	0.0653037640820902\\
6.05288560712306	0.065098525267282\\
6.14243515871903	0.0650602813785408\\
6.231984710315	0.0651591541204848\\
6.32153426191096	0.065360233574788\\
6.43552860544557	0.0657069467552542\\
6.54952294898017	0.0660909587761882\\
6.66351729251478	0.0664550264879464\\
6.77751163604938	0.0667584651293236\\
6.89150597958398	0.0669780104717228\\
7.00550032311859	0.0671063560991758\\
7.11949466665319	0.0671492398634435\\
7.24203666586693	0.0671169636588491\\
7.36457866508066	0.0670283188407996\\
7.4871206642944	0.0669104639499674\\
7.60966266350814	0.0667877888676589\\
7.73220466272187	0.0666789991834724\\
7.85474666193561	0.0665959345893556\\
7.98431875509124	0.0665416989105596\\
8.11389084824686	0.0665211211879031\\
8.24346294140249	0.066528304238477\\
8.37303503455812	0.0665545762087856\\
8.50260712771375	0.0665906184859253\\
8.63217922086938	0.0666280957277987\\
8.81321192723325	0.0666722087052376\\
8.99424463359713	0.0666980986084695\\
9.17527733996101	0.0667041586115493\\
9.35631004632488	0.0666951861465508\\
9.53734275268876	0.0666789821776046\\
9.71837545905263	0.0666629797181919\\
9.94135056481037	0.0666500730052908\\
10.1643256705681	0.0666472447052548\\
10.3873007763258	0.0666520252169489\\
10.6102758820836	0.0666594344323865\\
10.8332509878413	0.0666650694628537\\
11.1332509878413	0.0666671751639072\\
11.4332509878413	0.0666638650186098\\
11.7332509878413	0.0666594660244071\\
12.0332509878413	0.066657740619636\\
12.3332509878413	0.0666593417148622\\
12.6332509878413	0.0666622111019546\\
12.9332509878413	0.0666637606418048\\
13.2332509878413	0.0666635364395987\\
13.5332509878413	0.0666622614156847\\
13.8332509878413	0.0666609673555847\\
14.1332509878413	0.066660539332779\\
14.4332509878413	0.0666606614199861\\
14.7332509878413	0.0666609910868043\\
15	0.0666612798124921\\
};
\end{axis}
\end{tikzpicture}%

%% file: plots/transient_dyn_y_eta.tex
%
%
\definecolor{mycolor1}{RGB}{0,81,158}%
\definecolor{mycolor2}{RGB}{0, 190, 255}
\definecolor{usGray}{RGB}{139,143,148}%
\begin{tikzpicture}

\begin{axis}[%
width=1.15in,
height=0.8in,
at={(0.758in,0.481in)},
scale only axis,
xmin=0,
xmax=8,
ymin=-2,
ymax=3.5,
axis background/.style={fill=white},
xmajorgrids,
ymajorgrids,
ylabel style = {yshift=-3.5mm},
ylabel={$y$},
xlabel={$t$},
xlabel style = {yshift=1.5mm},
]

\addplot [color=gray, dashed, thick]
  table[row sep=crcr]{%
0	0.576082308650491\\
20	0.576082308650491\\
};
\addplot [color=gray, dashed, thick]
  table[row sep=crcr]{%
0	0.0627947023452091\\
20	0.0627947023452091\\
};
\addplot [color=mycolor1, thick]
  table[row sep=crcr]{%
0	-3.48914335047704\\
4.75332248397097e-08	-3.48914181185225\\
9.50664496794195e-08	-3.48914027321449\\
1.42599674519129e-07	-3.48913873456349\\
3.49255436102956e-07	-3.4891320449706\\
5.55911197686783e-07	-3.48912535513197\\
7.6256695927061e-07	-3.48911866504369\\
9.69222720854437e-07	-3.4891119747044\\
1.89468293915099e-06	-3.48908201047876\\
2.82014315744754e-06	-3.48905204121504\\
3.7456033757441e-06	-3.4890220669137\\
4.67106359404065e-06	-3.48899208757551\\
1.39256657770062e-05	-3.48869201744488\\
2.31802679599717e-05	-3.48839144450312\\
3.24348701429372e-05	-3.48809036947783\\
4.16894723259028e-05	-3.48778879313034\\
7.28225490253081e-05	-3.48677060076229\\
0.000103955625724713	-3.4857467724168\\
0.000135088702424119	-3.48471733712403\\
0.000166221779123524	-3.48368232379036\\
0.000197354855822929	-3.48264176120411\\
0.000338688446366252	-3.47784897077082\\
0.000480022036909575	-3.47294507988519\\
0.000621355627452898	-3.46793273268022\\
0.000762689217996221	-3.46281454878818\\
0.000904022808539544	-3.45759312150093\\
0.00104535639908287	-3.45227101622126\\
0.00135843550409311	-3.44013573255835\\
0.00167151460910336	-3.42754591313772\\
0.00198459371411361	-3.41452795461349\\
0.00229767281912385	-3.40110754711468\\
0.0026107519241341	-3.38730966058559\\
0.00292383102914434	-3.37315853522558\\
0.00336685904084663	-3.35257571687679\\
0.00380988705254891	-3.33139657013279\\
0.0042529150642512	-3.30968216865519\\
0.00469594307595348	-3.2874906013995\\
0.00513897108765577	-3.26487698761561\\
0.00558199909935805	-3.24189350298083\\
0.00602502711106034	-3.21858941690199\\
0.00666766713872375	-3.18430978807192\\
0.00731030716638716	-3.14958392037654\\
0.00795294719405057	-3.11453106473246\\
0.00859558722171398	-3.07925897634724\\
0.00923822724937739	-3.04386444011046\\
0.0098808672770408	-3.00843382454247\\
0.0105235073047042	-2.9730436601251\\
0.0112531065518746	-2.93299855498829\\
0.0119827057990449	-2.89317709132222\\
0.0127123050462153	-2.85365109824973\\
0.0134419042933856	-2.8144803243831\\
0.0141715035405559	-2.77571353869788\\
0.0149011027877263	-2.73738959699756\\
0.0156307020348966	-2.69953846789428\\
0.0165419803048852	-2.65295889745852\\
0.0174532585748738	-2.6071787462427\\
0.0183645368448624	-2.56221155790755\\
0.019275815114851	-2.51805888929636\\
0.0201870933848395	-2.47471248011307\\
0.0210983716548281	-2.4321562073239\\
0.0220096499248167	-2.39036782351166\\
0.0229492275691682	-2.3480573146712\\
0.0238888052135198	-2.30650151896143\\
0.0248283828578714	-2.26566501155507\\
0.0257679605022229	-2.22551110612477\\
0.0267075381465745	-2.18600278369072\\
0.027647115790926	-2.14710345168294\\
0.0285866934352776	-2.10877754833795\\
0.0295310460883144	-2.07080029682518\\
0.0304753987413512	-2.03333520105326\\
0.031419751394388	-1.99635168965237\\
0.0323641040474248	-1.95982148689907\\
0.0333084567004616	-1.92371869691925\\
0.0342528093534985	-1.88801982434688\\
0.0351971620065353	-1.85270374323123\\
0.0362881467364139	-1.81235605880271\\
0.0373791314662925	-1.77246881577514\\
0.0384701161961712	-1.73302025552407\\
0.0395611009260498	-1.69399195542335\\
0.0406520856559284	-1.65536850639703\\
0.0417430703858071	-1.61713718594847\\
0.0428340551156857	-1.5792876320166\\
0.0442254555563078	-1.53155669372396\\
0.0456168559969298	-1.48442009399082\\
0.0470082564375519	-1.43786826306514\\
0.048399656878174	-1.39189447322853\\
0.049791057318796	-1.34649412140454\\
0.0511824577594181	-1.30166411828391\\
0.0525738582000402	-1.25740237886701\\
0.0542922166946045	-1.20352187560699\\
0.0560105751891688	-1.15050352612563\\
0.0577289336837332	-1.0983449633369\\
0.0594472921782975	-1.04704333167206\\
0.0611656506728618	-0.996595017914027\\
0.0628840091674261	-0.946995504299083\\
0.0646023676619905	-0.898239326136368\\
0.0665178343926226	-0.844876625744143\\
0.0684333011232548	-0.792543877794125\\
0.070348767853887	-0.741229993468247\\
0.0722642345845191	-0.690923028175023\\
0.0741797013151513	-0.641610383973227\\
0.0760951680457834	-0.593278993798767\\
0.0780106347764156	-0.545915485507851\\
0.0809159074846546	-0.475894166197677\\
0.0838211801928937	-0.408020888618437\\
0.0867264529011328	-0.342248547644982\\
0.0896317256093718	-0.278530664532781\\
0.0925369983176109	-0.216821534090259\\
0.09544227102585	-0.157076282207816\\
0.098347543734089	-0.099250867385833\\
0.102048498309089	-0.0283007450569961\\
0.10574945288409	0.0396918190615874\\
0.10945040745909	0.104812835401062\\
0.11315136203409	0.167146750006639\\
0.116852316609091	0.226776525402582\\
0.120553271184091	0.283783663096081\\
0.125217947914857	0.352023171985149\\
0.129882624645623	0.416380314577594\\
0.13454730137639	0.477008299577467\\
0.139211978107156	0.534056584256495\\
0.143876654837922	0.587670849668782\\
0.148541331568688	0.637992991233898\\
0.155766865156762	0.709763573333436\\
0.162992398744837	0.774461046804521\\
0.170217932332911	0.832559518054903\\
0.177443465920985	0.884511918687561\\
0.184668999509059	0.930750336859128\\
0.191894533097134	0.971686396504187\\
0.200982798135329	1.0162485672592\\
0.210071063173525	1.0537754332334\\
0.219159328211721	1.08495491780741\\
0.228247593249916	1.11042917145673\\
0.237335858288112	1.13079629942555\\
0.246424123326308	1.14661215550076\\
0.255512388364503	1.15839217662762\\
0.26565527192151	1.16735722457082\\
0.275798155478518	1.17249750893863\\
0.285941039035525	1.17436828148891\\
0.296083922592532	1.17347415000922\\
0.306226806149539	1.17027216051773\\
0.316369689706546	1.16517480310105\\
0.326512573263554	1.15855293126928\\
0.337855175130169	1.14974959250173\\
0.349197776996785	1.13986517076896\\
0.3605403788634	1.12925889143919\\
0.371882980730016	1.11824240987056\\
0.383225582596632	1.10708391451394\\
0.394568184463247	1.09601200804696\\
0.405910786329863	1.08521936175146\\
0.420005236250031	1.0724358535233\\
0.4340996861702	1.06057383699237\\
0.448194136090369	1.04981732527422\\
0.462288586010537	1.04029727106871\\
0.476383035930706	1.03209887268755\\
0.490477485850874	1.02526823062481\\
0.504571935771043	1.01981835862489\\
0.520134249913814	1.01538643597292\\
0.535696564056585	1.01256756214591\\
0.551258878199356	1.011282879822\\
0.566821192342126	1.01143262339553\\
0.582383506484897	1.01290145664987\\
0.597945820627668	1.01556310536716\\
0.613508134770439	1.01928432225309\\
0.631434358459007	1.0247056668982\\
0.649360582147575	1.03114105514191\\
0.667286805836143	1.03838360071226\\
0.685213029524711	1.04623338177835\\
0.703139253213279	1.05450032200757\\
0.721065476901846	1.06300638508374\\
0.738991700590414	1.07158713307773\\
0.765535105707778	1.0841141421114\\
0.792078510825141	1.09605425543369\\
0.818621915942505	1.10705608137849\\
0.845165321059868	1.11683962152881\\
0.871708726177231	1.12519212490132\\
0.898252131294595	1.1319628927011\\
0.924795536411958	1.13705718902402\\
0.951338941529322	1.14042949626664\\
0.977882346646685	1.14207657830579\\
1.00442575176405	1.14203075582805\\
1.03096915688141	1.14035356615449\\
1.05751256199878	1.13712984496964\\
1.08405596711614	1.13246227741718\\
1.1105993722335	1.12646649537487\\
1.13714277735087	1.11926676831781\\
1.16368618246823	1.11099227893313\\
1.19449073706786	1.10021447902426\\
1.22529529166749	1.08836805989579\\
1.25609984626713	1.07564997552053\\
1.28690440086676	1.06224723048441\\
1.31241953469769	1.05075261013187\\
1.33793466852862	1.03900031524243\\
1.36344980235955	1.02707539474759\\
1.38896493619049	1.01505575082207\\
1.41448007002142	1.00301194592105\\
1.43999520385235	0.991007208771349\\
1.47789961360104	0.973364045678889\\
1.51580402334973	0.956085376696826\\
1.55370843309842	0.939295961165198\\
1.59161284284711	0.923092153663333\\
1.6295172525958	0.907544880362154\\
1.6674216623445	0.89270267046824\\
1.70532607209319	0.878594712735354\\
1.74323048184188	0.865233839247726\\
1.78113489159057	0.852619314223082\\
1.81903930133926	0.840739364600837\\
1.85738672386942	0.829447067108337\\
1.89573414639957	0.818857311009425\\
1.93408156892973	0.808936657361051\\
1.97242899145988	0.799648115204726\\
2.01077641399004	0.790952502720732\\
2.04912383652019	0.78280957867182\\
2.08747125905035	0.775178969953389\\
2.13480783757584	0.766407121861316\\
2.18214441610134	0.758283227490896\\
2.22948099462683	0.750739107864459\\
2.27681757315233	0.74371133317451\\
2.32415415167782	0.737141688275308\\
2.37149073020332	0.730977406203547\\
2.41882730872881	0.72517119343457\\
2.48418311392144	0.717666370728334\\
2.54953891911408	0.710670051902096\\
2.61489472430671	0.704104234717249\\
2.68025052949934	0.697905628795864\\
2.74560633469198	0.692023638489229\\
2.81096213988461	0.686418447861701\\
2.87631794507724	0.6810591995248\\
2.94167375026987	0.67592229826011\\
3.00702955546251	0.670989908787084\\
3.07238536065514	0.666248683648169\\
3.1392430554008	0.661585983731335\\
3.20610075014645	0.657105089675043\\
3.27295844489211	0.652800097154238\\
3.33981613963776	0.648666433442903\\
3.40667383438342	0.644700366401738\\
3.47353152912907	0.640898636491911\\
3.54038922387473	0.637258187157056\\
3.64229405180416	0.63201278956711\\
3.74419887973359	0.627123652174539\\
3.84610370766302	0.622578268157524\\
3.94800853559245	0.618362952495176\\
4.04991336352189	0.614462982086705\\
4.15181819145132	0.610862789915155\\
4.25372301938075	0.607546217442312\\
4.39239390708575	0.603459036837581\\
4.53106479479075	0.59982525077338\\
4.66973568249575	0.596604547198532\\
4.80840657020075	0.593758368362156\\
4.94707745790576	0.591250317332878\\
5.08574834561076	0.5890463918747\\
5.22441923331576	0.587115120257708\\
5.3882501095173	0.585145428918207\\
5.55208098571885	0.583472852100266\\
5.7159118619204	0.582058737534914\\
5.87974273812194	0.580868630648063\\
6.04357361432349	0.579871951213557\\
6.20740449052503	0.579041685638398\\
6.37123536672658	0.578354063064217\\
6.5656979581504	0.577694430473042\\
6.76016054957421	0.577176364402786\\
6.95462314099803	0.576774281472199\\
7.14908573242185	0.576466596063791\\
7.34354832384566	0.576235198218074\\
7.53801091526948	0.576064970594611\\
7.73247350669329	0.575943358560562\\
7.97625128211937	0.575843897346151\\
8.22002905754545	0.575788336691604\\
8.46380683297153	0.575764041593862\\
8.70758460839761	0.575761371072956\\
8.95136238382369	0.575773063484946\\
9.19514015924976	0.575793729348618\\
9.43891793467584	0.575819439962328\\
9.73891793467584	0.575853955028601\\
10.0389179346758	0.575888452083833\\
10.3389179346758	0.575920834622795\\
10.6389179346758	0.575949978769403\\
10.9389179346758	0.57597541824806\\
11.2389179346758	0.575997103204315\\
11.5389179346758	0.576015228711558\\
11.8389179346758	0.576030124184875\\
12.1389179346758	0.57604218345925\\
12.4389179346758	0.57605181683329\\
12.7389179346759	0.57605941767313\\
13.0389179346759	0.576065343269023\\
13.3389179346759	0.576069907911218\\
13.6389179346759	0.576073382500546\\
13.9389179346759	0.576075995900066\\
14.2389179346759	0.576077937128146\\
14.5389179346759	0.57607935934853\\
14.8389179346759	0.576080385323873\\
15	0.576080807701696\\
};
\addplot [color=mycolor2, thick]
  table[row sep=crcr]{%
0	-0.862605362039618\\
4.75332248397097e-08	-0.862605305247019\\
9.50664496794195e-08	-0.862605248454358\\
1.42599674519129e-07	-0.862605191661633\\
3.49255436102956e-07	-0.862604944748514\\
5.55911197686783e-07	-0.862604697834212\\
7.6256695927061e-07	-0.862604450918706\\
9.69222720854437e-07	-0.862604204001991\\
1.89468293915099e-06	-0.862603098227502\\
2.82014315744754e-06	-0.862601992428712\\
3.7456033757441e-06	-0.86260088660561\\
4.67106359404065e-06	-0.862599780758184\\
1.39256657770062e-05	-0.862588720942434\\
2.31802679599717e-05	-0.862577658682593\\
3.24348701429372e-05	-0.862566593968804\\
4.16894723259028e-05	-0.862555526790757\\
7.28225490253081e-05	-0.862518277870887\\
0.000103955625724713	-0.862481000557289\\
0.000135088702424119	-0.862443694459003\\
0.000166221779123524	-0.862406359187986\\
0.000197354855822929	-0.862368994358996\\
0.000338688446366252	-0.86219898859354\\
0.000480022036909575	-0.862028330508734\\
0.000621355627452898	-0.861856985611313\\
0.000762689217996221	-0.861684920200823\\
0.000904022808539544	-0.861512101380431\\
0.00104535639908287	-0.861338497064592\\
0.00135843550409311	-0.860950979943627\\
0.00167151460910336	-0.860559129644569\\
0.00198459371411361	-0.860162638178592\\
0.00229767281912385	-0.8597612154871\\
0.0026107519241341	-0.859354588959517\\
0.00292383102914434	-0.858942502952589\\
0.00336685904084663	-0.858349582692531\\
0.00380988705254891	-0.857744630551912\\
0.0042529150642512	-0.857127084038239\\
0.00469594307595348	-0.856496438774907\\
0.00513897108765577	-0.855852245620296\\
0.00558199909935805	-0.855194107778394\\
0.00602502711106034	-0.854521677928721\\
0.00666766713872375	-0.853520269788113\\
0.00731030716638716	-0.852487409501443\\
0.00795294719405057	-0.851422518456443\\
0.00859558722171398	-0.850325178591732\\
0.00923822724937739	-0.849195116503667\\
0.0098808672770408	-0.848032188304754\\
0.0105235073047042	-0.846836365285833\\
0.0112531065518746	-0.845438952432581\\
0.0119827057990449	-0.843999481598282\\
0.0127123050462153	-0.842518323216588\\
0.0134419042933856	-0.840995953292698\\
0.0141715035405559	-0.839432934908757\\
0.0149011027877263	-0.837829901422526\\
0.0156307020348966	-0.836187541353151\\
0.0165419803048852	-0.834082092776164\\
0.0174532585748738	-0.831917920186329\\
0.0183645368448624	-0.829696549676611\\
0.019275815114851	-0.827419522322577\\
0.0201870933848395	-0.825088372475862\\
0.0210983716548281	-0.822704609932278\\
0.0220096499248167	-0.820269705816151\\
0.0229492275691682	-0.817707141916061\\
0.0238888052135198	-0.815093216544421\\
0.0248283828578714	-0.812429353517268\\
0.0257679605022229	-0.809716902152658\\
0.0267075381465745	-0.806957135641151\\
0.027647115790926	-0.804151251152039\\
0.0285866934352776	-0.801300371369737\\
0.0295310460883144	-0.798390724896183\\
0.0304753987413512	-0.795437683534569\\
0.031419751394388	-0.792442179137135\\
0.0323641040474248	-0.789405082629528\\
0.0333084567004616	-0.786327208885649\\
0.0342528093534985	-0.7832093217095\\
0.0351971620065353	-0.780052138803165\\
0.0362881467364139	-0.776356694030321\\
0.0373791314662925	-0.772610687303557\\
0.0384701161961712	-0.76881504536124\\
0.0395611009260498	-0.764970647243881\\
0.0406520856559284	-0.761078331913489\\
0.0417430703858071	-0.757138905093165\\
0.0428340551156857	-0.753153145326227\\
0.0442254555563078	-0.74800382834687\\
0.0456168559969298	-0.742781907606775\\
0.0470082564375519	-0.737488879732309\\
0.048399656878174	-0.732126217784249\\
0.049791057318796	-0.726695378744938\\
0.0511824577594181	-0.721197809129692\\
0.0525738582000402	-0.715634948889265\\
0.0542922166946045	-0.708676930182906\\
0.0560105751891688	-0.701624219026489\\
0.0577289336837332	-0.69447950823181\\
0.0594472921782975	-0.687245483416417\\
0.0611656506728618	-0.679924820524941\\
0.0628840091674261	-0.672520182801584\\
0.0646023676619905	-0.665034217354508\\
0.0665178343926226	-0.65659691541034\\
0.0684333011232548	-0.648065427205431\\
0.070348767853887	-0.639443317303168\\
0.0722642345845191	-0.630734105589488\\
0.0741797013151513	-0.621941264113827\\
0.0760951680457834	-0.613068214712597\\
0.0780106347764156	-0.604118327336955\\
0.0809159074846546	-0.590404390486053\\
0.0838211801928937	-0.576532644287001\\
0.0867264529011328	-0.562514145369816\\
0.0896317256093718	-0.548359671056184\\
0.0925369983176109	-0.534079724010789\\
0.09544227102585	-0.51968453694721\\
0.098347543734089	-0.505184077447399\\
0.102048498309089	-0.48657515177559\\
0.10574945288409	-0.467830594728488\\
0.10945040745909	-0.448969215363577\\
0.11315136203409	-0.430009177716233\\
0.116852316609091	-0.410968012345205\\
0.120553271184091	-0.391862627163329\\
0.125217947914857	-0.367715922253493\\
0.129882624645623	-0.343524315841289\\
0.13454730137639	-0.319317604503344\\
0.139211978107156	-0.295124144820339\\
0.143876654837922	-0.270970889863581\\
0.148541331568688	-0.246883425600658\\
0.155766865156762	-0.209758342517642\\
0.162992398744837	-0.172933042263504\\
0.170217932332911	-0.136484470389979\\
0.177443465920985	-0.100482795579103\\
0.184668999509059	-0.0649917298844485\\
0.191894533097134	-0.0300688364506443\\
0.200982798135329	0.0129711837722932\\
0.210071063173525	0.0549403729301998\\
0.219159328211721	0.0957610656454264\\
0.228247593249916	0.135367618851284\\
0.237335858288112	0.173705556740253\\
0.246424123326308	0.210730730860661\\
0.255512388364503	0.246408511217687\\
0.26565527192151	0.284603964888881\\
0.275798155478518	0.321064928969093\\
0.285941039035525	0.355778448244992\\
0.296083922592532	0.388741709554457\\
0.306226806149539	0.41996098181402\\
0.316369689706546	0.449450622431426\\
0.326512573263554	0.477232150213434\\
0.337855175130169	0.506311062215868\\
0.349197776996785	0.533335792524427\\
0.3605403788634	0.558360607169133\\
0.371882980730016	0.581445881315471\\
0.383225582596632	0.602657014259148\\
0.394568184463247	0.622063439327145\\
0.405910786329863	0.639737726707306\\
0.420005236250031	0.659399505176987\\
0.4340996861702	0.676649477281309\\
0.448194136090369	0.69163654457497\\
0.462288586010537	0.704509901441582\\
0.476383035930706	0.715417629373779\\
0.490477485850874	0.724505491449633\\
0.504571935771043	0.731915921827539\\
0.520134249913814	0.738315605762291\\
0.535696564056585	0.743018009451947\\
0.551258878199356	0.746194351715701\\
0.566821192342126	0.74800718210203\\
0.582383506484897	0.748609931441657\\
0.597945820627668	0.748146624473717\\
0.613508134770439	0.746751740154297\\
0.631434358459007	0.744152596896562\\
0.649360582147575	0.740658174162288\\
0.667286805836143	0.736427556037719\\
0.685213029524711	0.731604071096446\\
0.703139253213279	0.726315946869643\\
0.721065476901846	0.720677050077195\\
0.738991700590414	0.714787697845052\\
0.765535105707778	0.705791035151446\\
0.792078510825141	0.696672630746655\\
0.818621915942505	0.687619891233661\\
0.845165321059868	0.678776110716258\\
0.871708726177231	0.670246453770007\\
0.898252131294595	0.662103549063689\\
0.924795536411958	0.654392679242273\\
0.951338941529322	0.647136539483107\\
0.977882346646685	0.640339520007914\\
1.00442575176405	0.633991493067446\\
1.03096915688141	0.628071124290627\\
1.05751256199878	0.622548744383451\\
1.08405596711614	0.617388811473499\\
1.1105993722335	0.612551989603347\\
1.13714277735087	0.607996874222965\\
1.16368618246823	0.603681402602454\\
1.19449073706786	0.598918869909986\\
1.22529529166749	0.594360565940646\\
1.25609984626713	0.589947615333506\\
1.28690440086676	0.585626086563249\\
1.31241953469769	0.58208106225979\\
1.33793466852862	0.578540903864792\\
1.36344980235955	0.574983802785198\\
1.38896493619049	0.571390653496309\\
1.41448007002142	0.567745020064157\\
1.43999520385235	0.56403304061066\\
1.47789961360104	0.558372232951099\\
1.51580402334973	0.552511928238539\\
1.55370843309842	0.546433633690532\\
1.59161284284711	0.540127443416202\\
1.6295172525958	0.533590796980986\\
1.6674216623445	0.526827285040073\\
1.70532607209319	0.519845513862235\\
1.74323048184188	0.512658057257437\\
1.78113489159057	0.505280522316419\\
1.81903930133926	0.497730734558268\\
1.85738672386942	0.489937182676583\\
1.89573414639957	0.482008090911013\\
1.93408156892973	0.473964838434644\\
1.97242899145988	0.46582877592749\\
2.01077641399004	0.457620872674151\\
2.04912383652019	0.449361439025901\\
2.08747125905035	0.441069913759281\\
2.13480783757584	0.430817771590543\\
2.18214441610134	0.420576942469657\\
2.22948099462683	0.410376835877163\\
2.27681757315233	0.400243988698015\\
2.32415415167782	0.390202083290222\\
2.37149073020332	0.38027202018069\\
2.41882730872881	0.370472039151732\\
2.48418311392144	0.357184105346943\\
2.54953891911408	0.344208056803962\\
2.61489472430671	0.331570632214818\\
2.68025052949934	0.319292026701074\\
2.74560633469198	0.307386794692168\\
2.81096213988461	0.295864678994124\\
2.87631794507724	0.284731372203945\\
2.94167375026987	0.273989205185167\\
3.00702955546251	0.263637749714667\\
3.07238536065514	0.253674332578745\\
3.1392430554008	0.24387878717031\\
3.20610075014645	0.234478269371292\\
3.27295844489211	0.225465198728493\\
3.33981613963776	0.216831052863696\\
3.40667383438342	0.20856660037107\\
3.47353152912907	0.200662089904883\\
3.54038922387473	0.193107403237229\\
3.64229405180416	0.182241989780829\\
3.74419887973359	0.172127908873047\\
3.84610370766302	0.162727729230294\\
3.94800853559245	0.154004238964031\\
4.04991336352189	0.145920740793611\\
4.15181819145132	0.138441286482506\\
4.25372301938075	0.13153084891922\\
4.39239390708575	0.12298019580679\\
4.53106479479075	0.115338105866027\\
4.66973568249575	0.10852638911269\\
4.80840657020075	0.102471142134868\\
4.94707745790576	0.0971028790834403\\
5.08574834561076	0.0923565992190358\\
5.22441923331576	0.088171778399281\\
5.3882501095173	0.0838747573447991\\
5.55208098571885	0.080198463322068\\
5.7159118619204	0.0770661403797144\\
5.87974273812194	0.0744086892526079\\
6.04357361432349	0.072164205633702\\
6.20740449052503	0.0702774866367699\\
6.37123536672658	0.068699523228779\\
6.5656979581504	0.0671678720466025\\
6.76016054957421	0.0659473719502355\\
6.95462314099803	0.064984098297578\\
7.14908573242185	0.0642321857381623\\
7.34354832384566	0.063652833500254\\
7.53801091526948	0.0632133962903065\\
7.73247350669329	0.0628865572246626\\
7.97625128211937	0.0626014632040013\\
8.22002905754545	0.0624216523757708\\
8.46380683297153	0.0623187832402134\\
8.70758460839761	0.0622709372211654\\
8.95136238382369	0.0622613629251259\\
9.19514015924976	0.0622774158555908\\
9.43891793467584	0.0623096830580012\\
9.73891793467584	0.0623616669167868\\
10.0389179346758	0.0624190373409785\\
10.3389179346758	0.0624761279076935\\
10.6389179346758	0.0625295724042809\\
10.9389179346758	0.0625775940063127\\
11.2389179346758	0.0626194637919237\\
11.5389179346758	0.0626551144543646\\
11.8389179346758	0.0626848799382502\\
12.1389179346758	0.062709319123635\\
12.4389179346758	0.0627290926230764\\
12.7389179346759	0.0627448793477106\\
13.0389179346759	0.0627573261834602\\
13.3389179346759	0.0627670219675327\\
13.6389179346759	0.0627744860086823\\
13.9389179346759	0.0627801646483004\\
14.2389179346759	0.0627844330517418\\
14.5389179346759	0.062787600753528\\
14.8389179346759	0.0627899193175505\\
15	0.0627908860492656\\
};

\end{axis}
\end{tikzpicture}%

%% file: plots/transient_con_soft.tex
%
%
\definecolor{mycolor1}{RGB}{0,81,158}%
\definecolor{mycolor2}{RGB}{0, 190, 255}
\definecolor{usGray}{RGB}{139,143,148}%
\begin{tikzpicture}

\begin{axis}[%
    name=mainaxis, 
    at={(0.758in,1.5in)},
    width=1.15in,   
    height=0.8in,
    scale only axis,
    xmin=0,
    xmax=15,
    ymin=0 ,
    ymax=1.2,
    axis background/.style={fill=white},
    xmajorgrids,
    ymajorgrids,
    ylabel={$u$},
]
    \fill[red!80!black, opacity=0.05] (axis cs:0,1) rectangle (axis cs:18,1.2);
    \draw[red!80!black, opacity=0.5,dashed] (axis cs:0,1) -- (axis cs:18,1); 
    
    \fill[red!80!black, opacity=0.05] (axis cs:0,0.2) rectangle (axis cs:18,0);
    \draw[red!80!black,opacity=0.5,dashed] (axis cs:0,0.2) -- (axis cs:18,0.2);

    \addplot [color=gray, dashed] table[row sep=crcr]{0 0.9091\\ 20 0.9091\\};\label{tik:gray2}

    \addplot [color=mycolor1, thick] table[row sep=crcr]{
5.94903675262398e-05	5.12645528746174e-05\\
8.92455558703722e-05	8.95152258671139e-05\\
0.000119000744214505	0.000136079065262711\\
0.000148755932558637	0.000190884608888524\\
0.000178511120902769	0.000253860748042774\\
0.000330834722065931	0.000701616652245767\\
0.000483158323229092	0.00135228292487617\\
0.000635481924392254	0.00219666047222181\\
0.000787805525555415	0.00322575477345711\\
0.000940129126718577	0.00443077389371305\\
0.00140835038828275	0.00914634745466696\\
0.00187657164984693	0.01520623604648\\
0.0023447929114111	0.022389949009016\\
0.00281301417297528	0.0304941993876231\\
0.00328123543453946	0.0393325592395082\\
0.00374945669610363	0.048735008039157\\
0.00447981992343722	0.0641672594375682\\
0.00521018315077081	0.0800889830787808\\
0.00594054637810441	0.0960722592452866\\
0.006670909605438	0.111766168049736\\
0.00740127283277159	0.126889642486693\\
0.00813163606010518	0.141224364962043\\
0.00903228644112706	0.157577157375168\\
0.00993293682214894	0.172281531148125\\
0.0108335872031708	0.185227797287282\\
0.0117342375841927	0.196382262728156\\
0.0126348879652146	0.205771592182768\\
0.0135355383462365	0.213469069042934\\
0.0146992805209031	0.221088608967286\\
0.0158630226955697	0.226353265058608\\
0.0170267648702363	0.22958291631057\\
0.0181905070449029	0.23110727062834\\
0.0193542492195695	0.231248241378411\\
0.0205179913942362	0.230306505793802\\
0.0216817335689028	0.228552501336998\\
0.0231149783675184	0.225622327461463\\
0.0245482231661341	0.222179714122926\\
0.0259814679647497	0.218507231936772\\
0.0274147127633654	0.21481366750463\\
0.028847957561981	0.21124363274172\\
0.0302812023605967	0.207888017409363\\
0.0317144471592123	0.204794947765293\\
0.033650275737842	0.201058560487801\\
0.0355861043164717	0.197805076357206\\
0.0375219328951014	0.194970824517499\\
0.0394577614737311	0.192475275836922\\
0.0413935900523608	0.190236858378737\\
0.0433294186309905	0.188182506815884\\
0.0452652472096201	0.186252756691079\\
0.0475787000401237	0.184048698689238\\
0.0498921528706272	0.181910391295477\\
0.0522056057011308	0.179811383438328\\
0.0545190585316343	0.177740499440315\\
0.0568325113621379	0.175696044995416\\
0.0591459641926414	0.173681829857916\\
0.061459417023145	0.171704065221717\\
0.0648229065114633	0.168906364293713\\
0.0681863959997816	0.166212050534014\\
0.0715498854880998	0.163626092975228\\
0.0749133749764181	0.161148239254026\\
0.0782768644647364	0.158776103151736\\
0.0816403539530547	0.156506107984013\\
0.085003843441373	0.154333541099407\\
0.0910185701975383	0.150676304351784\\
0.0970332969537036	0.147290184727128\\
0.103048023709869	0.144156926982057\\
0.109062750466034	0.141262467473445\\
0.1150774772222	0.138592220489333\\
0.121092203978365	0.136130331195346\\
0.130716276413736	0.132589890255467\\
0.140340348849108	0.129496237853484\\
0.149964421284479	0.126804770853417\\
0.15958849371985	0.124474353450363\\
0.169212566155222	0.122467195324663\\
0.182262532166814	0.12020228496987\\
0.195312498178407	0.118397574291941\\
0.20836246419	0.116990510769273\\
0.221412430201593	0.11592586104169\\
0.234462396213186	0.115155604162214\\
0.247512362224779	0.114638891191766\\
0.269151649268809	0.11424716293346\\
0.290790936312839	0.114322176073251\\
0.312430223356869	0.114757948719349\\
0.334069510400898	0.115471439597259\\
0.355708797444928	0.116396956307432\\
0.377348084488958	0.117482198778994\\
0.408128360995123	0.119221773399987\\
0.438908637501287	0.121115036563489\\
0.469688914007451	0.123099669915656\\
0.500469190513616	0.125130289942586\\
0.53124946701978	0.127174449223182\\
0.562029743525944	0.129209317226316\\
0.605001962060835	0.13200450722999\\
0.647974180595726	0.134722334446758\\
0.690946399130617	0.137346163084529\\
0.732696637621808	0.139795712895583\\
0.732696637621815	0.139795712895583\\
0.733334386748233	0.139885317713283\\
0.733972135874651	0.140082063458818\\
0.73460988500107	0.140433858553414\\
0.735749614797906	0.141500308553818\\
0.736889344594742	0.143283468138026\\
0.738029074391578	0.14592614228689\\
0.739168804188414	0.149518302453796\\
0.741258966692088	0.158666236435443\\
0.743349129195762	0.171140828426552\\
0.745439291699435	0.18669944242674\\
0.747529454203109	0.204941832874387\\
0.749619616706782	0.225417760805211\\
0.752436764555991	0.255775038148634\\
0.755253912405201	0.28837629563591\\
0.75807106025441	0.322308082471068\\
0.760888208103619	0.356787628873818\\
0.763705355952828	0.391171672882696\\
0.765956147597659	0.418256821510194\\
0.765956147597666	0.418256821510278\\
0.766553424930181	0.425325987642265\\
0.767150702262695	0.43234986035764\\
0.76774797959521	0.439324494004261\\
0.770352911932362	0.469119444130871\\
0.772957844269515	0.497760583532321\\
0.775562776606667	0.525086816305984\\
0.77816770894382	0.550977724434039\\
0.781599780219385	0.582781565525418\\
0.78503185149495	0.611885069072839\\
0.788463922770516	0.638308878196046\\
0.791895994046081	0.662152437301642\\
0.795328065321646	0.683559255689819\\
0.801128551274248	0.714673486581756\\
0.806929037226851	0.740159506377223\\
0.812729523179453	0.76091041148979\\
0.818530009132055	0.777793702427366\\
0.824330495084657	0.791593850107227\\
0.831665342931552	0.805673663871482\\
0.839000190778447	0.817012859461226\\
0.839167188149244	0.817246253393082\\
0.839167188149251	0.817246253393092\\
0.839954494843681	0.818268818260165\\
0.840741801538112	0.819157589105651\\
0.841529108232543	0.819867016996362\\
0.842946191970531	0.820671279621718\\
0.844363275708518	0.820749805420171\\
0.845780359446506	0.820035547954698\\
0.847197443184494	0.818537027205677\\
0.849673140830041	0.814318725923222\\
0.852148838475588	0.808483331172652\\
0.854624536121135	0.801600879873794\\
0.857100233766682	0.794195957547446\\
0.860733162433376	0.783145430654764\\
0.864366091100069	0.772625246896421\\
0.867999019766763	0.762962128568367\\
0.871631948433457	0.754100311929924\\
0.875264877100151	0.745844044018139\\
0.880598683728518	0.734519317191713\\
0.885932490356885	0.723927246491696\\
0.891266296985251	0.714007940640104\\
0.896600103613618	0.704774193791497\\
0.901933910241985	0.696218450327601\\
0.914633337161827	0.678289958441805\\
0.927332764081668	0.663402662383998\\
0.94003219100151	0.651108726096151\\
0.952731617921352	0.641088316722967\\
0.965431044841194	0.633060055770007\\
0.980471820521895	0.625744471994258\\
0.995512596202597	0.62042956135459\\
1.0105533718833	0.616809679285582\\
1.025594147564	0.614630306713919\\
1.0406349232447	0.613649845361286\\
1.0556756989254	0.613662837846129\\
1.08047763748093	0.61544593690835\\
1.10527957603646	0.618933586970519\\
1.13008151459198	0.623664660704176\\
1.15488345314751	0.629287744698222\\
1.17968539170304	0.635550708035296\\
1.20448733025857	0.64225483395778\\
1.21576320689173	0.645403649372267\\
1.21576320689175	0.64540364937227\\
1.2217942941732	0.647110510899116\\
1.22782538145466	0.648828407734777\\
1.23283413053473	0.650262328269942\\
1.2378428796148	0.651702405076706\\
1.24285162869488	0.65314775798721\\
1.264314550307	0.659377183731677\\
1.28577747191913	0.665646068686109\\
1.30724039353125	0.671919422551738\\
1.32870331514337	0.6781667180705\\
1.36372040693861	0.688234142944939\\
1.39873749873386	0.698095489247766\\
1.42748173839686	0.706001395799692\\
1.45622597805985	0.713712806479686\\
1.48497021772285	0.721217199899383\\
1.51371445738585	0.728507046602072\\
1.57893541937361	0.744238360034259\\
1.64415638136137	0.758836605220758\\
1.70937734334913	0.77232857161939\\
1.77459830533689	0.784765291962772\\
1.83981926732465	0.796207049660977\\
1.93263478168043	0.810897541307297\\
2.0254502960362	0.823887087961802\\
2.11826581039198	0.835360448833676\\
2.21108132474776	0.845494548397701\\
2.30389683910354	0.854455405602579\\
2.39671235345932	0.862396123810782\\
2.51084280561197	0.870965851491633\\
2.62497325776462	0.87843128105883\\
2.73910370991727	0.884989242957303\\
2.85323416206992	0.89080635486946\\
2.96736461422257	0.896021804577401\\
3.08149506637522	0.900750381574369\\
3.24562182862528	0.906881528535452\\
3.40974859087534	0.912409155262406\\
3.57387535312541	0.917483136146293\\
3.73800211537547	0.92220687003241\\
3.90212887762554	0.92664970951788\\
4.0662556398756	0.930856389305493\\
4.25489403319794	0.935435304464715\\
4.44353242652027	0.939762258710573\\
4.6321708198426	0.943849760728437\\
4.82080921316494	0.947705306866119\\
5.00944760648727	0.951334449038612\\
5.19808599980961	0.954742416822977\\
5.4600657467971	0.959119949335044\\
5.72204549378459	0.963101941346338\\
5.98402524077208	0.966712292504898\\
6.24600498775956	0.969978139809001\\
6.50798473474705	0.972927724886311\\
6.76996448173454	0.97558885803512\\
7.12996448173454	0.978823218565893\\
7.48996448173455	0.981626868788343\\
7.84996448173455	0.98405756613457\\
8.20996448173454	0.986165683341198\\
8.56996448173454	0.987994557835314\\
8.92996448173454	0.989581425176752\\
9.28996448173454	0.990958411041998\\
9.64996448173454	0.992153334005337\\
10.0099644817345	0.993190448666251\\
10.3699644817345	0.99409063384844\\
10.4436607789625	0.994259710258543\\
10.4436607789626	0.99425971025854\\
10.4567102670559	0.990931287612234\\
10.4697597551492	0.984763770533323\\
10.4828092432425	0.977717464637447\\
10.5004271336021	0.969080309575755\\
10.5180450239618	0.962920795842283\\
10.5356629143214	0.959570721573762\\
10.553280804681	0.958109090261321\\
10.5814887693115	0.956958674094093\\
10.609696733942	0.954667554152181\\
10.6379046985725	0.951250389762663\\
10.666112663203	0.94977409193929\\
10.6943206278335	0.952396572208087\\
10.7299519988119	0.953059852547593\\
10.7655833697903	0.94875395136658\\
10.8012147407688	0.944879986606367\\
10.8368461117472	0.944208648158369\\
10.8724774827256	0.944272387767931\\
10.9131109509035	0.942920934986862\\
10.9537444190815	0.940876794520609\\
10.9943778872594	0.939443260799329\\
11.0350113554373	0.938553348779809\\
11.0756448236152	0.93760829512539\\
11.1581077821789	0.935489914734731\\
11.2405707407426	0.933599863699477\\
11.3230336993063	0.931994201835135\\
11.40549665787	0.93057047519229\\
11.5213299583129	0.928791435193347\\
11.6371632587558	0.927221611590134\\
11.7529965591986	0.925822067436638\\
11.8688298596415	0.924545080132905\\
11.9846631600843	0.923358817458087\\
12.1698835346324	0.921617680281642\\
12.3551039091805	0.920029954873876\\
12.3862919288963	0.919776615034596\\
12.3862919288965	0.919776615034603\\
12.4731829173858	0.91910034794747\\
12.5600739058752	0.918454693586231\\
12.6469648943645	0.917840533334502\\
12.8322381481414	0.916623933358768\\
13.0175114019183	0.915543748547616\\
13.2027846556951	0.914600603773685\\
13.388057909472	0.91378859644416\\
13.6921006892173	0.912712062720782\\
13.9961434689627	0.911903482846956\\
14.300186248708	0.911299020197379\\
14.6042290284533	0.910842002742566\\
14.9082718081986	0.910491069691794\\
15.2682718081986	0.910178241305946\\
15.6282718081986	0.909950118391393\\
15.9882718081986	0.909787902419922\\
16.3482718081986	0.909676098458558\\
16.7082718081986	0.909600868714442\\
17.0682718081986	0.909550952360114\\
17.4282718081986	0.909518423192058\\
17.7882718081986	0.909498219182499\\
18	0.90949062743551\\ 
    };
    \coordinate (zoomBoxBL) at (axis cs:10, 0.95);
    \coordinate (zoomBoxBR) at (axis cs:11, 0.95);
    \coordinate (zoomBoxTR) at (axis cs:11, 1.05);
    \coordinate (zoomBoxTL) at (axis cs:10, 1.05);
    \draw[black, line width=0.01pt] (zoomBoxBL) rectangle (zoomBoxTR);

\end{axis}

\begin{axis}[%
    scale only axis,
    at={(mainaxis.south west)},
    width=0.28in,
    height=0.2in,
    xshift=0.5in, 
    yshift=0.3in, 
    xmin=10.4,    
    xmax=10.53,    
    ymin=0.94,   
    ymax=1.02,
    axis background/.style={fill=white},
    xtick=\empty,
    ytick=\empty,
    xlabel={},
    ylabel={},
    xmajorgrids,
    ymajorgrids,
    grid style={line width=0.05pt, draw=gray!20}
]
    \addplot [color=mycolor1, thick] table[row sep=crcr]{
7.48996448173455	0.981626868788343\\
7.84996448173455	0.98405756613457\\
8.20996448173454	0.986165683341198\\
8.56996448173454	0.987994557835314\\
8.92996448173454	0.989581425176752\\
9.28996448173454	0.990958411041998\\
9.64996448173454	0.992153334005337\\
10.0099644817345	0.993190448666251\\
10.3699644817345	0.99409063384844\\
10.4436607789625	0.994259710258543\\
10.4436607789626	0.99425971025854\\
10.4567102670559	0.990931287612234\\
10.4697597551492	0.984763770533323\\
10.4828092432425	0.977717464637447\\
10.5004271336021	0.969080309575755\\
10.5180450239618	0.962920795842283\\
10.5356629143214	0.959570721573762\\
10.553280804681	0.958109090261321\\
10.5814887693115	0.956958674094093\\
10.609696733942	0.954667554152181\\
10.6379046985725	0.951250389762663\\
10.666112663203	0.94977409193929\\
10.6943206278335	0.952396572208087\\
10.7299519988119	0.953059852547593\\
10.7655833697903	0.94875395136658\\
10.8012147407688	0.944879986606367\\
10.8368461117472	0.944208648158369\\
10.8724774827256	0.944272387767931\\
10.9131109509035	0.942920934986862\\
10.9537444190815	0.940876794520609\\
10.9943778872594	0.939443260799329\\
11.0350113554373	0.938553348779809\\
11.0756448236152	0.93760829512539\\
11.1581077821789	0.935489914734731\\
11.2405707407426	0.933599863699477\\
11.3230336993063	0.931994201835135\\
11.40549665787	0.93057047519229\\
11.5213299583129	0.928791435193347\\
11.6371632587558	0.927221611590134\\
11.7529965591986	0.925822067436638\\
    };

    \coordinate (insetBR) at (rel axis cs:1,0);
    \coordinate (insetTL) at (rel axis cs:0,1);

\end{axis}

Draws dashed indicator lines from the main axis target region to the inset limits
\draw[black, dotted, thin] (zoomBoxBR) -- (insetBR);
\draw[black, dotted, thin] (zoomBoxTL) -- (insetTL);

\begin{axis}[%
width=1.15in,
height=0.8in,
at={(0.758in,0.481in)},
scale only axis,
xmin=0,
xmax=17,
ymin=-1,
ymax=2,
axis background/.style={fill=white},
xmajorgrids,
ymajorgrids,
ylabel style = {yshift=-3.5mm},
ylabel={$y$},
xlabel={$t$},
xlabel style = {yshift=1.5mm},
legend style={
    at={(0.97,0.03)},
    anchor=south east,
    nodes={scale=0.65, transform shape}
  }
]

\addplot [color=gray, dashed, forget plot]
  table[row sep=crcr]{%
0	0.852272727272727\\
18	0.852272727272727\\
};

\addplot [color=gray, dashed, forget plot]
  table[row sep=crcr]{%
0	0.909090909090909\\
18	0.909090909090909\\
};

\addplot [color=mycolor1, thick]
  table[row sep=crcr]{%
2.97351791821074e-05	-3.48849130613148\\
5.94903675262398e-05	-3.48783889889966\\
8.92455558703722e-05	-3.48718656883422\\
0.000119000744214505	-3.48653431717222\\
0.000148755932558637	-3.48588214513832\\
0.000178511120902769	-3.48523005394513\\
0.000330834722065931	-3.48189316025343\\
0.000483158323229092	-3.47855857244443\\
0.000635481924392254	-3.47522644065488\\
0.000787805525555415	-3.47189690773907\\
0.000940129126718577	-3.46857010949185\\
0.00140835038828275	-3.45836245834657\\
0.00187657164984693	-3.4481851970856\\
0.0023447929114111	-3.43804115076264\\
0.00281301417297528	-3.42793264466\\
0.00328123543453946	-3.41786154579975\\
0.00374945669610363	-3.40782930110913\\
0.00447981992343722	-3.39226032723038\\
0.00521018315077081	-3.37679081357667\\
0.00594054637810441	-3.36142142629759\\
0.006670909605438	-3.34615142521693\\
0.00740127283277159	-3.33097893983204\\
0.00813163606010518	-3.31590119935544\\
0.00903228644112706	-3.2974333276758\\
0.00993293682214894	-3.27909676235409\\
0.0108335872031708	-3.26088332865147\\
0.0117342375841927	-3.24278459007623\\
0.0126348879652146	-3.22479214492906\\
0.0135355383462365	-3.20689785082489\\
0.0146992805209031	-3.18390926835802\\
0.0158630226955697	-3.16105653623878\\
0.0170267648702363	-3.1383263392619\\
0.0181905070449029	-3.11570734431233\\
0.0193542492195695	-3.09319017129282\\
0.0205179913942362	-3.07076727603922\\
0.0216817335689028	-3.04843279587626\\
0.0231149783675184	-3.02104120693353\\
0.0245482231661341	-2.99377171729225\\
0.0259814679647497	-2.96662117336927\\
0.0274147127633654	-2.93958816433715\\
0.028847957561981	-2.91267256265975\\
0.0302812023605967	-2.88587509072395\\
0.0317144471592123	-2.8591969470132\\
0.033650275737842	-2.82335575882705\\
0.0355861043164717	-2.78773829945471\\
0.0375219328951014	-2.75234749465009\\
0.0394577614737311	-2.71718554882598\\
0.0413935900523608	-2.68225389908574\\
0.0433294186309905	-2.64755323059734\\
0.0452652472096201	-2.61308358053191\\
0.0475787000401237	-2.57219216664606\\
0.0498921528706272	-2.53162841899928\\
0.0522056057011308	-2.49139045290697\\
0.0545190585316343	-2.45147616752971\\
0.0568325113621379	-2.41188337797923\\
0.0591459641926414	-2.37260994784372\\
0.061459417023145	-2.33365384301374\\
0.0648229065114633	-2.27757826540887\\
0.0681863959997816	-2.22216350129465\\
0.0715498854880998	-2.1674039614168\\
0.0749133749764181	-2.11329422117995\\
0.0782768644647364	-2.05982890569176\\
0.0816403539530547	-2.00700256891247\\
0.085003843441373	-1.95480967381973\\
0.0910185701975383	-1.86303806082619\\
0.0970332969537036	-1.7732427130673\\
0.103048023709869	-1.68539193481049\\
0.109062750466034	-1.59945410065083\\
0.1150774772222	-1.51539768605876\\
0.121092203978365	-1.43319126767399\\
0.130716276413736	-1.30542433605916\\
0.140340348849108	-1.18218645724234\\
0.149964421284479	-1.06335156830064\\
0.15958849371985	-0.948795215603366\\
0.169212566155222	-0.838394692140627\\
0.182262532166814	-0.69511723480375\\
0.195312498178407	-0.558963590320648\\
0.20836246419	-0.429645836607501\\
0.221412430201593	-0.306883731913471\\
0.234462396213186	-0.19040487853533\\
0.247512362224779	-0.0799449049939092\\
0.269151649268809	0.0906383377413278\\
0.290790936312839	0.246491861278924\\
0.312430223356869	0.388672676563214\\
0.334069510400898	0.518177260970829\\
0.355708797444928	0.635943275064043\\
0.377348084488958	0.74285166152855\\
0.408128360995123	0.87782912750996\\
0.438908637501287	0.994700908018703\\
0.469688914007451	1.09546593669438\\
0.500469190513616	1.18193884540968\\
0.53124946701978	1.25576409468374\\
0.562029743525944	1.31842913514397\\
0.605001962060835	1.38977090457918\\
0.647974180595726	1.44518930668337\\
0.690946399130617	1.48746583432199\\
0.732696637621808	1.51822486462916\\
0.732696637621815	1.51822486462916\\
0.733334386748233	1.51862533644518\\
0.733972135874651	1.51902349075773\\
0.73460988500107	1.51941883064684\\
0.735749614797906	1.52011764112927\\
0.736889344594742	1.52080305832952\\
0.738029074391578	1.52147041638541\\
0.739168804188414	1.52211457898652\\
0.741258966692088	1.52321846265683\\
0.743349129195762	1.5241934259238\\
0.745439291699435	1.52500682671582\\
0.747529454203109	1.52562964673768\\
0.749619616706782	1.52603783353624\\
0.752436764555991	1.5262167586464\\
0.755253912405201	1.52593900249207\\
0.75807106025441	1.52518503088364\\
0.760888208103619	1.523947866664\\
0.763705355952828	1.52223071146915\\
0.765956147597659	1.52052258981038\\
0.765956147597666	1.52052258981038\\
0.766553424930181	1.52001258755811\\
0.767150702262695	1.51948292327942\\
0.76774797959521	1.51893337320963\\
0.770352911932362	1.51632638302502\\
0.772957844269515	1.51336965425752\\
0.775562776606667	1.51007673650903\\
0.77816770894382	1.50646724905099\\
0.781599780219385	1.50126378867242\\
0.78503185149495	1.49560051603568\\
0.788463922770516	1.48953208987594\\
0.791895994046081	1.48311074416554\\
0.795328065321646	1.47638506803081\\
0.801128551274248	1.46444762609461\\
0.806929037226851	1.45196028025185\\
0.812729523179453	1.43908796615116\\
0.818530009132055	1.42596532973129\\
0.824330495084657	1.41269952042613\\
0.831665342931552	1.39584930948186\\
0.839000190778447	1.37902752524348\\
0.839167188149244	1.37864559715354\\
0.839167188149251	1.37864559715353\\
0.839954494843681	1.37684625034697\\
0.840741801538112	1.37504893582525\\
0.841529108232543	1.373254347367\\
0.842946191970531	1.37003242349995\\
0.844363275708518	1.36682526794846\\
0.845780359446506	1.36363837130819\\
0.847197443184494	1.3604772679728\\
0.849673140830041	1.35503407044296\\
0.852148838475588	1.34970617148522\\
0.854624536121135	1.3445083727259\\
0.857100233766682	1.33944921047908\\
0.860733162433376	1.33228023548568\\
0.864366091100069	1.32541041470681\\
0.867999019766763	1.31882342897629\\
0.871631948433457	1.312499236807\\
0.875264877100151	1.30642029353689\\
0.880598683728518	1.29791339210973\\
0.885932490356885	1.28987367697354\\
0.891266296985251	1.28227313975846\\
0.896600103613618	1.27508538119269\\
0.901933910241985	1.26828411179364\\
0.914633337161827	1.25350431565918\\
0.927332764081668	1.24047188694344\\
0.94003219100151	1.22892495056639\\
0.952731617921352	1.21864161637572\\
0.965431044841194	1.20943148176873\\
0.980471820521895	1.1996866298945\\
0.995512596202597	1.19098826335529\\
1.0105533718833	1.18314793968095\\
1.025594147564	1.17600871338274\\
1.0406349232447	1.16944002929142\\
1.0556756989254	1.16333485160177\\
1.08047763748093	1.15405264629383\\
1.10527957603646	1.14549790489458\\
1.13008151459198	1.13746042817929\\
1.15488345314751	1.12979374244635\\
1.17968539170304	1.12239809815439\\
1.20448733025857	1.11520846518975\\
1.21576320689173	1.11199564387862\\
1.21576320689175	1.11199564387861\\
1.2217942941732	1.11029410229592\\
1.22782538145466	1.10860102423268\\
1.23283413053473	1.10720096036668\\
1.2378428796148	1.10580666839147\\
1.24285162869488	1.10441803849529\\
1.264314550307	1.09852572987946\\
1.28577747191913	1.09273195829758\\
1.30724039353125	1.08703828091744\\
1.32870331514337	1.08144726619616\\
1.36372040693861	1.07255795273582\\
1.39873749873386	1.06397197997086\\
1.42748173839686	1.05716356606115\\
1.45622597805985	1.05058287515816\\
1.48497021772285	1.04423817262999\\
1.51371445738585	1.03813535017584\\
1.57893541937361	1.02519505942554\\
1.64415638136137	1.01351185593494\\
1.70937734334913	1.00303468925647\\
1.77459830533689	0.993669005807379\\
1.83981926732465	0.98529399356251\\
1.93263478168043	0.974827673551184\\
2.0254502960362	0.965721391035339\\
2.11826581039198	0.957626044703233\\
2.21108132474776	0.950248918176718\\
2.30389683910354	0.943357336822046\\
2.39671235345932	0.936775470099505\\
2.51084280561197	0.928927599487196\\
2.62497325776462	0.921227852251608\\
2.73910370991727	0.913622664164816\\
2.85323416206992	0.906113706688922\\
2.96736461422257	0.898734516144567\\
3.08149506637522	0.891533362317657\\
3.24562182862528	0.881588610543147\\
3.40974859087534	0.872248500861655\\
3.57387535312541	0.863599160849066\\
3.73800211537547	0.855668446133061\\
3.90212887762554	0.848437113310407\\
4.0662556398756	0.841856161506367\\
4.25489403319794	0.835012359164963\\
4.44353242652027	0.828834487677536\\
4.6321708198426	0.823218049653013\\
4.82080921316494	0.818072794109609\\
5.00944760648727	0.813326624673106\\
5.19808599980961	0.808924735557383\\
5.4600657467971	0.803307365410929\\
5.72204549378459	0.798200860746452\\
5.98402524077208	0.793555680944731\\
6.24600498775956	0.789334772980606\\
6.50798473474705	0.785505685774588\\
6.76996448173454	0.782037601053305\\
7.12996448173454	0.777807154081572\\
7.48996448173455	0.774130152770937\\
7.84996448173455	0.770938404597235\\
8.20996448173454	0.768169328143289\\
8.56996448173454	0.765766979365083\\
8.92996448173454	0.763682500473369\\
9.28996448173454	0.761873704081226\\
9.64996448173454	0.76030414835133\\
10.0099644817345	0.758942034960808\\
10.3699644817345	0.757759891625382\\
10.4436607789625	0.757537866981587\\
10.4436607789626	0.757537866981587\\
10.4567102670559	0.757675589290153\\
10.4697597551492	0.758156863974282\\
10.4828092432425	0.759036605376303\\
10.5004271336021	0.760697158753438\\
10.5180450239618	0.762794767343008\\
10.5356629143214	0.765120487398187\\
10.553280804681	0.767480597549759\\
10.5814887693115	0.771123338412671\\
10.609696733942	0.774532723131301\\
10.6379046985725	0.777854314083316\\
10.666112663203	0.781074269614569\\
10.6943206278335	0.783850895370234\\
10.7299519988119	0.78665623732622\\
10.7655833697903	0.789205252840505\\
10.8012147407688	0.791918078180868\\
10.8368461117472	0.794622015408662\\
10.8724774827256	0.797009309812718\\
10.9131109509035	0.799352849778948\\
10.9537444190815	0.801488690868755\\
10.9943778872594	0.803519373839468\\
11.0350113554373	0.805403588207542\\
11.0756448236152	0.807107036770001\\
11.1581077821789	0.810172144948\\
11.2405707407426	0.812816162745326\\
11.3230336993063	0.815133355716135\\
11.40549665787	0.817194262974794\\
11.5213299583129	0.819765020563928\\
11.6371632587558	0.822080844552243\\
11.7529965591986	0.824232020985648\\
11.8688298596415	0.826268400296893\\
11.9846631600843	0.828215014179005\\
12.1698835346324	0.831158213180208\\
12.3551039091805	0.833896392907069\\
12.3862919288963	0.83433548078208\\
12.3862919288965	0.834335480782083\\
12.4731829173858	0.835503560713903\\
12.5600739058752	0.836615853596902\\
12.6469648943645	0.837670013982214\\
12.8322381481414	0.839745255820744\\
13.0175114019183	0.841568837356598\\
13.2027846556951	0.843143817352934\\
13.388057909472	0.844486997069573\\
13.6921006892173	0.846253533715161\\
13.9961434689627	0.847580746331316\\
14.300186248708	0.848586559586971\\
14.6042290284533	0.849363330007954\\
14.9082718081986	0.849971209407475\\
15.2682718081986	0.850520009739114\\
15.6282718081986	0.850923074344802\\
15.9882718081986	0.851212338861381\\
16.3482718081986	0.851415998294265\\
16.7082718081986	0.851558233056841\\
17.0682718081986	0.851657317297838\\
17.4282718081986	0.851725532288848\\
17.7882718081986	0.851770954561502\\
18	0.851789544534929\\
};
\addlegendentry{$y_1$}
\addplot [color=mycolor2, thick]
  table[row sep=crcr]{%
1.81405864337411e-05	-0.862674385488145\\
2.39378828079243e-05	-0.862696441551286\\
2.97351791821074e-05	-0.862718496559785\\
5.94903675262398e-05	-0.862831679442221\\
8.92455558703722e-05	-0.862944834545777\\
0.000119000744214505	-0.863057961873841\\
0.000148755932558637	-0.863171061429847\\
0.000178511120902769	-0.863284133217303\\
0.000330834722065931	-0.863862538712664\\
0.000483158323229092	-0.864440217106897\\
0.000635481924392254	-0.865017168962194\\
0.000787805525555415	-0.865593394892213\\
0.000940129126718577	-0.86616889555887\\
0.00140835038828275	-0.867933370048604\\
0.00187657164984693	-0.869691021600641\\
0.0023447929114111	-0.871441876326343\\
0.00281301417297528	-0.873185962979233\\
0.00328123543453946	-0.874923312450833\\
0.00374945669610363	-0.87665395733835\\
0.00447981992343722	-0.879340228946506\\
0.00521018315077081	-0.882010402440486\\
0.00594054637810441	-0.884664613258384\\
0.006670909605438	-0.887302996646439\\
0.00740127283277159	-0.889925685648357\\
0.00813163606010518	-0.892532809599772\\
0.00903228644112706	-0.8957265473435\\
0.00993293682214894	-0.89889702533174\\
0.0108335872031708	-0.902044448965546\\
0.0117342375841927	-0.905169007092319\\
0.0126348879652146	-0.908270871883597\\
0.0135355383462365	-0.911350199243972\\
0.0146992805209031	-0.915295890641251\\
0.0158630226955697	-0.919204454529107\\
0.0170267648702363	-0.923076121713425\\
0.0181905070449029	-0.926911091829246\\
0.0193542492195695	-0.930709538450116\\
0.0205179913942362	-0.934471613914896\\
0.0216817335689028	-0.938197453788076\\
0.0231149783675184	-0.942736520360848\\
0.0245482231661341	-0.94722102197059\\
0.0259814679647497	-0.951651157437836\\
0.0274147127633654	-0.956027119426083\\
0.028847957561981	-0.960349099099798\\
0.0302812023605967	-0.964617289451374\\
0.0317144471592123	-0.968831887440077\\
0.033650275737842	-0.974439664936773\\
0.0355861043164717	-0.979950556885858\\
0.0375219328951014	-0.985365093453931\\
0.0394577614737311	-0.990683818378326\\
0.0413935900523608	-0.995907285599732\\
0.0433294186309905	-1.0010360562633\\
0.0452652472096201	-1.00607069592674\\
0.0475787000401237	-1.01196477564508\\
0.0498921528706272	-1.01772619587448\\
0.0522056057011308	-1.0233559207816\\
0.0545190585316343	-1.0288549097992\\
0.0568325113621379	-1.03422411710892\\
0.0591459641926414	-1.03946449146737\\
0.061459417023145	-1.04457697641215\\
0.0648229065114633	-1.0517837249359\\
0.0681863959997816	-1.05872499613884\\
0.0715498854880998	-1.06540364074611\\
0.0749133749764181	-1.07182249281358\\
0.0782768644647364	-1.0779843702609\\
0.0816403539530547	-1.08389207465319\\
0.085003843441373	-1.08954839044432\\
0.0910185701975383	-1.09904522254838\\
0.0970332969537036	-1.10776267989172\\
0.103048023709869	-1.1157162566939\\
0.109062750466034	-1.12292125910574\\
0.1150774772222	-1.12939280355731\\
0.121092203978365	-1.13514581541104\\
0.130716276413736	-1.14289303168847\\
0.140340348849108	-1.14889754490765\\
0.149964421284479	-1.15321745519894\\
0.15958849371985	-1.15590953997153\\
0.169212566155222	-1.15702928871007\\
0.182262532166814	-1.15613290337577\\
0.195312498178407	-1.15257669815508\\
0.20836246419	-1.14648798350401\\
0.221412430201593	-1.13798998345377\\
0.234462396213186	-1.127201935597\\
0.247512362224779	-1.11423915834753\\
0.269151649268809	-1.08824958665824\\
0.290790936312839	-1.05707127851746\\
0.312430223356869	-1.02116296436412\\
0.334069510400898	-0.980957682457495\\
0.355708797444928	-0.936864042126287\\
0.377348084488958	-0.889267406987228\\
0.408128360995123	-0.816239577139826\\
0.438908637501287	-0.737828838068701\\
0.469688914007451	-0.65492812326166\\
0.500469190513616	-0.568356358943322\\
0.53124946701978	-0.47886351248518\\
0.562029743525944	-0.38713508852212\\
0.605001962060835	-0.25650883834717\\
0.647974180595726	-0.124256372697006\\
0.690946399130617	0.00829749260669921\\
0.732696637621808	0.136280258821093\\
0.732696637621815	0.136280258821114\\
0.733334386748233	0.138224021676195\\
0.733972135874651	0.140167387786418\\
0.73460988500107	0.142110344808749\\
0.735749614797906	0.14558173141263\\
0.736889344594742	0.149051834505556\\
0.738029074391578	0.152520564517498\\
0.739168804188414	0.155987850606909\\
0.741258966692088	0.162342558555976\\
0.743349129195762	0.168691645364071\\
0.745439291699435	0.175034438032258\\
0.747529454203109	0.181370117142326\\
0.749619616706782	0.187697737737983\\
0.752436764555991	0.196211564271985\\
0.755253912405201	0.204706093136758\\
0.75807106025441	0.213178410995615\\
0.760888208103619	0.221625513468889\\
0.763705355952828	0.23004438457868\\
0.765956147597659	0.236748473713357\\
0.765956147597666	0.236748473713379\\
0.766553424930181	0.238523421878093\\
0.767150702262695	0.240296855867472\\
0.76774797959521	0.242068718597674\\
0.770352911932362	0.249778499643778\\
0.772957844269515	0.257456558006653\\
0.775562776606667	0.265100544182288\\
0.77816770894382	0.272708415863482\\
0.781599780219385	0.282674041101548\\
0.78503185149495	0.292570262118755\\
0.788463922770516	0.302393677370544\\
0.791895994046081	0.312141353537824\\
0.795328065321646	0.321810788986392\\
0.801128551274248	0.337969486158998\\
0.806929037226851	0.353890485682962\\
0.812729523179453	0.369567689852921\\
0.818530009132055	0.384997040039146\\
0.824330495084657	0.400176066274504\\
0.831665342931552	0.419009929113601\\
0.839000190778447	0.437441084168186\\
0.839167188149244	0.437856036817282\\
0.839167188149251	0.4378560368173\\
0.839954494843681	0.439807948505892\\
0.840741801538112	0.441755352573877\\
0.841529108232543	0.443698158448374\\
0.842946191970531	0.447185001064311\\
0.844363275708518	0.450657731348143\\
0.845780359446506	0.454115851396177\\
0.847197443184494	0.457559237098424\\
0.849673140830041	0.463539842341139\\
0.852148838475588	0.469476236018473\\
0.854624536121135	0.475369081386165\\
0.857100233766682	0.481219113175151\\
0.860733162433376	0.489727660587785\\
0.864366091100069	0.498147996460873\\
0.867999019766763	0.506482517108259\\
0.871631948433457	0.514733455196075\\
0.875264877100151	0.522902862839532\\
0.880598683728518	0.53475328130482\\
0.885932490356885	0.546437774741741\\
0.891266296985251	0.557961712001354\\
0.896600103613618	0.569330134291193\\
0.901933910241985	0.580547779354499\\
0.914633337161827	0.60667418489157\\
0.927332764081668	0.632024054364298\\
0.94003219100151	0.656643898088745\\
0.952731617921352	0.680573121816296\\
0.965431044841194	0.703845319300818\\
0.980471820521895	0.730597830752451\\
0.995512596202597	0.756509728850512\\
1.0105533718833	0.781614765553034\\
1.025594147564	0.805941362171626\\
1.0406349232447	0.829513574755362\\
1.0556756989254	0.852351806033351\\
1.08047763748093	0.888454182948483\\
1.10527957603646	0.922672177626743\\
1.13008151459198	0.955058179720035\\
1.15488345314751	0.985657156721496\\
1.17968539170304	1.01450949325857\\
1.20448733025857	1.04165281845684\\
1.21576320689173	1.05343804939675\\
1.21576320689175	1.05343804939676\\
1.2217942941732	1.05956732534119\\
1.22782538145466	1.06560191555717\\
1.23283413053473	1.07054511126548\\
1.2378428796148	1.07542168876258\\
1.24285162869488	1.08023219242513\\
1.264314550307	1.10016039138668\\
1.28577747191913	1.11893074209359\\
1.30724039353125	1.1365482583784\\
1.32870331514337	1.15303184089207\\
1.36372040693861	1.17756994789426\\
1.39873749873386	1.19927852342143\\
1.42748173839686	1.21506343758853\\
1.45622597805985	1.22909567456672\\
1.48497021772285	1.2414485025741\\
1.51371445738585	1.25219597978685\\
1.57893541937361	1.27101102597285\\
1.64415638136137	1.28281590229922\\
1.70937734334913	1.28849381583326\\
1.77459830533689	1.28890180898939\\
1.83981926732465	1.28485300034994\\
1.93263478168043	1.27289034797187\\
2.0254502960362	1.25541092778647\\
2.11826581039198	1.23412340214749\\
2.21108132474776	1.21046102068328\\
2.30389683910354	1.18559378729651\\
2.39671235345932	1.16044833993571\\
2.51084280561197	1.13017750381958\\
2.62497325776462	1.10146659295135\\
2.73910370991727	1.07487447603666\\
2.85323416206992	1.05068129398031\\
2.96736461422257	1.02896105095398\\
3.08149506637522	1.00964221811021\\
3.24562182862528	0.985737310491698\\
3.40974859087534	0.965769614066034\\
3.57387535312541	0.94898629247118\\
3.73800211537547	0.934696955363778\\
3.90212887762554	0.922319701600704\\
4.0662556398756	0.91138657044353\\
4.25489403319794	0.900141435607994\\
4.44353242652027	0.889941440733512\\
4.6321708198426	0.880535090891496\\
4.82080921316494	0.871768946519554\\
5.00944760648727	0.863551621827213\\
5.19808599980961	0.855828573753168\\
5.4600657467971	0.845859616115218\\
5.72204549378459	0.836727187025621\\
5.98402524077208	0.828390849475013\\
6.24600498775956	0.820804155345211\\
6.50798473474705	0.813915588525959\\
6.76996448173454	0.807673346235584\\
7.12996448173454	0.800057651043046\\
7.48996448173455	0.793438383692223\\
7.84996448173455	0.787692334313212\\
8.20996448173454	0.782706544495973\\
8.56996448173454	0.778380848736935\\
8.92996448173454	0.774627908642283\\
9.28996448173454	0.771371862639762\\
9.64996448173454	0.768546786731519\\
10.0099644817345	0.766095109113559\\
10.3699644817345	0.763967286067873\\
10.4436607789625	0.763567641960673\\
10.4436607789626	0.763567641960672\\
10.4567102670559	0.763502472884128\\
10.4697597551492	0.763450871690955\\
10.4828092432425	0.763423840492437\\
10.5004271336021	0.763433170714644\\
10.5180450239618	0.763513813484491\\
10.5356629143214	0.763679999439115\\
10.553280804681	0.763937399942572\\
10.5814887693115	0.764536363983869\\
10.609696733942	0.765356098775333\\
10.6379046985725	0.766383420734908\\
10.666112663203	0.767607389499489\\
10.6943206278335	0.769007991172668\\
10.7299519988119	0.77098125326023\\
10.7655833697903	0.773140202354377\\
10.8012147407688	0.77546911085361\\
10.8368461117472	0.777965217275561\\
10.8724774827256	0.780612590023465\\
10.9131109509035	0.783778087551196\\
10.9537444190815	0.787061847140816\\
10.9943778872594	0.79043580158357\\
11.0350113554373	0.793878024029265\\
11.0756448236152	0.797365738579825\\
11.1581077821789	0.804468963702043\\
11.2405707407426	0.811514408837534\\
11.3230336993063	0.818383980125734\\
11.40549665787	0.824983272148717\\
11.5213299583129	0.833678220986109\\
11.6371632587558	0.841588245908968\\
11.7529965591986	0.848663137549214\\
11.8688298596415	0.854915634225642\\
11.9846631600843	0.860399036006902\\
12.1698835346324	0.867764120315037\\
12.3551039091805	0.87369374935231\\
12.3862919288963	0.874575050763252\\
12.3862919288965	0.874575050763257\\
12.4731829173858	0.876821962030458\\
12.5600739058752	0.878893325431167\\
12.6469648943645	0.880807064288508\\
12.8322381481414	0.88450108910217\\
13.0175114019183	0.887704669698436\\
13.2027846556951	0.890503483741424\\
13.388057909472	0.892967816472569\\
13.6921006892173	0.896396991815767\\
13.9961434689627	0.899181979650701\\
14.300186248708	0.901401001439303\\
14.6042290284533	0.903121889120075\\
14.9082718081986	0.904429593126838\\
15.2682718081986	0.905572001207291\\
15.6282718081986	0.906399777537759\\
15.9882718081986	0.90700752389121\\
16.3482718081986	0.907453002704244\\
16.7082718081986	0.90777263317492\\
17.0682718081986	0.90799506806408\\
17.4282718081986	0.908146091401342\\
17.7882718081986	0.908247390644299\\
18	0.908290755815753\\
};
\addlegendentry{$y_2$}
\end{axis}
\end{tikzpicture}%

%% file: plots/transient_con_hard.tex
%
%
\definecolor{mycolor1}{RGB}{0,81,158}%
\definecolor{mycolor2}{RGB}{0, 190, 255}
\definecolor{usGray}{RGB}{139,143,148}%
\begin{tikzpicture}

\begin{axis}[%
    name=mainaxis, 
    at={(0.758in,1.5in)},
    width=1.15in,   
    height=0.8in,
    scale only axis,
    xmin=0,
    xmax=15,
    ymin=0 ,
    ymax=1.2,
    axis background/.style={fill=white},
    xmajorgrids,
    ymajorgrids,
    ylabel={$u$},
]
    \fill[red!80!black, opacity=0.05] (axis cs:0,1) rectangle (axis cs:18,1.2);
    \draw[red!80!black, opacity=0.5,dashed] (axis cs:0,1) -- (axis cs:18,1); 
    
    \fill[red!80!black, opacity=0.05] (axis cs:0,0.2) rectangle (axis cs:18,0);
    \draw[red!80!black,opacity=0.5,dashed] (axis cs:0,0.2) -- (axis cs:18,0.2);

    \addplot [color=gray, dashed] table[row sep=crcr]{0 0.9091\\ 20 0.9091\\};

    \addplot [color=mycolor1, thick] table[row sep=crcr]{
0	0.2\\
3.29148401249502e-07	0.2\\
6.58296802499004e-07	0.2\\
9.87445203748505e-07	0.2\\
2.40185241278055e-06	0.2\\
3.8162596218126e-06	0.2\\
5.23066683084465e-06	0.2\\
6.6450740398767e-06	0.2\\
1.25028029768298e-05	0.2\\
1.83605319137828e-05	0.2\\
2.42182608507359e-05	0.2\\
3.0075989787689e-05	0.2\\
6.00482120075413e-05	0.2\\
9.00204342273936e-05	0.2\\
0.000119992656447246	0.2\\
0.000149964878667098	0.2\\
0.000179937100886951	0.2\\
0.000332323598807076	0.2\\
0.000484710096727201	0.2\\
0.000637096594647326	0.2\\
0.000789483092567451	0.2\\
0.000941869590487576	0.2\\
0.00138404454605448	0.2\\
0.00182621950162138	0.2\\
0.00226839445718828	0.2\\
0.00271056941275518	0.2\\
0.00315274436832208	0.2\\
0.00359491932388899	0.2\\
0.00433362333684666	0.2\\
0.00507232734980433	0.2\\
0.005811031362762	0.2\\
0.00654973537571967	0.2\\
0.00728843938867734	0.2\\
0.00802714340163501	0.2\\
0.00893054121565263	0.2\\
0.00983393902967025	0.2\\
0.0107373368436879	0.2\\
0.0116407346577055	0.2\\
0.0125441324717231	0.2\\
0.0134475302857407	0.2\\
0.0146224277507307	0.2\\
0.0157973252157207	0.2\\
0.0169722226807107	0.2\\
0.0181471201457007	0.2\\
0.0193220176106907	0.2\\
0.0204969150756807	0.2\\
0.0216718125406707	0.2\\
0.0231041812400222	0.2\\
0.0245365499393737	0.2\\
0.0259689186387252	0.2\\
0.0274012873380767	0.2\\
0.0288336560374282	0.2\\
0.0302660247367797	0.2\\
0.0316983934361312	0.2\\
0.0336294714734546	0.2\\
0.035560549510778	0.2\\
0.0374916275481013	0.2\\
0.0394227055854247	0.2\\
0.0413537836227481	0.2\\
0.0432848616600715	0.2\\
0.0452159396973949	0.2\\
0.0475942398069608	0.2\\
0.0499725399165266	0.2\\
0.0523508400260925	0.2\\
0.0547291401356584	0.2\\
0.0571074402452243	0.2\\
0.0594857403547902	0.2\\
0.0618640404643561	0.2\\
0.0654611559218425	0.2\\
0.0690582713793289	0.2\\
0.0726553868368153	0.2\\
0.0762525022943017	0.2\\
0.0798496177517881	0.2\\
0.0834467332092745	0.2\\
0.0870438486667609	0.2\\
0.0948395748196346	0.2\\
0.102635300972508	0.2\\
0.110431027125382	0.2\\
0.118226753278256	0.2\\
0.126022479431129	0.2\\
0.133818205584003	0.2\\
0.145679573186037	0.2\\
0.157540940788072	0.2\\
0.169402308390106	0.2\\
0.181263675992141	0.2\\
0.193125043594175	0.2\\
0.20498641119621	0.2\\
0.222976361272891	0.2\\
0.240966311349572	0.2\\
0.258956261426252	0.2\\
0.276946211502933	0.2\\
0.294936161579614	0.2\\
0.312926111656295	0.2\\
0.338332798699639	0.2\\
0.363739485742983	0.2\\
0.389146172786327	0.2\\
0.41455285982967	0.2\\
0.439959546873014	0.2\\
0.465366233916358	0.2\\
0.501191457304087	0.2\\
0.537016680691815	0.2\\
0.572841904079544	0.2\\
0.608667127467273	0.2\\
0.644492350855001	0.2\\
0.68031757424273	0.2\\
0.736084579616188	0.2\\
0.774701678957587	0.2\\
0.774701678957594	0.2\\
0.77508750920008	0.2\\
0.775473339442565	0.2\\
0.775859169685051	0.2\\
0.776959412267718	0.2\\
0.778059654850384	0.2\\
0.779159897433051	0.2\\
0.780260140015717	0.2\\
0.782245372849324	0.2\\
0.784230605682932	0.2\\
0.786215838516539	0.2\\
0.788201071350146	0.2\\
0.790186304183754	0.2\\
0.792778950437179	0.2\\
0.795371596690605	0.2\\
0.797964242944031	0.2\\
0.800556889197457	0.2\\
0.803149535450883	0.2\\
0.807069734006657	0.2\\
0.810989932562431	0.2\\
0.814910131118206	0.2\\
0.81883032967398	0.2\\
0.822750528229754	0.2\\
0.829904494132436	0.2\\
0.837058460035119	0.2\\
0.844212425937801	0.2\\
0.851366391840484	0.2\\
0.853573891945413	0.2\\
0.85357389194542	0.200000000000001\\
0.854501628534333	0.202882072064435\\
0.855429365123245	0.20583777612759\\
0.856357101712158	0.208896119208401\\
0.858281664089495	0.215606288025064\\
0.860206226466832	0.222907375070729\\
0.86213078884417	0.230862262205257\\
0.864055351221507	0.239472064184042\\
0.867292922924893	0.255247736304183\\
0.870530494628279	0.272353162741123\\
0.873768066331665	0.290383416909899\\
0.877005638035052	0.308938084017443\\
0.881670228187441	0.335937460801839\\
0.88633481833983	0.362474945395258\\
0.890999408492219	0.387943997250156\\
0.895663998644608	0.411997918843317\\
0.898715552345879	0.426893303714033\\
0.898715552345886	0.426893303714072\\
0.899350330113535	0.429893147486201\\
0.899985107881185	0.432864267139179\\
0.900619885648835	0.435806035651259\\
0.901737821448125	0.440924769612581\\
0.902855757247416	0.445956579984667\\
0.903973693046706	0.450898225615875\\
0.905091628845997	0.455748981359591\\
0.90713097398682	0.464365951161692\\
0.909170319127642	0.472687890427002\\
0.911209664268465	0.480723028463155\\
0.913249009409288	0.488481193377092\\
0.915288354550111	0.495973051076358\\
0.91858842760321	0.5075611829169\\
0.92188850065631	0.518530579167711\\
0.925188573709409	0.528931849524858\\
0.928488646762508	0.538815963646057\\
0.931788719815608	0.548232811260188\\
0.936125538778493	0.559976391567108\\
0.940462357741379	0.571094938759233\\
0.944799176704265	0.581679485698162\\
0.949135995667151	0.591810192832122\\
0.953472814630037	0.601556419246313\\
0.961899024816125	0.61959890044445\\
0.970325235002213	0.636765338625404\\
0.978751445188301	0.653311440163039\\
0.987177655374389	0.669379441196853\\
0.995603865560477	0.685035361666104\\
1.01047569450743	0.711800654169044\\
1.02534752345438	0.737522724596996\\
1.04021935240133	0.762369598105266\\
1.05509118134828	0.786585861729115\\
1.06996301029524	0.81018811772345\\
1.08483483924219	0.832767619403709\\
1.10600607332612	0.862592757657334\\
1.12717730741006	0.89007087065765\\
1.14834854149399	0.916529188636741\\
1.16951977557792	0.942171196374977\\
1.19069100966186	0.965585233355375\\
1.21186224374579	0.98601530608761\\
1.22767455421788	1\\
1.2276745542179	1\\
1.23427198020141	1\\
1.23957340651417	1\\
1.24487483282693	1\\
1.25017625913969	1\\
1.2637203604155	1\\
1.27726446169132	1\\
1.29080856296713	1\\
1.30435266424295	1\\
1.32714784544659	1\\
1.34994302665023	1\\
1.37273820785388	1\\
1.39553338905752	1\\
1.41832857026116	1\\
1.45292233798962	1\\
1.48751610571808	1\\
1.52210987344654	1\\
1.556703641175	1\\
1.59129740890347	1\\
1.64184733780232	1\\
1.69239726670117	1\\
1.74294719560002	1\\
1.79349712449887	1\\
1.84404705339773	1\\
1.91678327088859	1\\
1.98951948837946	1\\
2.06225570587033	1\\
2.1349919233612	1\\
2.20772814085207	1\\
2.29665508334235	1\\
2.38558202583264	1\\
2.47450896832292	1\\
2.56343591081321	1\\
2.6523628533035	1\\
2.74128979579378	1\\
2.86573345209315	1\\
2.99017710839252	1\\
3.11462076469189	1\\
3.23906442099125	1\\
3.36350807729062	1\\
3.48795173358999	1\\
3.61843705168917	1\\
3.61843705168919	1\\
3.63204764124076	1\\
3.64565823079233	1\\
3.6592688203439	1\\
3.67906822172875	1\\
3.68820949052243	1\\
3.68820949052247	0.999999999999993\\
3.70331819837544	0.994226230542544\\
3.71842690622841	0.98747670108482\\
3.73353561408137	0.980918464893218\\
3.75574411716739	0.972723654049574\\
3.77795262025341	0.966749396774127\\
3.80016112333942	0.962609380175447\\
3.82236962642544	0.959334274547352\\
3.85357645144836	0.95495185606157\\
3.88478327647129	0.950603710474851\\
3.91599010149421	0.946557632940166\\
3.94719692651714	0.94311610297208\\
4.00309693317951	0.938115565865047\\
4.05899693984188	0.933969812671251\\
4.11489694650425	0.93039774392188\\
4.17079695316662	0.927658515505275\\
4.226696959829	0.925731029003497\\
4.29976365035829	0.923990487986895\\
4.37283034088758	0.922906350984139\\
4.44589703141687	0.922402941513937\\
4.51896372194617	0.922355627590075\\
4.59203041247546	0.922586079566202\\
4.69650948133496	0.923174996280697\\
4.80098855019445	0.92386573077022\\
4.90546761905394	0.924490554520133\\
5.00994668791344	0.92492935777322\\
5.11442575677293	0.925123536908188\\
5.24376976254039	0.925011923034166\\
5.37311376830786	0.924532127831284\\
5.50245777407532	0.923748779703413\\
5.63180177984278	0.922749133020066\\
5.76114578561024	0.921622683329218\\
5.92583541398891	0.920125901566362\\
6.09052504236758	0.918676795493941\\
6.25521467074625	0.917360561572451\\
6.41990429912492	0.916220367783922\\
6.58459392750359	0.91526561227009\\
6.77528722102658	0.914372800500725\\
6.96598051454957	0.913672190676561\\
7.15667380807257	0.913118315781755\\
7.34736710159556	0.9126690336197\\
7.53806039511855	0.912291514552422\\
7.82653856220132	0.911813109486229\\
8.11501672928408	0.911406673585276\\
8.40349489636684	0.911055884284803\\
8.6919730634496	0.910758888765542\\
8.98045123053236	0.910515227076257\\
9.34045123053236	0.910278025787694\\
9.70045123053236	0.910100892768416\\
10.0604512305324	0.909967432710836\\
10.4204512305324	0.909864795512441\\
10.7804512305324	0.909785012600206\\
11.1404512305324	0.909723470752553\\
11.5004512305324	0.909676809733984\\
11.8604512305324	0.909641833298956\\
12.2204512305324	0.909615541829591\\
12.5804512305324	0.909595547044805\\
12.9404512305324	0.909580236431776\\
13.3004512305324	0.909568582455957\\
13.6604512305324	0.909559839055522\\
14.0204512305324	0.909553356599158\\
14.3804512305324	0.909548557806843\\
14.7404512305324	0.909544982508699\\
15.1004512305324	0.909542307600976\\
15.4604512305323	0.909540318374553\\
15.8204512305323	0.909538860945961\\
16.1804512305323	0.909537809504819\\
16.5404512305323	0.909537057334617\\
16.9004512305323	0.909536520090203\\
17.2604512305323	0.909536137520526\\
17.6204512305323	0.909535868900036\\
17.9804512305323	0.909535685528201\\
18	0.909535677508032\\
    };

\end{axis}


\begin{axis}[%
width=1.15in,
height=0.8in,
at={(0.758in,0.481in)},
scale only axis,
xmin=0,
xmax=17,
ymin=-1,
ymax=2,
axis background/.style={fill=white},
xmajorgrids,
ymajorgrids,
ylabel style = {yshift=-3.5mm},
ylabel={$y$},
xlabel={$t$},
xlabel style = {yshift=1.5mm},
legend style={
    at={(0.97,0.03)},
    anchor=south east,
    nodes={scale=0.65, transform shape}
}
]

\addplot [color=gray, dashed, forget plot]
  table[row sep=crcr]{%
0	0.852272727272727\\
18	0.852272727272727\\
};

\addplot [color=gray, dashed, forget plot]
  table[row sep=crcr]{%
0	0.909090909090909\\
18	0.909090909090909\\
};

\addplot [color=mycolor1, thick]
  table[row sep=crcr]{%
0	-3.48914335047704\\
3.29148401249502e-07	-3.48913646152557\\
6.58296802499004e-07	-3.48912957258226\\
9.87445203748505e-07	-3.48912268364725\\
2.40185241278055e-06	-3.48909308079867\\
3.8162596218126e-06	-3.48906347810081\\
5.23066683084465e-06	-3.48903387555608\\
6.6450740398767e-06	-3.48900427316533\\
1.25028029768298e-05	-3.48888167731167\\
1.83605319137828e-05	-3.48875908410179\\
2.42182608507359e-05	-3.48863649353575\\
3.0075989787689e-05	-3.48851390561352\\
6.00482120075413e-05	-3.48788670176816\\
9.00204342273936e-05	-3.48725956713014\\
0.000119992656447246	-3.48663250169238\\
0.000149964878667098	-3.48600550544791\\
0.000179937100886951	-3.48537857838976\\
0.000332323598807076	-3.48219218959172\\
0.000484710096727201	-3.47900758814049\\
0.000637096594647326	-3.47582477311936\\
0.000789483092567451	-3.47264374361198\\
0.000941869590487576	-3.46946449870256\\
0.00138404454605448	-3.46024948048423\\
0.00182621950162138	-3.45104945795645\\
0.00226839445718828	-3.44186440881268\\
0.00271056941275518	-3.4326943107784\\
0.00315274436832208	-3.42353914161145\\
0.00359491932388899	-3.41439887910229\\
0.00433362333684666	-3.39916221454774\\
0.00507232734980433	-3.38396698870223\\
0.005811031362762	-3.36881309861569\\
0.00654973537571967	-3.35370044158859\\
0.00728843938867734	-3.33862891517153\\
0.00802714340163501	-3.32359841716455\\
0.00893054121565263	-3.30527249503556\\
0.00983393902967025	-3.28700759680367\\
0.0107373368436879	-3.26880353705272\\
0.0116407346577055	-3.25066013091807\\
0.0125441324717231	-3.23257719408482\\
0.0134475302857407	-3.21455454278618\\
0.0146224277507307	-3.19120539714571\\
0.0157973252157207	-3.16795750591558\\
0.0169722226807107	-3.14481046899505\\
0.0181471201457007	-3.12176388782857\\
0.0193220176106907	-3.09881736539949\\
0.0204969150756807	-3.07597050622372\\
0.0216718125406707	-3.05322291634334\\
0.0231041812400222	-3.02562405299239\\
0.0245365499393737	-2.99817144343386\\
0.0259689186387252	-2.97086438292912\\
0.0274012873380767	-2.94370217005026\\
0.0288336560374282	-2.91668410666356\\
0.0302660247367797	-2.88980949791301\\
0.0316983934361312	-2.86307765220404\\
0.0336294714734546	-2.82726331080912\\
0.035560549510778	-2.79170552239192\\
0.0374916275481013	-2.75640262010613\\
0.0394227055854247	-2.72135294763017\\
0.0413537836227481	-2.68655485909682\\
0.0432848616600715	-2.65200671902326\\
0.0452159396973949	-2.61770690224166\\
0.0475942398069608	-2.57580238250403\\
0.0499725399165266	-2.5342690796513\\
0.0523508400260925	-2.49310402126805\\
0.0547291401356584	-2.45230425797276\\
0.0571074402452243	-2.41186686322908\\
0.0594857403547902	-2.37178893315889\\
0.0618640404643561	-2.33206758635695\\
0.0654611559218425	-2.27266127202311\\
0.0690582713793289	-2.21405433747598\\
0.0726553868368153	-2.15623710062216\\
0.0762525022943017	-2.09919999223019\\
0.0798496177517881	-2.04293355454621\\
0.0834467332092745	-1.98742843992731\\
0.0870438486667609	-1.93267540949334\\
0.0948395748196346	-1.81655154593762\\
0.102635300972508	-1.70382610049173\\
0.110431027125382	-1.59441010563554\\
0.118226753278256	-1.48821681281845\\
0.126022479431129	-1.38516163372563\\
0.133818205584003	-1.28516208584468\\
0.145679573186037	-1.13869158325496\\
0.157540940788072	-0.998834759953678\\
0.169402308390106	-0.865327448515778\\
0.181263675992141	-0.737915364693236\\
0.193125043594175	-0.616353729026336\\
0.20498641119621	-0.500406914148164\\
0.222976361272891	-0.334762964376787\\
0.240966311349572	-0.180760673554586\\
0.258956261426252	-0.0376879144627849\\
0.276946211502933	0.0951276587133425\\
0.294936161579614	0.218320823299024\\
0.312926111656295	0.332490807061108\\
0.338332798699639	0.47944995508255\\
0.363739485742983	0.611022164148893\\
0.389146172786327	0.728577010401849\\
0.41455285982967	0.833376760704607\\
0.439959546873014	0.926584349471166\\
0.465366233916358	1.0092705394696\\
0.501191457304087	1.10987584705188\\
0.537016680691815	1.19400880528365\\
0.572841904079544	1.2638875385684\\
0.608667127467273	1.32148017786181\\
0.644492350855001	1.36852990476868\\
0.68031757424273	1.4065768678013\\
0.736084579616188	1.45104750621075\\
0.774701678957587	1.4732920596593\\
0.774701678957594	1.47329205965931\\
0.77508750920008	1.47348383660328\\
0.775473339442565	1.4736750725425\\
0.775859169685051	1.47386575713506\\
0.776959412267718	1.47440678722937\\
0.778059654850384	1.47494349758379\\
0.779159897433051	1.47547580693173\\
0.780260140015717	1.47600370331135\\
0.782245372849324	1.47694512141695\\
0.784230605682932	1.47787236531385\\
0.786215838516539	1.47878559175298\\
0.788201071350146	1.47968495544935\\
0.790186304183754	1.48057060320118\\
0.792778950437179	1.48170681287641\\
0.795371596690605	1.48282019046028\\
0.797964242944031	1.48391105146252\\
0.800556889197457	1.48497970879449\\
0.803149535450883	1.48602647224396\\
0.807069734006657	1.48756832533566\\
0.810989932562431	1.48906186900733\\
0.814910131118206	1.49050813419493\\
0.81883032967398	1.49190813564454\\
0.822750528229754	1.49326287250179\\
0.829904494132436	1.49562166651212\\
0.837058460035119	1.49783883173008\\
0.844212425937801	1.49992006246612\\
0.851366391840484	1.50187088603939\\
0.853573891945413	1.50244735768569\\
0.85357389194542	1.50244735768569\\
0.854501628534333	1.50267416834933\\
0.855429365123245	1.50288601776327\\
0.856357101712158	1.50308207965288\\
0.858281664089495	1.50344145882495\\
0.860206226466832	1.50372939684911\\
0.86213078884417	1.50393807686279\\
0.864055351221507	1.50406083329927\\
0.867292922924893	1.50405725567959\\
0.870530494628279	1.50377685990763\\
0.873768066331665	1.50320618857467\\
0.877005638035052	1.50233783850914\\
0.881670228187441	1.50056267439754\\
0.88633481833983	1.49817987708664\\
0.890999408492219	1.49521905299173\\
0.895663998644608	1.49172165029761\\
0.898715552345879	1.48916571331184\\
0.898715552345886	1.48916571331183\\
0.899350330113535	1.48860605903495\\
0.899985107881185	1.48803824827227\\
0.900619885648835	1.48746228590546\\
0.901737821448125	1.48643119157591\\
0.902855757247416	1.48537696723862\\
0.903973693046706	1.4842992427831\\
0.905091628845997	1.4831983064012\\
0.90713097398682	1.48113203667467\\
0.909170319127642	1.47899393889632\\
0.911209664268465	1.47678775980213\\
0.913249009409288	1.47451709917339\\
0.915288354550111	1.47218535623681\\
0.91858842760321	1.46829085411698\\
0.92188850065631	1.46425776052367\\
0.925188573709409	1.46009805766544\\
0.928488646762508	1.45582268323286\\
0.931788719815608	1.45144158617146\\
0.936125538778493	1.44553826382817\\
0.940462357741379	1.43948598896416\\
0.944799176704265	1.43330061572134\\
0.949135995667151	1.42699582387007\\
0.953472814630037	1.42058341807106\\
0.961899024816125	1.40785241984233\\
0.970325235002213	1.39481183868308\\
0.978751445188301	1.38150823618536\\
0.987177655374389	1.36797846292139\\
0.995603865560477	1.35425399036494\\
1.01047569450743	1.32964135678427\\
1.02534752345438	1.30464759625219\\
1.04021935240133	1.27938104440368\\
1.05509118134828	1.25391715966057\\
1.06996301029524	1.22831934977988\\
1.08483483924219	1.20268184548416\\
1.10600607332612	1.16636020603184\\
1.12717730741006	1.13052266400467\\
1.14834854149399	1.09532272811358\\
1.16951977557792	1.06080496207669\\
1.19069100966186	1.02710419388324\\
1.21186224374579	0.994485809612886\\
1.22767455421788	0.970925208204639\\
1.2276745542179	0.970925208204618\\
1.23427198020141	0.961516819243888\\
1.23957340651417	0.954154155685179\\
1.24487483282693	0.946983755876758\\
1.25017625913969	0.940001449718074\\
1.2637203604155	0.922899394494377\\
1.27726446169132	0.906903167457943\\
1.29080856296713	0.891987028691062\\
1.30435266424295	0.878104330486552\\
1.32714784544659	0.856930792124662\\
1.34994302665023	0.838291637009107\\
1.37273820785388	0.821948786226456\\
1.39553338905752	0.807685667400338\\
1.41832857026116	0.795306207132807\\
1.45292233798962	0.779724967549847\\
1.48751610571808	0.767502882926669\\
1.52210987344654	0.75813511163549\\
1.556703641175	0.751179807433193\\
1.59129740890347	0.746249698900697\\
1.64184733780232	0.742007546423403\\
1.69239726670117	0.740472731673377\\
1.74294719560002	0.740905862676691\\
1.79349712449887	0.742716282807773\\
1.84404705339773	0.745433173910545\\
1.91678327088859	0.750226384002821\\
1.98951948837946	0.755310096712766\\
2.06225570587033	0.760125679750584\\
2.1349919233612	0.764338690867105\\
2.20772814085207	0.767770545810452\\
2.29665508334235	0.770799763275545\\
2.38558202583264	0.772578198827957\\
2.47450896832292	0.773237256959622\\
2.56343591081321	0.772958037794082\\
2.6523628533035	0.771939936742441\\
2.74128979579378	0.770382029703189\\
2.86573345209315	0.767639758142181\\
2.99017710839252	0.764624565213658\\
3.11462076469189	0.761652898036619\\
3.23906442099125	0.758929795691513\\
3.36350807729062	0.756569980580715\\
3.48795173358999	0.754621289757564\\
3.61843705168917	0.753020792195959\\
3.61843705168919	0.753020792195959\\
3.63204764124076	0.752879626071011\\
3.64565823079233	0.752742751777924\\
3.6592688203439	0.752610196357541\\
3.67906822172875	0.752423806689681\\
3.68820949052243	0.752340551596694\\
3.68820949052247	0.752340551596693\\
3.70331819837544	0.752535308739805\\
3.71842690622841	0.753181219042063\\
3.73353561408137	0.754267926599778\\
3.75574411716739	0.75640018818896\\
3.77795262025341	0.759012513208932\\
3.80016112333942	0.761878805064469\\
3.82236962642544	0.764849501204461\\
3.85357645144836	0.769087527711101\\
3.88478327647129	0.7733963962141\\
3.91599010149421	0.777739487105709\\
3.94719692651714	0.782030839157548\\
4.00309693317951	0.789369263237256\\
4.05899693984188	0.796065244348439\\
4.11489694650425	0.802053785885601\\
4.17079695316662	0.807321649592897\\
4.226696959829	0.811846935227821\\
4.29976365035829	0.81666608069571\\
4.37283034088758	0.820351729718229\\
4.44589703141687	0.823066203252172\\
4.51896372194617	0.824975385953676\\
4.59203041247546	0.826242568251616\\
4.69650948133496	0.82724723662191\\
4.80098855019445	0.82766393465966\\
4.90546761905394	0.827815138952901\\
5.00994668791344	0.827933672992394\\
5.11442575677293	0.828170918602343\\
5.24376976254039	0.828740302603465\\
5.37311376830786	0.82966415221406\\
5.50245777407532	0.83090942001285\\
5.63180177984278	0.832393034569069\\
5.76114578561024	0.834014689315003\\
5.92583541398891	0.836132172722727\\
6.09052504236758	0.838163221427778\\
6.25521467074625	0.840001292708212\\
6.41990429912492	0.84159439628951\\
6.58459392750359	0.842935530542909\\
6.77528722102658	0.844205363533107\\
6.96598051454957	0.845224289457204\\
7.15667380807257	0.84605446460553\\
7.34736710159556	0.846750373351104\\
7.53806039511855	0.847352045609583\\
7.82653856220132	0.848132125372305\\
8.11501672928408	0.848798422291735\\
8.40349489636684	0.84936653295214\\
8.6919730634496	0.849839109457199\\
8.98045123053236	0.850221381304463\\
9.34045123053236	0.850591071532042\\
9.70045123053236	0.85086827387417\\
10.0604512305324	0.8510795159883\\
10.4204512305324	0.851243432791275\\
10.7804512305324	0.851370984885492\\
11.1404512305324	0.851468971272039\\
11.5004512305324	0.851543085562713\\
11.8604512305324	0.851598864716224\\
12.2204512305324	0.851641177598645\\
12.5804512305324	0.851673619101143\\
12.9404512305324	0.851698525902824\\
13.3004512305324	0.851717456370567\\
13.6604512305324	0.851731655385615\\
14.0204512305324	0.851742238394845\\
14.3804512305324	0.851750154285443\\
14.7404512305324	0.851756115175332\\
15.1004512305324	0.851760605216251\\
15.4604512305323	0.851763956019453\\
15.8204512305323	0.85176642300914\\
16.1804512305323	0.8517682219905\\
16.5404512305323	0.851769531228785\\
16.9004512305323	0.851770485203027\\
17.2604512305323	0.851771177509673\\
17.6204512305323	0.85177167289726\\
17.9804512305323	0.851772019790554\\
18	0.851772035245452\\
};
\addlegendentry{$y_1$}
\addplot [color=mycolor2, thick]
  table[row sep=crcr]{%
0	-0.862605362039618\\
3.29148401249502e-07	-0.862606614513501\\
6.58296802499004e-07	-0.862607866984299\\
9.87445203748505e-07	-0.862609119451948\\
2.40185241278055e-06	-0.862614501487115\\
3.8162596218126e-06	-0.862619883465203\\
5.23066683084465e-06	-0.862625265385296\\
6.6450740398767e-06	-0.862630647247071\\
1.25028029768298e-05	-0.862652935459392\\
1.83605319137828e-05	-0.862675222670452\\
2.42182608507359e-05	-0.862697508880223\\
3.0075989787689e-05	-0.862719794088726\\
6.00482120075413e-05	-0.862833805073528\\
9.00204342273936e-05	-0.86294778984807\\
0.000119992656447246	-0.86306174841517\\
0.000149964878667098	-0.863175680777605\\
0.000179937100886951	-0.863289586938144\\
0.000332323598807076	-0.863868310035677\\
0.000484710096727201	-0.86444635625969\\
0.000637096594647326	-0.865023725975205\\
0.000789483092567451	-0.865600419547084\\
0.000941869590487576	-0.866176437339954\\
0.00138404454605448	-0.867844027482884\\
0.00182621950162138	-0.869505939747539\\
0.00226839445718828	-0.8711621830138\\
0.00271056941275518	-0.872812766148005\\
0.00315274436832208	-0.874457698002814\\
0.00359491932388899	-0.876096987417114\\
0.00433362333684666	-0.878823045543851\\
0.00507232734980433	-0.881533421568512\\
0.005811031362762	-0.884228156448409\\
0.00654973537571967	-0.886907291035231\\
0.00728843938867734	-0.88957086607533\\
0.00802714340163501	-0.892218922210033\\
0.00893054121565263	-0.895436326217081\\
0.00983393902967025	-0.898630654610497\\
0.0107373368436879	-0.90180198108938\\
0.0116407346577055	-0.904950379121316\\
0.0125441324717231	-0.908075921943273\\
0.0134475302857407	-0.911178682562489\\
0.0146224277507307	-0.91517995955562\\
0.0157973252157207	-0.919142985714951\\
0.0169722226807107	-0.923067919905833\\
0.0181471201457007	-0.926954920348179\\
0.0193220176106907	-0.930804144619712\\
0.0204969150756807	-0.934615749659244\\
0.0216718125406707	-0.938389891769944\\
0.0231041812400222	-0.942940661672062\\
0.0245365499393737	-0.94743626216101\\
0.0259689186387252	-0.951876972670514\\
0.0274012873380767	-0.956263071261375\\
0.0288336560374282	-0.960594834629931\\
0.0302660247367797	-0.964872538116467\\
0.0316983934361312	-0.969096455713554\\
0.0336294714734546	-0.974706364691276\\
0.035560549510778	-0.980219673987197\\
0.0374916275481013	-0.98563704352041\\
0.0394227055854247	-0.99095912888461\\
0.0413537836227481	-0.996186581383735\\
0.0432848616600715	-1.00132004806729\\
0.0452159396973949	-1.00636017176531\\
0.0475942398069608	-1.01244021999216\\
0.0499725399165266	-1.01838083668574\\
0.0523508400260925	-1.02418319680419\\
0.0547291401356584	-1.02984846594771\\
0.0571074402452243	-1.0353778004526\\
0.0594857403547902	-1.04077234748412\\
0.0618640404643561	-1.04603324512839\\
0.0654611559218425	-1.05373878458834\\
0.0690582713793289	-1.06114503967313\\
0.0726553868368153	-1.06825583249711\\
0.0762525022943017	-1.07507493999313\\
0.0798496177517881	-1.08160609458714\\
0.0834467332092745	-1.08785298486174\\
0.0870438486667609	-1.09381925620802\\
0.0948395748196346	-1.10580348432597\\
0.102635300972508	-1.11652256085624\\
0.110431027125382	-1.12601159187459\\
0.118226753278256	-1.1343048165422\\
0.126022479431129	-1.14143563457748\\
0.133818205584003	-1.1474366314347\\
0.145679573186037	-1.15447135448284\\
0.157540940788072	-1.15907255375943\\
0.169402308390106	-1.16134492693044\\
0.181263675992141	-1.16138943988316\\
0.193125043594175	-1.15930349157627\\
0.20498641119621	-1.15518106513813\\
0.222976361272891	-1.14524412147447\\
0.240966311349572	-1.13113343669228\\
0.258956261426252	-1.11313633773137\\
0.276946211502933	-1.09152563423473\\
0.294936161579614	-1.06656047982518\\
0.312926111656295	-1.03848713749778\\
0.338332798699639	-0.99399700221273\\
0.363739485742983	-0.944393495783175\\
0.389146172786327	-0.890254993155507\\
0.41455285982967	-0.832120818605615\\
0.439959546873014	-0.77049389459469\\
0.465366233916358	-0.705843049863967\\
0.501191457304087	-0.610381161263576\\
0.537016680691815	-0.5108958457531\\
0.572841904079544	-0.408402907093034\\
0.608667127467273	-0.303818555507347\\
0.644492350855001	-0.197966745543078\\
0.68031757424273	-0.0915855903790637\\
0.736084579616188	0.0735329537026557\\
0.774701678957587	0.186561366159034\\
0.774701678957594	0.186561366159054\\
0.77508750920008	0.187682657544552\\
0.775473339442565	0.188803775124403\\
0.775859169685051	0.189924714710067\\
0.776959412267718	0.193120322370385\\
0.778059654850384	0.19631451056767\\
0.779159897433051	0.199507233256413\\
0.780260140015717	0.202698465818333\\
0.782245372849324	0.208452784013998\\
0.784230605682932	0.214202126181967\\
0.786215838516539	0.219946420389221\\
0.788201071350146	0.225685595900222\\
0.790186304183754	0.231419581240841\\
0.792778950437179	0.238900036532578\\
0.795371596690605	0.246371360289473\\
0.797964242944031	0.253833396811699\\
0.800556889197457	0.261285992192265\\
0.803149535450883	0.26872899413959\\
0.807069734006657	0.279964596615613\\
0.810989932562431	0.291177407382615\\
0.814910131118206	0.302366920022159\\
0.81883032967398	0.313532636508187\\
0.822750528229754	0.324674066939882\\
0.829904494132436	0.344941880398899\\
0.837058460035119	0.365124355930555\\
0.844212425937801	0.385218714791592\\
0.851366391840484	0.405222266912726\\
0.853573891945413	0.411376027406271\\
0.85357389194542	0.411376027406291\\
0.854501628534333	0.413959000081816\\
0.855429365123245	0.416540373334697\\
0.856357101712158	0.419120081544966\\
0.858281664089495	0.424466797647642\\
0.860206226466832	0.429806252594866\\
0.86213078884417	0.435137892411667\\
0.864055351221507	0.440461274349809\\
0.867292922924893	0.449396496660097\\
0.870530494628279	0.45830519725414\\
0.873768066331665	0.46718533399682\\
0.877005638035052	0.476034764973001\\
0.881670228187441	0.488726638565864\\
0.88633481833983	0.501344048696189\\
0.890999408492219	0.513880618314854\\
0.895663998644608	0.526330400676205\\
0.898715552345879	0.53442544104586\\
0.898715552345886	0.534425441045879\\
0.899350330113535	0.536103637898688\\
0.899985107881185	0.537780089837392\\
0.900619885648835	0.539454748899905\\
0.901737821448125	0.542400262366724\\
0.902855757247416	0.545340415609723\\
0.903973693046706	0.548274940411166\\
0.905091628845997	0.551203702947203\\
0.90713097398682	0.556531326695621\\
0.909170319127642	0.561839176429831\\
0.911209664268465	0.567126937163958\\
0.913249009409288	0.57239431446055\\
0.915288354550111	0.577641020693438\\
0.91858842760321	0.586086769608428\\
0.92188850065631	0.594476582556187\\
0.925188573709409	0.602809461565474\\
0.928488646762508	0.611084498837232\\
0.931788719815608	0.619300865250596\\
0.936125538778493	0.630008003511458\\
0.940462357741379	0.640610950369627\\
0.944799176704265	0.651108306733815\\
0.949135995667151	0.661498824674274\\
0.953472814630037	0.671781383436075\\
0.961899024816125	0.691447583011013\\
0.970325235002213	0.71069600432491\\
0.978751445188301	0.729521144607988\\
0.987177655374389	0.747918407744995\\
0.995603865560477	0.765883912161704\\
1.01047569450743	0.796528789008917\\
1.02534752345438	0.825804916121434\\
1.04021935240133	0.85370323934457\\
1.05509118134828	0.880218368198667\\
1.06996301029524	0.905347819081779\\
1.08483483924219	0.929092240428591\\
1.10600607332612	0.960515753736334\\
1.12717730741006	0.989172468979071\\
1.14834854149399	1.0151060904243\\
1.16951977557792	1.03836824076228\\
1.19069100966186	1.05901634046221\\
1.21186224374579	1.07711801602167\\
1.22767455421788	1.08902408816973\\
1.2276745542179	1.08902408816974\\
1.23427198020141	1.09351702253728\\
1.23957340651417	1.09698107513876\\
1.24487483282693	1.10030247356005\\
1.25017625913969	1.10348404236652\\
1.2637203604155	1.11106314750482\\
1.27726446169132	1.11781529071177\\
1.29080856296713	1.12375620038405\\
1.30435266424295	1.12891740739164\\
1.32714784544659	1.1359420903858\\
1.34994302665023	1.14103270909615\\
1.37273820785388	1.14435892422687\\
1.39553338905752	1.14607762789322\\
1.41832857026116	1.1463331899899\\
1.45292233798962	1.14422109219684\\
1.48751610571808	1.13948050923328\\
1.52210987344654	1.13250577020598\\
1.556703641175	1.12365137490196\\
1.59129740890347	1.11323626480831\\
1.64184733780232	1.09578883556331\\
1.69239726670117	1.07639933819328\\
1.74294719560002	1.05573608019754\\
1.79349712449887	1.03435598034846\\
1.84404705339773	1.01272111750207\\
1.91678327088859	0.98187014774544\\
1.98951948837946	0.952145880248395\\
2.06225570587033	0.92417719054524\\
2.1349919233612	0.898368781433706\\
2.20772814085207	0.874953387914191\\
2.29665508334235	0.849725364080195\\
2.38558202583264	0.828205092958938\\
2.47450896832292	0.810205748457283\\
2.56343591081321	0.79544109567071\\
2.6523628533035	0.783568062753558\\
2.74128979579378	0.774215876177025\\
2.86573345209315	0.764656198913446\\
2.99017710839252	0.758331305914104\\
3.11462076469189	0.754384153182483\\
3.23906442099125	0.752110437309087\\
3.36350807729062	0.750956270918799\\
3.48795173358999	0.750501658031991\\
3.61843705168917	0.750442861734651\\
3.61843705168919	0.750442861734651\\
3.63204764124076	0.750452373749595\\
3.64565823079233	0.750463591212621\\
3.6592688203439	0.750476373992189\\
3.67906822172875	0.750497071184416\\
3.68820949052243	0.750507443270632\\
3.68820949052247	0.750507443270632\\
3.70331819837544	0.750535227535241\\
3.71842690622841	0.750587092057881\\
3.73353561408137	0.750677181704254\\
3.75574411716739	0.750886774608568\\
3.77795262025341	0.751214376182034\\
3.80016112333942	0.751679413898707\\
3.82236962642544	0.752290757067085\\
3.85357645144836	0.753401167277091\\
3.88478327647129	0.754804747484847\\
3.91599010149421	0.756498764345257\\
3.94719692651714	0.758475691882054\\
4.00309693317951	0.762679808567377\\
4.05899693984188	0.767646503999126\\
4.11489694650425	0.773250187891547\\
4.17079695316662	0.779356927950692\\
4.226696959829	0.785832473607522\\
4.29976365035829	0.794636429833171\\
4.37283034088758	0.803573887731033\\
4.44589703141687	0.812402913776401\\
4.51896372194617	0.820921641586984\\
4.59203041247546	0.828970645984112\\
4.69650948133496	0.839451152631597\\
4.80098855019445	0.848531733097112\\
4.90546761905394	0.856141108021159\\
5.00994668791344	0.862326161682583\\
5.11442575677293	0.867216056633307\\
5.24376976254039	0.871748245715032\\
5.37311376830786	0.874949182774378\\
5.50245777407532	0.877201775352582\\
5.63180177984278	0.87885176279258\\
5.76114578561024	0.880179664173889\\
5.92583541398891	0.881709557963682\\
6.09052504236758	0.883305115805604\\
6.25521467074625	0.885079411395373\\
6.41990429912492	0.887030559944732\\
6.58459392750359	0.889091013065246\\
6.77528722102658	0.891492807041795\\
6.96598051454957	0.893786774057209\\
7.15667380807257	0.895872895488806\\
7.34736710159556	0.897696190752611\\
7.53806039511855	0.899244458609431\\
7.82653856220132	0.901116641726537\\
8.11501672928408	0.902521940599444\\
8.40349489636684	0.903597258166583\\
8.6919730634496	0.904456399496104\\
8.98045123053236	0.905168555278292\\
9.34045123053236	0.905894729373804\\
9.70045123053236	0.906469795532473\\
10.0604512305324	0.906912703938787\\
10.4204512305324	0.90724642343565\\
10.7804512305324	0.907497823957638\\
11.1404512305324	0.907690468615278\\
11.5004512305324	0.907840203152885\\
11.8604512305324	0.907956218751896\\
12.2204512305324	0.90804462470991\\
12.5804512305324	0.908110989425415\\
12.9404512305324	0.908160744937271\\
13.3004512305324	0.90819846610087\\
13.6604512305324	0.908227359894072\\
14.0204512305324	0.908249446111723\\
14.3804512305324	0.908266102439713\\
14.7404512305324	0.90827848686355\\
15.1004512305324	0.908287651801254\\
15.4604512305323	0.908294472271564\\
15.8204512305323	0.908299581621324\\
16.1804512305323	0.908303399114923\\
16.5404512305323	0.908306213690176\\
16.9004512305323	0.908308254927058\\
17.2604512305323	0.908309719790215\\
17.6204512305323	0.908310768886012\\
17.9804512305323	0.908311519907472\\
18	0.90831155422311\\
};
\addlegendentry{$y_2$}
\end{axis}
\end{tikzpicture}%

%% file: sections/conclusion.tex
\section{Conclusion}\label{sec:conclusion}

In this paper, we developed a robust control framework for feedback optimization of LTI plants that places stability analysis, performance, dynamic controller synthesis, and constrained optimization within a single framework. By charaterizing first-order oracles through dynamic Zames--Falb multipliers, we analyze and systematically design controllers that stabilize a generalized plant comprising the plant dynamics and an optimality model. The framework establishes a formal link between IQC and timescale separation based analysis of gradient flows, and goes beyond these by enabling the synthesis of dynamic output-feedback controllers with certified performance levels.

Future directions include the development of internal model principle based robust regulation frameworks to address time-varying and composite optimization problems. Moreover, the robust control baseline lends itself to generalizations based on linear parameter-varying or sum-of-squares techniques, which could extend the framework beyond LTI plants.

%% file: sections/appendix.tex
\subsection{Generalization to non-invertible $A$}
\label{appendix:generalization}

When $A$ is not invertible, the steady-state equation \mbox{$0 = A\bar x + B \bar u$} induces the linear subspace
\begin{equation}
    \bmat{\bar x \\ \bar u} = \mathcal{N} \xi, \quad \mathrm{ran}\, \mathcal{N} = \mathrm{null} \bmat{A & B},
\end{equation}
where $\xi \in \mathbb{R}^{n_\xi}$ is a free variable with the dimension \mbox{$n_\xi = \mathrm{dim}(\mathrm{null}\bmat{A & B})$}. We can re-parametrize the OSS problem in $\xi$, which eliminates the steady-state constraint and gives
\begin{equation}\label{eq:gen-OSS}
    \min_{\xi\in \mathbb{R}^{n_\xi}} \;\; \Phi_1(\Pi_u \xi) + \Phi_2(\Pi_y \xi + b(w)),
\end{equation}
where \mbox{$\Pi_u \coloneq \bmat{0 & I_{n_u}} \mathcal{N}$}, \mbox{$\,\Pi_y \coloneq \bmat{C & D} \mathcal{N}$} and $b(w)$ is an affine function of $w$. The first-order optimality condition gives
\begin{equation}\label{eq:first-order-opt:generalized}
    0 = \Pi_u^\top \nabla \Phi_1(\Pi_u \xi) + \Pi_y^\top \nabla \Phi_2(\Pi_y \xi + b(w)).
\end{equation}
We obtain an optimality model by integrating~\eqref{eq:first-order-opt:generalized} and replacing $\Pi_u \xi$ and $\Pi_y \xi + b(w)$ with $u$ and $y$, respectively, giving
\begin{equation}\label{eq:optimality-model:generalized}
    \dot{\eta} = \Pi_u^\top \nabla \Phi_1(u) + \Pi_y^\top \nabla \Phi_2(y).
\end{equation}
If $A$ is invertible, we have \mbox{$\mathcal{N} = \bmat{\Pi_{xu} & I_{n_u}}$}, \mbox{$\Pi_u = I_{n_u}$}, \mbox{$\Pi_y = \Pi_{yu}$} and \mbox{$\xi = u$}, recovering \eqref{eq:controller-eta}. All derivations that include $\eta$ may be generalized with \eqref{eq:optimality-model:generalized}.

\subsection{Proof of Proposition~\ref{prop:stability-general}}\label{appendix:a}

The proof follows standard arguments, e.g. \cite{schererLMIs}. Introducing a state $\zeta$, the dynamics of \eqref{eq:augmented-plant} can be written as
\begin{equation}\label{eq:zeta_dynamics}
    \dot{\zeta} = \mathcal{A} \zeta + \mathcal{B} \bmat{\tilde{p}_1 \\ \tilde{p}_2}, \quad \psi = \mathcal{C} \zeta + \mathcal{D} \bmat{\tilde{p}_1 \\ \tilde{p}_2},
\end{equation}
where $\psi = \mathrm{col}(\psi_1, \psi_2)$. 
Define the positive definite storage function $V(\zeta)\coloneqq \zeta^\top\mathcal X\zeta$. Then, by left and right multiplying~\eqref{eq:LMI-zf-general} with $\mathrm{col}(\zeta, \tilde{p}_1, \tilde{p}_2)$ we can conclude that \mbox{$\dot{V} + \lambda (\psi_1^\top J_{n_u} \psi_1 + \psi_2^\top J_{n_y} \psi_2) < 0$}. By the strict inequality, we can use a perturbation argument such that for some $\delta>0$ 
\[
    \dot V + \lambda (\psi_1^\top J_{n_u} \psi_1 + \psi_2^\top J_{n_y} \psi_2) \le -\delta \!
    \left( \|\zeta\|^2 + \|\tilde p_1\|^2 + \|\tilde p_2\|^2
    \right).
\]
Integrating and leveraging that $\int_{0}^{T} \psi_1^\top J_{n_u} \psi_1 \, \mathrm{d}t \geq 0$ and $\int_{0}^{T} \psi_2^\top J_{n_y} \psi_2\, \mathrm{d}t \geq 0$, gives
\begin{equation*}\label{eq:V_bounded}
    V(T)-V(0) \le - \delta \int_0^T
    \left(
        \|\zeta(t)\|^2 + \|\tilde p_1(t)\|^2 + \|\tilde p_2(t)\|^2
    \right)\mathrm{d}t.
\end{equation*}
Since $V(T) \geq 0$, it holds $\int_{0}^{T} \|\zeta\|^2 + \|\tilde{p}_1\|^2 + \|\tilde{p}_2\|^2\,\mathrm{d}t \leq \tfrac{1}{\delta}V(0)$ for all $T\geq 0$, and letting $T\to\infty$ implies that the signals $\zeta, \tilde{p}_1, \tilde{p}_2$ must be square integrable. Moreover, \mbox{$V(T) \leq V(0)$}, which implies that $\zeta$ is bounded. Since $\nabla \Phi_1$ and $\nabla \Phi_2$ are both Lipschitz continuous functions in their arguments, and $\tilde{q}_1$ and $\tilde{q}_2$ can both be retrieved as linear functions of~$\zeta$, also $\tilde{p}_1, \tilde{p}_2$ are bounded. By~\eqref{eq:zeta_dynamics}, $\dot{\zeta}$ is hence bounded as well, implying uniform continuity of $\zeta$. Barbalat's Lemma then implies $\zeta(t) \to 0$, and by continuity also $\tilde{p}_1(t) \to 0$ and $\tilde{p}_2(t) \to 0$, as $t \to \infty$.

\subsection{Proof of Proposition~\ref{prop:zf-LMI}}
\label{appendix:b}

We first derive some structural implications of the integrator dynamics of $\tilde{u}$. Substituting the realization \eqref{eq:augmented-plant-realization} into~\eqref{eq:LMI-zf}, it can be observed that $\Lambda(\mathcal{X},\lambda, \varepsilon)$ has the structure
\begin{equation}\label{eq:LMI-structure}
    \Lambda(\mathcal X,\lambda,\varepsilon)
    =
    \begin{bmatrix}
        0
        &
        \star
        \\
        \star
        &
        \Lambda_0(\mathcal{X}_0,\lambda,\varepsilon)
    \end{bmatrix}.
\end{equation}
The condition $\Lambda(\mathcal X,\lambda,\varepsilon)\preceq 0$ forces the subdiagonal elements in~\eqref{eq:LMI-structure} to zero and we have $\Lambda_0(\mathcal{X}_0,\lambda,\varepsilon) \preceq 0$. Define accordingly the reduced realization by stripping the $\tilde{u}$-dynamics
\begin{align}
    \bar{\mathcal A}
    &\!=\!
    \begin{bmatrix}
        A & 0 \\
        B_H \!\otimes\! (L_yC) & A_H \!\otimes\! I_{n_y}
    \end{bmatrix},
    &
    \bar{\mathcal B}(\varepsilon)
    &\!=\!
    \begin{bmatrix}
        \varepsilon\Pi_{xu}\Pi_{yu}^\top \\
        -B_H \!\otimes\! I_{n_y}
    \end{bmatrix},
    \notag
    \\
    \bar{\mathcal C}
    &\!=\!
    \begin{bmatrix}
        \tilde{d} L_yC & -C_H \!\otimes\! I_{n_y} \\
        0 & 0
    \end{bmatrix},
    &
    \bar{\mathcal D}
    &\!=\!
    \begin{bmatrix}
        -\tilde{d} I_{n_y} \\
        I_{n_y}
    \end{bmatrix}.
    \label{eq:augmentedplant:reduced}
\end{align}
Recall that $\tilde{d} = 1 - D_H$ with $D_H = \lim_{s \to \infty} H_y(s)$, and note that $\tilde{d} \geq 0$ since $\| H_y \|_1 \leq 1$. The case $\tilde{d}=0$ is degenerate since it decouples the plant dynamics from \eqref{eq:augmentedplant:reduced} and can therefore be excluded.
Define $\bar{\mathcal B}_1 \!=\! \left[\! \begin{smallmatrix}\Pi_{xu} \Pi_{yu}^\top \\ 0\end{smallmatrix} \! \right]$ such that \mbox{$\bar{\mathcal B}(\varepsilon) = \bar{\mathcal B}(0) \!+\! \varepsilon \bar{\mathcal B}_1$}, and $\bar{\mathcal{C}}_1 = \bmat{\tilde{d} L_yC & \,\, -C_H \otimes I_{n_y}}$. 
Using Schur complement arguments, we have $\Lambda_0(\mathcal{X}_0,\lambda,\varepsilon) \preceq 0$ if and only if $F(\varepsilon)\preceq 0$, with 
\begin{equation*}
    F(\varepsilon) \coloneq R_0+\varepsilon R_1+\varepsilon^2R_2,
\end{equation*}
where
\begin{align*}
    R_0\!
    \coloneqq{}&
    \bar{\mathcal A}^\top\bar{\mathcal X}
    +
    \bar{\mathcal X}\bar{\mathcal A}
    \\[-1ex]
    & \qquad +
    \frac{1}{2\tilde{d}\lambda}
    \bigl(
        \bar{\mathcal X}\bar{\mathcal B}(0)
        +
        \lambda\bar{\mathcal C}_1^\top
    \bigr)
    \bigl(
        \bar{\mathcal X}\bar{\mathcal B}(0)
        +
        \lambda\bar{\mathcal C}_1^\top
    \bigr)^\top,
    \\
    R_1\!
    \coloneqq{}&
    \frac{1}{2\tilde{d}\lambda}
    \Big(
        \bar{\mathcal X}\bar{\mathcal B}_1
        \bigl(
            \bar{\mathcal X}\bar{\mathcal B}(0)
            \!+\!
            \lambda\bar{\mathcal C}_1^\top
        \bigr)^\top \!
        \!+\!
        \bigl(
            \bar{\mathcal X}\bar{\mathcal B}(0)
            \!+\!
            \lambda\bar{\mathcal C}_1^\top
        \bigr)
        \bar{\mathcal B}_1^\top\bar{\mathcal X}
    \Big),
    \\
    R_2\!
    \coloneqq{}&
    \frac{1}{2\tilde{d}\lambda}
    \bar{\mathcal X}\bar{\mathcal B}_1
    \bar{\mathcal B}_1^\top
    \bar{\mathcal X}.
\end{align*}
Note that $R_2 \succeq 0$. Moreover, the condition $\Lambda_0(\mathcal{X}_0,\lambda,0) \prec 0$ implies \mbox{$R_0 \prec 0$}. Hence, $F(\varepsilon)$ is a matrix-convex quadratic function with $F(0) \prec 0$ and $F(\varepsilon^\star_{\mathrm{ZF}})\preceq 0$. It follows that $F(\varepsilon)\prec 0$ for all  $\varepsilon\in(0,\varepsilon^\star_{\mathrm{ZF}})$, which is equivalent to $\Lambda_0(\mathcal{X}_0, \lambda, \varepsilon) \prec 0 \; \forall \varepsilon\in(0,\varepsilon^\star_{\mathrm{ZF}})$.

Now introduce $\zeta = \mathrm{col}(\tilde{u}, \tilde{e}, x_H)$ and the positive definite storage function $V(\zeta) =\zeta^\top \mathcal{X} \zeta$.
By left- and right-multiplying~\eqref{eq:LMI-zf} with $\mathrm{col}(\zeta, \tilde p_2)$ we infer $0 \leq \dot V + \lambda \psi^\top J_{n_y} \psi$, where $\psi = \mathcal{G}[\, \mathrm{col}(\tilde{p}_1, \tilde{p}_2) \,]$ satisfies $\int_{0}^{T} \psi^\top J_{n_y} \psi\, \mathrm{d}t \geq 0$ as before. Integrating over $[0,T]$ gives $V(T)-V(0) \leq 0$, which, following the same reasoning as in the proof of Proposition~\ref{prop:stability-general},
implies that all trajectories are bounded.

Now note that by the zero subdiagonals that arise in~\eqref{eq:LMI-structure}, we have
\begin{equation*}
    \dot V+\lambda \psi^\top J_{n_y} \psi = \smat{e \\ x_H \\ \tilde{p}_2}^\top \Lambda_0(\mathcal{X}_0,\lambda,\varepsilon) \smat{e \\ x_H \\ \tilde{p}_2}.
\end{equation*}
Now let $\varepsilon\in(0,\varepsilon^\star_{\mathrm{ZF}})$ such that $\Lambda_0(\mathcal{X}_0,\lambda,\varepsilon) \prec 0$. Using a perturbation argument, there exists $\delta > 0$ such that after integration we have
\[
    V(T)-V(0)
    \le
    -
    \delta
    \int_0^T
    \bigl(
        \|e(t)\|^2 +
        \|x_H(t)\|^2
        +
        \|\tilde p(t)\|^2
    \bigr)\,\mathrm{d}t.
\]
The same reasoning as in the proof of Proposition~\ref{prop:stability-general} implies $x_H(t)\to 0, e(t)\to 0, \tilde p_2(t)\to 0$ as $t \to \infty$.
Thus, $y(t) \to y^\star$, giving $\tilde q_2(t) \to 0$. From~\eqref{eq:lure}, $\Piyu \tilde u = \tilde{q}_2 - C e$, and full column rank of $\Piyu$ implies $\tilde u \to 0$, i.e., $u(t) \to u^\star$. Since the choice of $u^\star, y^\star$ is so that first-order optimality holds, the trajectories converge to a solution of~\eqref{eq:oss}.

\subsection{Proof of Corollary~\ref{corr:static-LMI}}
\label{appendix:c}

By construction, $\eps^\star_{\mathrm{static}} \leq \eps^\star_{\mathrm{ZF}}$, since $H=0$ is just one special case of Proposition~\ref{prop:zf-LMI}.

We now show that any $\varepsilon \leq \tfrac{\lmin(Q)}{2 \| X \Pixu \| \, \ell}$ renders \eqref{eq:LMI-zf} feasible for $H=0$. We choose the Ansatz
\begin{equation*}
    \mathcal{X} = \bmat{ X & X_{12} \\ X_{12}^\top & X_{22} }.
\end{equation*}
With $H=0$, the filter state $x_H$ is eliminated from \eqref{eq:augmented-plant-realization}, and inserting $\mathcal{A}, \mathcal{B}, \mathcal{C}, \mathcal{D}$ into \eqref{eq:LMI-zf} therefore leads to
\begin{equation}\label{eq:LMI-expanded}
    \bmat{
        A^\top X + X A & A^\top X_{12} & M_{13} \\
        X_{12}^\top A & 0 & M_{23} \\
        M_{13}^\top & M_{23}^\top & -2\lambda I_{n_y}
    } \preceq 0,
\end{equation}
where
\begin{align*}
    M_{13} &\coloneqq \varepsilon X \Pixu \Piyu^\top + \lambda L_y C^\top, \\
    M_{23} &\coloneqq \big( \varepsilon X_{12}^\top \Pixu - \varepsilon X_{22} + \lambda L_y I_{n_u} \big) \Piyu^\top.
\end{align*}
Negative semi-definiteness forces $A^\top X_{12} = 0$ and \mbox{$M_{23} = 0$}. Since $A$ is invertible, the former implies $X_{12}=0$. Since $\Piyu$ has full column rank, the latter implies $X_{22} = \tfrac{\lambda L_y}{\varepsilon} I_{n_u}$.
Thus, \eqref{eq:LMI-expanded} reduces to
\begin{equation}\label{eq:LMI-reduced}
    \bmat{ A^\top X + X A & \varepsilon X \Pixu \Piyu^\top + \lambda L_y C^\top \\ \star & -2\lambda I_{n_y} } \preceq 0.
\end{equation}
Now substitute $A^\top X + X A = -Q$ and perform a Schur complement around $-2\lambda I_{n_y}$, which gives
\begin{multline*}
    -Q + \tfrac{\lambda L_y^2}{2} C^\top C + \tfrac{\varepsilon^2}{2\lambda} X \Pixu \Piyu^\top \Piyu \Pixu^\top X 
    \\+ \tfrac{\varepsilon L_y}{2} \big( X \Pixu \Piyu^\top C + C^\top \Piyu \Pixu^\top X \big)  \preceq 0.
\end{multline*}
A sufficient scalar condition follows from Cauchy--Schwarz arguments, namely
\begin{multline}\label{eq:scalar-ineq}
    \tfrac{\varepsilon^2}{2\lambda} \| X \Pixu \|^2 \| \Piyu \|^2 + \varepsilon L_y \| X \Pixu \| \| \Piyu \| \| C \| 
    \\+ \tfrac{\lambda L_y^2}{2} \| C \|^2 \le \lmin(Q).
\end{multline}
Minimizing the left-hand side over $\lambda > 0$ yields the optimal multiplier
\begin{equation*}
    \lambda^\star = \tfrac{\varepsilon \| X \Pixu \| \| \Piyu \|}{L_y \| C \|},
\end{equation*}
giving $2 \varepsilon L_y \| X \Pixu \| \| \Piyu \| \| C \| \le \lmin(Q)$, which rearranges to $\varepsilon \leq \tfrac{\lmin(Q)}{2 \| X \Pixu \| \, \ell}$. Thus, any such $\varepsilon$ gives feasibility with the certificate
\begin{equation*}
    \setlength{\arraycolsep}{2pt}
    \mathcal{X} = \bmat{X & 0 \\ 0 & \tfrac{\| X \Pixu \| \| \Piyu \|}{\| C \|} I_{n_u}} \succ 0, 
    \quad 
    \lambda = \tfrac{\varepsilon \| X \Pixu \| \| \Piyu \|}{L_y \| C \|},
\end{equation*}
and conclusively, $\varepsilon^\star_{\mathrm{static}}$ must be greater or equal.

The last claim is equivalent to \eqref{eq:LMI-reduced} with $\varepsilon=0$ and strict feasibility, which by Schur complement is equivalent to \mbox{$A^\top X + X A + \frac{\lambda L_y^2}{2} C^\top C \prec 0$}.